\documentclass[11pt,english]{article}
\usepackage[T1]{fontenc}
\usepackage[utf8]{luainputenc}
\usepackage[letterpaper,top=1.in,bottom=1.in,left=1.2in,right=1.2in]{geometry}
\usepackage{amsmath}
\usepackage{amssymb}
\usepackage{amsthm}
\usepackage{bbm}
\usepackage{bm}
\usepackage{babel}
\usepackage[titletoc,title]{appendix}
\usepackage{titlesec}
\usepackage{titling}
\usepackage{graphicx}
\usepackage[space]{grffile}
\usepackage{tikz}
\usepackage{marginnote}
\usepackage{float}
\usepackage{centernot}
\usepackage{mdframed}
\usepackage[sc]{mathpazo}
\usepackage{hyperref}
\usepackage{enumitem}
\hypersetup{colorlinks=true,linkcolor=blue,citecolor=blue}
\newif\iftitlepage
\titlepagetrue
\newif\ifabstract
\abstracttrue
\newif\iftoc
\toctrue

\newcommand{\mm}[1]{\begin{align*}#1\end{align*}}
\newcommand{\etal}{\emph{et al.}}
\newcommand{\toc}[1]{\pagenumbering{roman}\setcounter{tocdepth}{#1}\tableofcontents\newpage\pagenumbering{arabic}}

\renewcommand{\qed}{\quad\hbox{\vrule width 8pt height 8pt depth 1.5pt}\lower 8.5pt\hbox{}}
\newcommand{\innerqed}{\renewcommand{\qed}{\hfill\ensuremath{\Box}}}
\newcounter{tag}[section]

\newcommand{\eqtag}[1]{\stepcounter{tag}\tag{\thesection.\thetag} \label{#1}}
\newcommand{\eqdef}{\ensuremath{\overset{\mathrm{def}}{=\joinrel=}}}
\newcommand{\mathdot}{\ensuremath{\;\text{.}}}
\newcommand{\mathcomma}{\ensuremath{\;\text{,}}}

\newcommand{\eps}{\ensuremath{\epsilon}}
\newcommand{\e}{\ensuremath{\eps}}
\newcommand{\al}{\ensuremath{\alpha}}
\newcommand{\de}{\ensuremath{\delta}}
\newcommand{\ra}{\ensuremath{\rightarrow}}

\newcommand{\norm}[1]{\ensuremath{\left\lVert#1\right\rVert}}
\newcommand{\abs}[1]{\ensuremath{\Big\vert#1\Big\vert}}
\newcommand{\ip}[1]{\ensuremath{\left\langle #1 \right\rangle}}

\newcommand{\ceil}[1]{\ensuremath{\left\lceil#1\right\rceil}}

\newcommand{\xor}{\ensuremath{\oplus}}
\newcommand{\rest}{\mathord{\upharpoonright}}

\newcommand{\F}{\ensuremath{\mathbb{F}}}
\newcommand{\N}{\ensuremath{\mathbb{N}}}

\newcommand{\R}{\ensuremath{\mathbb{R}}}

\newcommand{\E}{\ensuremath{\mathbb{E}}}
\newcommand{\bitset}{\ensuremath{\{0,1\}}}
\newcommand{\pmset}{\ensuremath{\{-1,1\}}}
\newcommand{\poly}{\ensuremath{\mathrm{poly}}}

\newcommand{\p}{\ensuremath{\mathcal{P}}}
\newcommand{\np}{\ensuremath{\mathcal{NP}}}
\newcommand{\nexp}{\ensuremath{\mathcal{NEXP}}}
\newcommand{\bpp}{\ensuremath{\mathcal{BPP}}}
\newcommand{\ac}{\ensuremath{\mathcal{AC}}}

\newcommand{\tc}{\ensuremath{\mathcal{TC}}}

\newcommand{\fanin}{\ensuremath{\mathrm{fan\text{-}in}}}
\newcommand{\acc}{\ensuremath{\mathtt{acc}}}
\newcommand{\sign}{\ensuremath{\mathrm{sgn}}}
\newcommand{\pmstar}{\ensuremath{\{-1,1,\star\}}}
\newcommand{\code}{\ensuremath{\mathtt{ECC}}}

\newtheoremstyle{dot}{\topsep}{\topsep}{\itshape}{0pt}{\bfseries}{.}{.5em}{}
\newtheoremstyle{nodot}{\topsep}{\topsep}{\itshape}{0pt}{\bfseries}{}{.5em}{}
\newcommand{\thmdot}[3]{\newtheorem{#1}[#2]{#3}\theoremstyle{dot}\newtheorem{#1dot}[#2]{#3}\theoremstyle{nodot}}

\theoremstyle{nodot}
\newtheorem{theorem}{Theorem}[section]
\thmdot{proposition}{theorem}{Proposition}
\thmdot{corollary}{theorem}{Corollary}
\thmdot{definition}{theorem}{Definition}
\thmdot{example}{theorem}{Example}
\thmdot{observation}{theorem}{Observation}
\thmdot{technique}{theorem}{Technique}
\thmdot{lemma}{theorem}{Lemma}
\thmdot{claim}{theorem}{Claim}
\thmdot{fact}{theorem}{Fact}

\theoremstyle{dot}
\newtheorem{subfact}{Fact}[theorem]
\newtheorem{subclaim}[subfact]{Claim}

\theoremstyle{nodot}
\newtheorem{problem}{Open Problem}

\begin{document}

% title
\iftitlepage\begin{titlepage}\fi
\title{Quantified Derandomization of Linear Threshold Circuits}
\author{
        Roei Tell \thanks{Department of Computer Science and Applied Mathematics, Weizmann Institute of Science, Rehovot, Israel. Email: {\tt roei.tell@weizmann.ac.il}}
}
\maketitle%\thispagestyle{fancy}
\iftitlepage\thispagestyle{empty}\fi

\ifabstract
\begin{abstract}
One of the prominent current challenges in complexity theory is the attempt to prove lower bounds for $\mathcal{TC}^0$, the class of constant-depth, polynomial-size circuits with majority gates. Relying on the results of Williams (2013), an appealing approach to prove such lower bounds is to construct a non-trivial derandomization algorithm for $\mathcal{TC}^0$. In this work we take a first step towards the latter goal, by proving the first positive results regarding the derandomization of $\mathcal{TC}^0$ circuits of depth $d>2$.

Our first main result is a \emph{quantified derandomization} algorithm for $\mathcal{TC}^0$ circuits with a super-linear number of wires. Specifically, we construct an algorithm that gets as input a $\mathcal{TC}^0$ circuit $C$ over $n$ input bits with depth $d$ and $n^{1+\exp(-d)}$ wires, runs in almost-polynomial-time, and distinguishes between the case that $C$ rejects at most $2^{n^{1-1/5d}}$ inputs and the case that $C$ accepts at most $2^{n^{1-1/5d}}$ inputs. In fact, our algorithm works even when the circuit $C$ is a linear threshold circuit, rather than just a $\mathcal{TC}^0$ circuit (i.e., $C$ is a circuit with linear threshold gates, which are stronger than majority gates).

Our second main result is that even a \emph{modest improvement} of our quantified derandomization algorithm would yield a non-trivial algorithm for \emph{standard derandomization} of all of $\mathcal{TC}^0$, and would consequently imply that $\mathcal{NEXP}\not\subseteq\mathcal{TC}^0$. Specifically, if there exists a quantified derandomization algorithm that gets as input a $\mathcal{TC}^0$ circuit with depth $d$ and $n^{1+O(1/d)}$ wires (rather than $n^{1+\exp(-d)}$ wires), runs in time at most $2^{n^{\exp(-d)}}$, and distinguishes between the case that $C$ rejects at most $2^{n^{1-1/5d}}$ inputs and the case that $C$ accepts at most $2^{n^{1-1/5d}}$ inputs, then there exists an algorithm with running time $2^{n^{1-\Omega(1)}}$ for \emph{standard derandomization} of $\mathcal{TC}^0$.
\end{abstract}
\fi
\iftitlepage\end{titlepage}\fi

\iftoc\toc{2}\fi

\section{Introduction} \label{sec:int}

The classical problem of \emph{derandomization of a circuit class $\mathcal{C}$} is the following: Given a circuit $C\in\mathcal{C}$, deterministically distinguish between the case that the acceptance probability of $C$ is at least $2/3$ and the case that the acceptance probability of $C$ is at most $1/3$. When $\mathcal{C}=\p/\poly$, this problem can be solved in polynomial time if and only if $promise\text{-}\bpp=promise\text{-}\p$. However, at the moment we do not know how to solve the problem in polynomial time even if $\mathcal{C}$ is the class of polynomial-sized CNFs.

The derandomization problem for a circuit class $\mathcal{C}$ is tightly related to lower bounds for $\mathcal{C}$. Relying on the classic hardness-randomness paradigm~\cite{yao82,bm84,nw94}, sufficiently strong lower bounds for a class $\mathcal{C}$ imply the existence of pseudorandom generators with short seed for $\mathcal{C}$, which allow to derandomize $\mathcal{C}$ (see, e.g.,~\cite[Chp. 20]{ab09},~\cite[Chp. 8.3]{gol08}). On the other hand, the existence of a non-trivial derandomization algorithm for a circuit class $\mathcal{C}$ typically implies (weak) lower bounds for $\mathcal{C}$. Specifically, for  many specific classes $\mathcal{C}$ (e.g., $\mathcal{C}=\p/\poly$), the existence of a derandomization algorithm for $\mathcal{C}$ running in time $2^n/n^{\omega(1)}$ implies that $\mathcal{E}^{\np}\not\subseteq\mathcal{C}$, and in some cases also that $\nexp\not\subseteq\mathcal{C}$ (see~\cite{wil13,sw13,bv14}, which build on~\cite{iw98,ikw02}).

Following Williams' proof that $\mathcal{ACC}$ does not contain $\nexp$~\cite{wil11}, one of the prominent current challenges in complexity theory is the attempt to prove similar lower bounds for the complexity class $\tc^0$ (i.e., the class of constant-depth, polynomial-sized circuits with majority gates, which extends $\mathcal{ACC}$). Even after extensive efforts during the last few decades (and with renewed vigor recently), the best-known lower bounds for $\tc^0$ assert the existence of functions in $\p$ that require $\tc^0$ circuits with a \emph{slightly super-linear} number of wires, or with a \emph{linear} number of gates (see Section~\ref{sec:bg} for further background).

Since derandomization algorithms imply lower bounds in general, an appealing approach to prove lower bounds for $\tc^0$ is to construct derandomization algorithms for this class. Moreover, a non-trivial derandomization of $\tc^0$ would separate $\tc^0$ from $\nexp$ (and not only from $\mathcal{E}^{\np}$; see~\cite{sw13,bv14}). Accordingly, the problem of either derandomizing $\tc^0$ or constructing a deterministic algorithm for \emph{satisfiability} of $\tc^0$ (which would be a stronger result) was recently suggested as a central open problem in complexity theory both by Williams~\cite{wil14} and by Aaronson~\cite{aar17}.~\footnote{See the first open problem in the Conclusions section in~\cite{aar17}, and Section 4.2 in~\cite{wil14}.}

An intensive recent effort has been devoted to constructing deterministic algorithms for satisfiability $\tc^0$. Such algorithms (with non-trivial running time) have been constructed for $\tc^0$ circuits of depth two, and for certain ``structured subclasses'' of $\tc^0$ (see~\cite{ips13,wil14b,as15,sstt16,tam16}). However, much less is known about \emph{derandomization} algorithms for $\tc^0$. Following an intensive effort to construct pseudorandom generators for a single \emph{linear threshold function}~\cite{dgjsv10,rs10,gowz10,krs12,mz13,kan11,kan14,km15,gkm15} (i.e., a single ``gate''; for background see Sections~\ref{sec:bg:dernd} and~\ref{sec:pre:ltf}), a first step towards derandomizing $\tc^0$ \emph{circuits} was very recently undertaken by Servedio and Tan~\cite{st17}, who considered the problem of derandomizing $\tc^0$ circuits of \emph{depth two}.~\footnote{Their manuscript is still unpublished, and so we describe their results in Section~\ref{sec:bg:dernd}.}

In this work we take a significant additional step towards the derandomization of $\tc^0$, by proving the \emph{first positive results} regarding the derandomization of $\tc^0$ circuits of \emph{any constant depth} $d\ge2$. Loosely speaking, we first construct an algorithm for a ``relaxed'' type of derandomization problem of sparse $\tc^0$ circuits of any constant depth $d\ge2$. As far as we are aware of, this is the first deterministic circuit-analysis algorithm for $\tc^0$ circuits of any constant depth that do not admit any special structure (other than being sparse). Then, we show that even a modest improvement in the parameters of the foregoing algorithm (for the ``relaxed'' problem) would yield a non-trivial algorithm for \emph{standard} derandomization of \emph{all of $\tc^0$}; indeed, as mentioned above, such a result would imply that $\nexp\not\subseteq\tc^0$. We thus suggest this approach (of the ``relaxed'' derandomization problem) as a potentially tractable line-of-attack towards proving $\nexp\not\subseteq\tc^0$ (see Section~\ref{sec:int:open}).

\subsection{Our results} \label{sec:int:res}

Our two main results lie within the framework of \emph{quantified derandomization}. Quantified derandomization, which was introduced by Goldreich and Wigderson~\cite{gw14}, is the relaxed derandomization problem of distinguishing between a circuit that accepts $1-o(1)$ of its inputs and a circuit that rejects $1-o(1)$ of its inputs (where the $1-o(1)$ term replaces the original $2/3$ term in standard derandomization). 

On the one hand, this relaxation potentially allows to construct more efficient derandomization algorithms. But on the other hand, the standard derandomization problem can be \emph{reduced to quantified derandomization}, by applying strong error-reduction within the relevant circuit class (such that a circuit with acceptance probability $2/3$ is transformed to a circuit with acceptance probability $1-o(1)$). Of course, a main goal underlying this approach is to reduce standard derandomization to a parameter setting for which we are able to construct a corresponding algorithm for quantified derandomization.

\subsubsection{A quantified derandomization algorithm} \label{sec:int:main}

Our first result is a \emph{quantified derandomization algorithm} for $\tc^0$ circuits with a slightly super-linear number of wires. %Specifically, we construct an algorithm that gets as input a $\tc^0$ circuit $C$ over $n$ input bits of depth $d$ with $n^{1+\exp(-d)}$ wires, runs in almost-polynomial-time, and distinguishes between the case that $C$ accepts all but $B(n)=2^{n^{1-1/O(d)}}$ of its inputs and the case that $C$ rejects all but $B(n)$ of its inputs.
In fact, our algorithm works not only for $\tc^0$, but also for the class of {\sf linear threshold circuits}: While in $\tc^0$ circuits each gate computes the majority function, in linear threshold circuits each gate computes a linear threshold function (i.e., a function of the form $g(x)=\sign\left(\sum_{i\in[n]}w_i\cdot x_i-\theta\right)$, for $w\in\R^n$ and $\theta\in\R$; see Section~\ref{sec:pre:ltf} for definitions). Towards stating this first result, denote by $\mathcal{C}_{n,d,w}$ the class of linear threshold circuits over $n$ input bits of depth $d$ and with at most $w$ wires. 

\begin{theorem} (quantified derandomization of linear threshold circuits). \label{thm:int:main}
There exists a deterministic algorithm that, when given as input a circuit $C\in\mathcal{C}_{n,d,n^{1+2^{-10d}}}$, runs in time $n^{O(\log\log(n))^2}$, and satisfies the following: 
\begin{enumerate}
	\item \label{it:int:main:acc} If $C$ accepts all but at most $B(n)=2^{n^{1-1/5d}}$ of its inputs, then the algorithm accepts $C$.
	\item \label{it:int:main:rej} If $C$ rejects all but at most $B(n)=2^{n^{1-1/5d}}$ of its inputs, then the algorithm rejects $C$.
\end{enumerate}
\end{theorem}

Observe that as $d$ grows larger, the algorithm in Theorem~\ref{thm:int:main} solves a more difficult derandomization task (since $B(n)$ is larger), but only has to handle circuits with fewer wires (i.e., $n^{1+\exp(-d)}$). Also note that the algorithm in Theorem~\ref{thm:int:main} is ``\emph{whitebox}'': That is, the algorithm gets as input an \emph{explicit description} of a \emph{specific} linear threshold circuit $C$, and uses this description when estimating the acceptance probability of $C$.~\footnote{The algorithm in Theorem~\ref{thm:int:main} works in any reasonable model of explicitly representing linear threshold circuits; see Section~\ref{sec:pre:ltf} for a brief discussion.} The actual algorithm that we construct works for a more general parameter regime, which exhibits a trade-off between the number $B(n)=2^{n^{1-\de}}$ of exceptional inputs for $C$ and the number $n^{1+\de\cdot\exp(-d)}$ of wires of $C$ (see Theorem~\ref{thm:main} for a precise statement).

The limitation on the number of wires of $C$ in Theorem~\ref{thm:int:main} (i.e., $n^{1+\exp(-d)}$) essentially \emph{matches the best-known lower bounds for linear threshold circuits}. This is no coincidence: Our algorithm construction follows a common theme in the design of circuit-analysis algorithms (e.g., derandomization algorithms or algorithms for satisfiability), which is the conversion of techniques that underlie lower bound proofs into algorithmic techniques. In this case, we observe that certain proof techniques for \emph{correlation bounds} for a circuit class $\mathcal{C}$ can be used to obtain algorithmic techniques for \emph{quantified derandomization} of $\mathcal{C}$. In particular, to construct the algorithm in Theorem~\ref{thm:int:main}, we leverage the techniques underlying the recent proof of Chen, Santhanam, and Srinivasan~\cite{css16} of correlation bounds for linear threshold circuits. A high-level description of our algorithm appears in Section~\ref{sec:tech:alg}.

\subsubsection{A reduction of standard derandomization to quantified derandomization}

Our second result reduces the \emph{standard} derandomization problem of $\tc^0$ to the \emph{quantified} derandomization problem of $\tc^0$ circuits with a \emph{super-linear number} of wires. In fact, we show that even a modest improvement of Theorem~\ref{thm:int:main} would yield a non-trivial algorithm for \emph{standard} derandomization of \emph{all of $\tc^0$}.

\begin{theorem} (a reduction of standard derandomization to quantified derandomization). \label{thm:int:threshold}
Assume that there exists a deterministic algorithm that, when given as input a circuit $C\in\mathcal{C}_{n,d,n^{1+O(1/d)}}$, runs in time at most $T(n)=2^{n^{1/4^{d}}}$, and for the parameter $B(n)=2^{n^{1-1/5d}}$ satisfies the following: If $C$ accepts all but at most $B(n)$ of its inputs then the algorithm accepts $C$, and if $C$ rejects all but at most $B(n)$ of its inputs then the algorithm rejects $C$.

Then, there exists an algorithm that for every $k\in\N$ and $d\in\N$, when given as input a circuit $C\in\mathcal{C}_{m,d,m^k}$, runs in time $2^{m^{1-\Omega(1)}}$, and satisfies the following: If $C$ accepts at least $2/3$ of its inputs then the algorithm accepts $C$, and if $C$ rejects at least $2/3$ of its inputs then the algorithm rejects $C$.
\end{theorem}

The gap between the algorithm constructed in Theorem~\ref{thm:int:main} and the algorithm assumed in the hypothesis of Theorem~\ref{thm:int:threshold} is quantitatively very small: Specifically, the algorithm in Theorem~\ref{thm:int:main} works when the number of wires in the input circuit $C$ is $n^{1+\exp(-d)}$, whereas the algorithm in the hypothesis of Theorem~\ref{thm:int:threshold} is required to work when the number of wires is $n^{1+O(1/d)}$. Moreover, Theorem~\ref{thm:int:threshold} holds even if this improvement (in the number of wires) comes at the expense of a longer running time; specifically, the conclusion of Theorem~\ref{thm:int:threshold} holds even if the algorithm runs in (sufficiently small) sub-exponential time.

As mentioned in the beginning of Section~\ref{sec:int}, a non-trivial derandomization of $\tc^0$ implies lower bounds for this class. Specifically, combining Theorem~\ref{thm:int:threshold} with~\cite[Thm 1.5]{sw13} (see also~\cite{bv14}), we obtain the following corollary:

\begin{corollary} (quantified derandomization implies lower bounds for $\tc^0$). \label{cor:int:threshold}
Assume that there exists a deterministic algorithm as in the hypothesis of Theorem~\ref{thm:int:threshold}. Then, $\nexp\not\subseteq\tc^0$.
\end{corollary}

The result that we actually prove is stronger and more general than the one stated in Theorem~\ref{thm:int:threshold} (see Theorem~\ref{thm:threshold}). First, the result holds even if we limit ourselves only to the class $\tc^0$, rather than to the class of linear threshold circuits (i.e., if we interpret the class $\mathcal{C}_{n,d,w}$ as the class of $\tc^0$ circuits over $n$ inputs of depth $d$ and with $w$ wires). And secondly, the hypothesis of the theorem can be modified via a trade-off between the number of exceptional inputs for the circuit $C$ and the number of wires in $C$.

The proof of Theorem~\ref{thm:int:threshold} is based on developing a very efficient method for error-reduction within sparse $\tc^0$. Specifically, we construct a seeded extractor such that there exists a $\tc^0$ circuit that gets input $x\in\bitset^n$ and computes the outputs of the extractor on $x$ and on all seeds using only a super-linear number of wires (i.e., a circuit of depth $d$ uses $n^{1+O(1/d)}$ wires);
%. Loosely speaking, the extractor gets $n$ input bits, outputs $n^{.01}$ bits, works for min-entropy $n^{.99}$, has seed length $t=1.01\cdot\log(n)$, and satisfies the following: There exists \emph{a $\tc^0$ circuit} that gets input $x\in\bitset^n$ and outputs the $2^t\approx n^{1.01}$ evaluations of the extractor on input $x$ and all seeds \emph{using only a super-linear number of wires}. Thus, intuitively, the output bits of the extractor are computed in a ``batch'', and computing each output bit only ``costs'' $n^{.01}$ wires. 
as far as we know, this is the first construction of a seeded extractor that is specific to $\tc^0$. This construction extends the study of randomness extraction in weak computational models, which has so far focused on $\ac^0$, on $\ac^0[\xor]$, and on streaming algorithms~\cite{brst02,vio05,hea08,gvw15,cl16}. The construction is described in high-level in Section~\ref{sec:tech:ext}, and a precise statement appears in Proposition~\ref{prop:threshold:sampler}.

\subsubsection{Restrictions for sparse $\tc^0$ circuits: A potential path towards $\nexp\not\subseteq\tc^0$} \label{sec:int:open}

Recall that the best-known lower bounds for $\tc^0$ circuits of arbitrary constant depth $d$ are for circuits with $n^{1+\exp(-d)}$ wires. Our results imply that a certain type of analysis of $\tc^0$ circuits \emph{with only $n^{1+O(1/d)}$ wires}, which is common when proving correlation bounds (i.e., average-case lower bounds), might suffice to deduce a lower bound for \emph{all of $\tc^0$}. 

Specifically, a common technique to prove correlation bounds for a circuit $C$ is the ``restriction method'', which (loosely speaking) consists of proving the existence of certain subsets of the domain on which $C$ ``simplifies'' (i.e., $C$ agrees with a simpler function on the subset). We pose the following open problem: Construct a deterministic algorithm that gets as input a $\tc^0$ circuit $C$ with $n^{1+O(1/d)}$ wires, runs in sufficiently small sub-exponential time, and finds a subset $S$ of size larger than $2^{n^{1-1/5d}}$ such that the acceptance probability of $C\rest_S$ can be approximated in sufficiently small sub-exponential time (see Open Problem~\ref{prob:rest} in Section~\ref{sec:open} for a precise statement). In Section~\ref{sec:open} we show that a resolution of the foregoing problem would imply that $\nexp\not\subseteq\tc^0$; this follows from Theorem~\ref{thm:int:threshold} and from the techniques that underlie the proof of Theorem~\ref{thm:int:main}.

\subsubsection{The special case of depth-2 circuits} \label{sec:int:depthtwo}

In addition to our main results, we also construct an alternative quantified derandomization algorithm for the special case of linear threshold circuits of \emph{depth two}. Specifically, we construct a pseudorandom generator with seed length $\tilde{O}(\log(n))$ for the class of depth-2 linear threshold circuits with $n^{3/2-\Omega(1)}$ wires that either accept all but $B(n)=2^{n^{\Omega(1)}}$ of their inputs or reject all but $B(n)$ of their inputs. This result is not a corollary of Theorem~\ref{thm:int:main}, and is incomparable to the pseudorandom generator of Servedio and Tan~\cite{st17}. 

The precise result statement and proof appear in Section~\ref{sec:depthtwo}. The generator construction is obtained by leveraging the techniques of Kane and Williams~\cite{kw16} for correlation bounds for linear threshold circuits of depth two.

\subsection{Organization}

In Section~\ref{sec:bg} we provide background and discuss some relevant previous works. In Section~\ref{sec:tech} we give high-level overviews of the proofs of Theorems~\ref{thm:int:main} and~\ref{thm:int:threshold}. After presenting preliminary formal definitions in Section~\ref{sec:pre}, we prove Theorem~\ref{thm:int:main} in Section~\ref{sec:main} and Theorem~\ref{thm:int:threshold} in Section~\ref{sec:threshold}. In Section~\ref{sec:depthtwo} we construct the pseudorandom generator mentioned in Section~\ref{sec:int:depthtwo}. Finally, in Section~\ref{sec:open} we formally pose the open problem that was mentioned in Section~\ref{sec:int:open} and show the consequences of a solution to the problem.

\section{Background and previous work} \label{sec:bg}
\addtocontents{toc}{\protect\setcounter{tocdepth}{1}}

\subsection{Lower bounds for linear threshold circuits} \label{sec:bg:lbs}

The best-known lower bounds for computing explicit functions by linear threshold circuits of a \emph{fixed small depth} have been recently proved by Kane and Williams~\cite{kw16}. Specifically, they showed that any depth-two linear threshold circuit computing Andreev's function requires $\tilde{\Omega}(n^{3/2})$ gates and $\tilde{\Omega}(n^{5/2})$ wires. They also showed correlation bounds (i.e,. average-case lower bounds with respect to the uniform distribution) for such circuits with Andreev's function. Extending their worst-case lower bounds to depth three, they proved that any depth-$3$ circuit with a \emph{top majority gate} that computes a specific polynomial-time computable function also requires $\tilde{\Omega}(n^{3/2})$ gates and $\tilde{\Omega}(n^{5/2})$ wires (the ``hard'' function is a modification of Andreev's function).

For linear threshold circuits of arbitrary constant depth $d\ge2$, the best-known lower bounds on the number of wires required to compute explicit functions are only slightly super-linear. Specifically, Impagliazzo, Paturi, and Saks~\cite{ips97} proved that any linear threshold circuit of depth $d$ requires at least $n^{1+\exp(-d)}$ wires to compute the parity function; Chen, Santhanam, and Srinivasan~\cite{css16} strengthened this by showing correlation bounds for such circuits with parity (as well as with the generalized Andreev function). These lower bounds for parity are essentially tight, since Beame, Brisson, and Ladner~\cite{bbl92} (and later~\cite{ps94}) constructed a linear threshold circuit with $n^{1+\exp(-d)}$ wires that computes parity. We also mention that linear lower bounds on the number of linear threshold \emph{gates} required to compute explicit functions (e.g., the inner-product function) have been proved in several works during the early `90s, and these gate lower bounds apply even for circuits of unrestricted depth (see~\cite{smo90,gt91,ros94,nis93}).

\subsection{Derandomization of LTFs and of functions of LTFs} \label{sec:bg:dernd}

There has been an intensive effort in the last decade to construct pseudorandom generators for a \emph{single linear threshold function}. This problem was first considered by Diakonikolas \etal~\cite{dgjsv10} (see also~\cite{rs10}), and the current state-of-the-art, following~\cite{gowz10,kan11,krs12,mz13,kan14,km15}, is the pseudorandom generator of Gopalan, Kane, and Meka~\cite{gkm15}, which $\e$-fools any LTF with $n$ input bits using a seed of length $\tilde{O}(\log(n/\e))$. Harsha, Klivans, and Meka~\cite{hkm12} considered a \emph{conjunction} of linear threshold functions, and constructed a pseudorandom generator for a subclass of such functions (i.e., for a conjunction of \emph{regular} LTFs; see Section~\ref{sec:pre:ltf} for a definition). Gopalan \etal~\cite{gowz10} constructed pseudorandom generators for small decision trees in which the leaves are linear threshold functions. 

Very recently, Servedio and Tan~\cite{st17} considered the problem of derandomizing \emph{linear threshold circuits}. For every $\e>0$, they constructed a pseudorandom generator that $1/\poly(n)$-fools any \emph{depth-2 linear threshold circuit with at most $n^{2-\e}$ wires,} using a seed of length $n^{1-\de}$, where $\de=\de_\e>0$ is a small constant that depends on $\e$. This yields a derandomization of depth-2 linear threshold circuits with $n^{2-\e}$ wires in time $2^{n^{1-\Omega(1)}}$.

\subsection{Quantified derandomization} \label{sec:bg:quant}

The quantified derandomization problem, which was introduced by Goldreich and Wigderson~\cite{gw14}, is a generalization of the standard derandomization problem. For a circuit class $\mathcal{C}$ and a parameter $B=B(n)$, the $(\mathcal{C},B)$-derandomization problem is the following: Given a description of a circuit $C\in\mathcal{C}$ over $n$ input bits, deterministically distinguish between the case that $C$ accepts all but $B(n)$ of its inputs and the case that $C$ rejects all but $B(n)$ of its inputs. Indeed, the standard derandomization problem is represented by the parameter value $B(n)=\frac{1}{3}\cdot2^n$. Similarly to standard derandomization, a solution for the quantified derandomization problem of a class $\mathcal{C}$ via a ``black-box'' algorithm (e.g., via a pseudorandom generator) yields a corresponding lower bound for $\mathcal{C}$ (see Appendix~\ref{app:quant:lb}). 

Prior to this work, quantified derandomization algorithms have been constructed for $\ac^0$, for subclasses of $\ac^0[\xor]$, for polynomials over $\F_2$ that vanish rarely, and for a subclass of $\mathcal{MA}$. On the other hand, reductions of standard derandomization to quantified derandomization are known for  $\ac^0$, for $\ac^0[\xor]$, for polynomials over large finite fields, and for the class $\mathcal{AM}$ (both the algorithms and the reductions appear in~\cite{gw14,tell17}). In some cases, most notably for $\ac^0$, the parameters of the known quantified derandomization algorithms are very close to the parameters of quantified derandomization to which standard derandomization can be reduced (see~\cite[Thms 1 \& 2]{tell17}).

\section{Overviews of the proofs} \label{sec:tech}
\addtocontents{toc}{\protect\setcounter{tocdepth}{2}}

\subsection{A quantified derandomization algorithm for linear threshold circuits} \label{sec:tech:alg}

The high-level strategy of the quantified derandomization algorithm is as follows. Given a circuit $C:\pmset^n\ra\pmset$, the algorithm deterministically finds a set $S\subseteq\pmset^n$ of size $|S|\gg B(n)$ on which the circuit $C$ simplifies; that is, $C$ agrees with a function from some ``simple'' class of functions on almost all points in $S$. If $C$ accepts all but $B(n)$ of its inputs, then the acceptance probability of $C\rest_S$ will be very high, and similarly, if $C$ rejects all but $B(n)$ of its inputs, then the acceptance probability of $C\rest_S$ will be very low. The algorithm then distinguishes between the two cases, by enumerating the seeds of a pseudorandom generator for the ``simple'' class of functions.

Our starting point in order to construct a deterministic algorithm that finds a suitable set $S$ is the recent proof of correlation bounds for sparse linear threshold circuits by Chen, Santhanam, and Srinivasan~\cite{css16}. Their proof is based on a \emph{randomized} ``whitebox'' algorithm that gets as input a linear threshold circuit with depth $d$ and $n^{1+\e}$ wires, and restricts all but $n^{1-\e\cdot\exp(d)}$ of the variables such that the restricted circuit can be approximated by a single linear threshold function. Thus, if we are able to modify their algorithm to a deterministic one, we will 
%Our goal will be to ``derandomize'' their restriction algorithm, and 
obtain a quantified derandomization algorithm with the parameters asserted in Theorem~\ref{thm:int:main} (i.e., if $\e=\exp(-d)$, then $B(n)\approx|S|/10>2^{n^{1-1/5d}}$).~\footnote{This approach follows the well-known theme of ``leveraging" techniques from lower bound proofs to algorithmic techniques, and in particular to techniques for constructing circuit-analysis algorithms; see, e.g.,~\cite{lmn93,san10,bra10,imz12,st12,imp12,bis12,gmr13,tx13,ckksz15,st17,st17b}. We also mention that in~\cite[Sec. 5]{css16} their randomized restriction algorithm is used to construct a \emph{randomized} algorithm for \emph{satisfiability} of sparse linear threshold circuits.}

Converting the randomized restriction algorithm into a deterministic algorithm poses several challenges, which will be our focus in this overview. Let us first describe the original algorithm, in high-level. The algorithm iteratively reduces the depth of the circuit. In each iteration it applies a random restriction that keeps every variable alive with probability $p=n^{-\Omega(1)}$, and otherwise assigns a random value to the variable. The main structural lemma of~\cite{css16} asserts that such a random restriction turns any LTF to be very biased (i.e., $\exp(-n^{\Omega(1)})$-close to a constant function), with probability $1-n^{-\Omega(1)}$. Hence, after applying the restriction, most gates in the bottom layer of the circuit become very biased, and the fan-in of the rest of the gates in the bottom layer significantly decreases (i.e., we expect it to reduce by a factor of $p=n^{-\Omega(1)}$). The algorithm replaces the very biased gates with the corresponding constants, thereby obtaining a circuit that \emph{approximates} the original circuit (i.e., the two circuits agree on all but $2^{-n^{\Omega(1)}}$ of the inputs); and in~\cite{css16} it is shown that the algorithm can afterwards fix relatively few variables such that the fan-in of each gate that did not become very biased decreases to be at most one (such a gate can be replaced by a variable or a constant). Thus, if the circuit $C_i$ in the beginning of the iteration was of depth $i$, we obtain a circuit $C_{i-1}$ of depth $i-1$ that approximates $C_i$.

One obvious challenge in converting the randomized restriction algorithm into a deterministic algorithm is ``derandomizing'' the main structural lemma; that is, we need to construct a pseudorandom distribution of restrictions that turns any LTF to be very biased, with high probability. The second challenge is more subtle: In each iteration we replace the ``current'' circuit $C_i$ by a circuit $C_{i-1}$ that agrees with $C_i$ on almost all inputs in the subcube of the $n$ living variables (i.e., the circuits disagree on at most $2^{n-n^{\Omega(1)}}$ inputs). However, in subsequent iterations we will fix almost all of these $n$ variables, such that only $n^{1-\Omega(1)}$ variables will remain alive. Thus, we have no guarantee that $C_i$ and $C_{i-1}$ will remain close after additional restrictions in subsequent iterations; in particular, $C_i$ and $C_{i-1}$ might disagree on \emph{all} of the inputs in the subcube of living variables in the end of the entire process. Of course, this is very unlikely to happen when values for fixed variables are chosen uniformly, but we need to construct a \emph{pseudorandom} distribution of restrictions such that the approximation of each $C_i$ by $C_{i-1}$ is likely to be maintained throughout the process.

\subsubsection{Derandomizing the main structural lemma of~\cite{css16}.}
Let $\Phi=(w,\theta)$ be an LTF over $n$ input bits, and consider a random restriction $\rho$ that keeps each variable alive with probability $p=n^{-\Omega(1)}$. Peres' theorem implies that the expected distance of $\Phi\rest_\rho$ from a constant function is approximately $\sqrt{p}$ (see, e.g.,~\cite[Sec. 5.5]{odo14}).~\footnote{Peres' theorem is usually phrased in terms of the noise sensitivity of $\Phi$, but the latter is propotional to its expected bias under a random restriction; for further details see~\cite[Prop. 8]{css16}.} A natural question is whether we can prove a concentration of measure for this distribution. As an illustrative example, consider the majority function $MAJ(x)=\sign(\sum_{i\in[n]}x_i)$; for any $t\ge1$, with probability roughly $1-t\cdot\sqrt{p}$ it holds that $MAJ\rest_\rho$ is $\exp(-t^2)$-close to a constant function (see Fact~\ref{fact:rest:ltf:init}). The main structural lemma in~\cite{css16} asserts that a similar statement indeed holds for \emph{any LTF $\Phi$}; specifically, they showed that with probability at least $1-p^{\Omega(1)}$ it holds that $\Phi\rest_\rho$ is $\exp(-p^{-\Omega(1)})$-close to a constant function. 

We construct a distribution over restrictions that can be efficiently sampled using \linebreak $\tilde{O}(\log(n))$ random bits such that for any LTF $\Phi$ and any $t\ge p^{-1/8}$, with probability at least $1-\tilde{O}(t^{2})\cdot \sqrt{p}$ it holds that $\Phi\rest_\rho$ is $\exp(-t^2)$-close to a constant function. (The actual statement that we prove is more general; see Proposition~\ref{prop:rest:ltf} for precise details.) Indeed, this is both an ``almost-full derandomization'' of the lemma of~\cite{css16} as well as a refinement of the quantitative bound in the lemma.

The original proof of~\cite{css16} relies on a technical case analysis that is reminiscent of other proofs that concern LTFs, and is based on the notion of a \emph{critical index} of a vector $w\in\R^n$ (they refer to the ideas underlying such analyses as ``the structural theory of linear threshold functions''; see, e.g.,~\cite{ser07,dgjsv10}, and Definitions~\ref{def:ltf:reg} and~\ref{def:ltf:crit}). In each case, the main technical tools that are used are concentration and anti-concentration theorems for random weighted sums (i.e., Hoeffding's inequality and the Berry-Ess\'{e}en theorem, respectively), which are used to bound the probability that several specific random weighted sums that are related to the restricted function $\Phi\rest_\rho$ fall in certain intervals.

To derandomize the original proof, an initial useful observation is the following. We say that a distribution ${\bf z}$ over $\pmset^n$ is {\sf $\e$-pseudorandomly concentrated} if for any $w\in\R^n$ and any interval $J\subseteq\R$, the probability that $\ip{w,{\bf z}}$ falls in $J$ is $\e$-close to the probability that $\ip{w,{\bf u}_n}$ falls in $J$ (where ${\bf u}_n$ is the uniform distribution over $\pmset^n$). In particular, the Berry-Ess\'{e}en theorem and Hoeffding's inequality approximately hold for pseudorandom sums $\ip{w,{\bf z}}$ when ${\bf z}$ is pseudorandomly concentrated. The observation is that being $\e$-pseudorandomly concentrated is essentially equivalent to being $\e$-pseudorandom for LTFs (see Claim~\ref{claim:concltfs}).~\footnote{This observation was communicated to us by Rocco Servedio, and is attributed to Li-Yang Tan.} In particular, if a distribution ${\bf z}$ over $\pmset^{n}$ is chosen using the pseudorandom generator of Gopalan, Kane, and Meka~\cite{gkm15} for LTFs, which has seed length $\tilde{O}(\log(n/\e))$, then ${\bf z}$ is $\e$-pseudorandomly concentrated.

The main part in the proof of the derandomized lemma is a (non-trivial) modification of the original case analysis, in order to obtain an analysis in which all claims hold under a suitably-chosen pseudorandom distribution of restrictions. Since this part of the proof is quite technical and low-level, we defer its detailed description to Section~\ref{sec:main:rest:ltf}. However, let us mention that our pseudorandom distribution itself is relatively simple: We first choose the variables to keep alive such that each variable is kept alive with probability approximately $p=n^{-\Omega(1)}$, and the choices are $O(1)$-wise independent; and then we independently choose values for the fixed variables, using the generator of~\cite{gkm15} with error parameter $\e=1/\poly(n)$. We also note that it is suprising that in our setting the case analysis can be modified in order to obtain an ``almost-full derandomization'' (i.e., seed length $\tilde{O}(\log(n))$), since previous derandomizations of similar case analyses regarding LTFs for different settings required much larger seed for error $\e=n^{-\Omega(1)}$ (see~\cite{dgjsv10}).

\subsubsection{Preserving the closeness of the circuit to its approximations.}
%\subsubsection{Maintaining the approximations: A ``bias preservation'' lemma.}

Consider some iteration of the restriction algorithm, in which we start with a circuit $C_i$ of depth $i$, and replace it by a circuit $C_{i-1}$ of depth $i-1$ that only \emph{approximates} $C_i$. (In particular, $C_i$ and $C_{i-1}$ disagree on more inputs than the number of inputs in the final subcube of living variables in the end of the entire restriction process.) Recall that $C_{i-1}$ was obtained by replacing very biased gates in $C_i$ with corresponding constants.  

Our goal now is to show how to choose subsequent restrictions such that with high probability $C_i$ and $C_{i-1}$ will remain close even after applying these restrictions. We will in fact choose each restriction $\rho$ such that the following holds: For each gate $\Phi$ that was replaced by a constant $\sigma\in\pmset$, with probability $1-\frac1{\poly(n)}$ over choice of restriction $\rho$ it holds that $\Phi\rest_\rho$ is still $\frac1{\poly(n)}$-close to $\sigma$ (i.e., $\Pr_x[\Phi\rest_\rho(x)\ne\sigma]\le\frac1{\poly(n)}$; the claim that $C_i$ and $C_{i-1}$ remain close with high probability follows by a union-bound on the gates). Specifically, we prove that if an LTF $\Phi$ is, say, $n^{-200}$-close to a constant $\sigma$, and a restriction $\rho$ is chosen such that the distribution of values for the fixed variables is \emph{$n^{-200}$-pseudorandom for LTFs}, then with probability $1-n^{-10}$ it holds that $\Phi\rest_\rho$ is $n^{-10}$-close to $\sigma$ (see Lemma~\ref{lem:bias}).~\footnote{Since each gate is initially $\exp(-n^{\Omega(1)})$-close to a constant, we can afford a constant number of losses in the polynomial power in the ``closeness'' parameter throught the execution of the restriction algorithm.}

A natural approach to prove such a statement is the following. For any fixed choice of a set $I\subseteq[n]$ of variables to keep alive, we want to choose the values for the fixed variables from a \emph{distribution that ``fools'' a test that checks whether or not $\Phi\rest_\rho$ is close to $\sigma$}. That is, consider a test $T:\pmset^{[n]\setminus I}\ra\pmset$ that gets as input values $z\in\pmset^{[n]\setminus I}$ for the fixed variables $[n]\setminus I$, and decides whether or not $\Phi$ remains close to $\sigma$ in the subcube corresponding to $\rho=\rho_{I,z}$. When $z$ is chosen uniformly, with high probability $\Phi\rest_\rho$ remains close to $\sigma$, and hence the acceptance probability of $T$ is high; thus, any distribution over $\pmset^{[n]\setminus I}$ that is pseudorandom for $T$ also yields, with high probability, values $z\in\pmset^{[n]\setminus I}$ such that $\Phi\rest_{\rho_{I,z}}$ remains close to $\sigma$. The \emph{problem with this approach} is that a test $T$ for such a task above might be very inefficient, since it needs to evaluate $\Phi$ on all points in the subcube corresponding to $\rho=\rho_{I,z}$; thus, we might not be able to construct a pseudorandom generator with short seed to ``fool'' such a ``complicated'' test.

To solve this problem, we use the following general technique that was introduced in our previous work~\cite{tell17}, which is called \emph{randomized tests}. Loosely speaking, a lemma from our previous work implies the following: Assume that there exists a \emph{distribution ${\bf T}$ over tests} $\pmset^{[n]\setminus I}\ra\pmset$ such that for every \emph{fixed input $z$} for which $\Phi\rest_{\rho_{I,z}}$ is $n^{-100}$-close to $\sigma$ it holds that ${\bf T}(z)=-1$, with high probability, and for every \emph{fixed input $z$} for which $\Phi\rest_{\rho_{I,z}}$ is not $n^{-10}$-close to $\sigma$ it holds that ${\bf T}(z)=1$, with high probability. That is, the distribution ${\bf T}$ constitutes a ``randomized test'' that distinguishes, with high probability, between ``excellent'' $z$'s (such that $\Phi\rest_{\rho_{I,z}}$ is very close to $\sigma$) and ``bad'' $z$'s (such that $\Phi\rest_{\rho_{I,z}}$ is relatively far from $\sigma$). Also assume that almost all tests $T:\pmset^{[n]\setminus I}\ra\pmset$ in the support of ${\bf T}$ are ``fooled'' by a pseudorandom generator $G$. Then, with high probability over choice of seed for the pseudorandom generator $G$, the generator outputs $z$ such that $\Phi\rest_{\rho_{I,z}}$ is $n^{-10}$-close to $\sigma$ (see Lemma~\ref{lem:randtest} for a precise and general statement). The main point is that the distribution ${\bf T}$, which may have very high entropy, is \emph{only part of the analysis}; the actual algorithm that generates ${\bf z}$ is simply the pseudorandom generator $G$.

The distribution ${\bf T}$ that we will use is equivalent to the following random process: Given $z\in\pmset^{[n]\setminus I}$, uniformly sample $\poly(n)$ points in the subcube corresponding to $\rho_{I,z}$, and accept $z$ if $\Phi$ evaluates to the constant $\sigma$ on all the sample points. We show how to construct such a distribution ${\bf T}$ such that \emph{almost all} of the residual deterministic tests $T\in\mathtt{support}({\bf T})$ are \emph{conjunctions of $p(n)=\poly(n)$ LTFs}, and have very high acceptance probability (at least $1-1/\poly(p(n))$). Thus, any distribution that is $(1/\poly(n))$-pseudorandom for LTFs is also $(1/\poly(n))$-pseudorandom for almost all tests in the support of ${\bf T}$ (for details see the proof of Lemma~\ref{lem:bias:full}). Combining this statement with the aforementioned general lemma, we deduce the following: If whenever we fix variables we choose the values for the fixed variables according to a distribution that is $(1/\poly(n))$-pseudorandom for LTFs, then with high probability the circuit $C_i$ will remain close to the circuit $C_{i-1}$.

\subsection{Reduction of standard derandomization to quantified derandomization} \label{sec:tech:ext}

Given a $\tc^0$ circuit $C$ of depth $d$ over $m$ input bits, our goal is to construct a $\tc^0$ circuit $C'$ of depth $d'>d$ over $n=\poly(m)$ input bits such that if $C$ accepts (resp., rejects) at least $2/3$ of its inputs then $C'$ accepts (resp., rejects) all but $B(n)=2^{n^{0.99}}$ of its inputs.~\footnote{Throughout the overview we will be somewhat informal with respect to the precise parameter values, e.g. we will use the value $B(n)=2^{n^{0.99}}$ instead of the more precise $B(n)=2^{n^{1-1/5d}}$.} The circuit $C'$ will use its input in order to sample inputs for $C$ by a seeded extractor, and then compute the majority of the evaluations of $C$ on these inputs. Specifically, fixing an extractor $E:\bitset^n\times\bitset^{t}\ra\bitset^m$ for min-entropy $k=n^{0.99}$,~\footnote{The number $B(n)$ of exceptional inputs for $C'$ is upper-bounded by $2^{k}$, and we want to have $B(n)=2^{n^{0.99}}$.} the circuit $C'$ gets input $x\in\bitset^n$, and outputs the majority of the values $\{C(E(x,z)):z\in\bitset^{t}\}$. 

The main technical challenge underlying this strategy is to construct an extractor $E$ such that the mapping of input $x\in\bitset^n$ to the $2^t$ outputs of the extractor on all seeds (i.e., the mapping $x\mapsto\{E(x,z)\}_{z\in\bitset^t}$) can be computed by a $\tc^0$ circuit with \emph{as few wires as possible}. In our construction, the seed length will be $t=1.01\cdot\log(n)$, and thus the number of output bits will be $2^t\cdot m\approx n^{1.01}$; we will construct a $\tc^0$ circuit that computes the mapping of $x$ to these $n^{1.01}$ output bits with only a \emph{super-linear number of wires} (i.e., the number of wires is only slightly larger than the number of output bits). Indeed, a crucial point in our construction is that we will efficiently compute the outputs of the extractor on all seeds in a ``batch'', rather than compute the extractor separately for each seed.

\subsubsection{Our starting point: A construction of $C'$ with $n^{3.01}$ wires}

As our starting point, let us construct a suitable circuit $C'$ that has $n^{3.01}$ wires and is based on Trevisan's extractor~\cite{tre01}. Given an input $x\in\bitset^n$ and seed $z\in\bitset^{t}$, Trevisan's extractor first computes an encoding $\bar{x}$ of $x$ by an $(1/m^2)$-balanced error-correcting code (i.e., a code in which every non-zero codeword has relative Hamming weight $1/2\pm m^{-2}$).~\footnote{Trevisan's extractor only needs a $(1/2-O(1/m),\poly(m))$-list-decodable code, but we will not rely on this potential relaxation.} Fixing a suitable combinatorial design of $m$ sets $S_1,...,S_m$ of size $|S_i|=\log(|\bar{x}|)$ in a universe of size $t$, the output of $E(x,z)$ is the $m$ bits of $\bar{x}$ in the coordinates specified by $z\rest_{S_1},...,z\rest_{S_m}$.

An initial important observation is that the circuit $C'$ only needs to compute the encoding $\bar{x}$ of $x$ \emph{once}, and then each of the $2^{t}$ copies of $C$ can take its inputs directly from the bits of $\bar{x}$ (i.e., each copy of $C$ corresponds to a fixed seed $z$, and takes its inputs from locations in $\bar{x}$ that are determined by $z$ and by the predetermined combinatorial design). This is indeed a form of ``batch computation'' of the extractor on all seeds.

Let us see why this construction uses $n^{3.01}$ wires. To encode $x$ into $\bar{x}$ we can use known polynomial-time constructions of suitable \emph{linear} codes that map $n$ bits to $n\cdot\poly(m)<n^{1.01}$ bits (e.g.,~\cite{nn93,abnnr92,tas17}). Since the code is linear in $x\in\bitset^n$, each bit of $\bar{x}\in\bitset^{n^{1.01}}$ can be computed by a $\tc^0$ circuit with $n^{1.01}$ wires, and thus the number of wires that we use to compute $\bar{x}$ is $n^{2.02}$. Now, recall that we want the extractor to work for min-entropy $k=n^{0.99}$; relying on Trevisan's proof and on standard constructions of combinatorial designs, the required seed length is $t<3\cdot\log(n)$.~\footnote{Trevisan's proof requires a design such that $|S_i\cap S_j|\le\log(k/2m)$ (see~\cite[Sec. 3.3]{tre01}). Relying on standard constructions of combinatorial designs (see, e.g.,~\cite[Lem. 8]{tre01}), a suitable design can be constructed with a universe size of $t=e^{\ln(m)/\log(2k/m)+1}\cdot\frac{\log^2(|x|)}{\log(k/2m)}\approx 1.01\cdot e\cdot \log(n)<3\cdot\log(n)$.} Therefore, the number of copies of $C$ in $C'$ is $2^{t}=n^{3}$, and the overall number of wires in $C'$ is $n^{2.02}+n^3\cdot m<n^{3.01}$.

\subsubsection{The actual construction of $C'$ with $n^{1.01}$ wires}

There are two parts in the construction above that led us to use a large number of wires: First, the seed length of the extractor is $t=3\cdot\log(n)$, which yields $2^{t}=n^3$ copies of $C$; and secondly, the number of wires required to compute the encoding $\bar{x}$ of $x$ is super-quadratic, rather than super-linear. Let us now describe how to handle each of these two problems, and obtain a construction with only $n^{1.01}$ wires.

To reduce the seed length $t$ of the extractor, we follow the approach of Raz, Reingold, and Vadhan~\cite{rrv02}. They showed that Trevisan's extractor works even if we replace standard combinatorial designs by a more relaxed notion that they called \emph{weak designs} (see Definition~\ref{def:threshold:design}). Indeed, weak designs can be constructed with a smaller universe size $t$, which yields a smaller seed length for the extractor. Their construction yields $t=2\cdot\log(n)$, and we show a modified construction of weak designs that for our setting of parameters yields $t=1.01\cdot\log(n)$ (see Lemma~\ref{lem:threshold:designs}).

The second challenge is to construct an $\e$-balanced error-correcting code that maps $n$ bits to $n\cdot\poly(1/\e)$ bits, and can be computed by a $\tc^0$ circuit of depth $d$ with $n^{1+O(1/d)}+n\cdot\poly(1/\e)$ wires (this is the code that we will use to compute $\bar{x}$ from $x$; see Corollary~\ref{cor:threshold:code}). To describe the code, we describe the encoding process of $x\in\bitset^n$, which has two steps: First we encode $x$ by a code with constant rate and constant relative distance, and then perform a second encoding that amplifies the distance of the code to $1/2-\e$.

\paragraph{Computing a code with distance $\Omega(1)$.} In the first step, we encode $x$ by a linear error-correcting code that has distance $\Omega(1)$, instead of $1/2-\e$, and also has rate $\Omega(1)$ and can be computed in $\tc^0$ with $n^{1.01}$ wires. This will be done using \emph{tensor codes} that are based on any (arbitrary) initial good linear error-correcting code.
	
To see why tensor codes are helpful, assume that $n=r^2$, for some $r\in\N$, and fix a linear code $\code$ that maps $r$ bits to $O(r)$ bits and has constant relative distance. Thinking of the input $x\in\bitset^n$ as an $r\times r$ matrix, we first encode each row of the matrix $x$ using $\code$, to obtain an $r\times O(r)$ matrix $x'$, and then encode each column of $x'$ using $\code$, to obtain an $O(r)\times O(r)$ matrix $\hat{x}$. By well-known properties of tensor codes, this yields a linear error-correcting code with constant rate and constant relative distance. Moreover, computing the code in $\tc^0$ only requires $n^{1.51}$ wires: This is because the strings that we encode with $\code$ (which are the rows of $x$ in the first step and then the columns of $x'$ in the second step) are each of length $r=\sqrt{n}$. Thus, each of the $O(n)$ bits in $\hat{x}$ is a linear function of $\sqrt{n}$ bits, and the latter can be computed by $\tc^0$ circuit with $n^{.51}$ wires.
	
To obtain a code with $n^{1.01}$ wires instead of $n^{1.51}$ wires we can use a tensor code of higher order. Specifically, assume that $n=r^{d_0}$, for some large constant $d_0$, and think of $x$ as a tensor of dimensions $[r]^{d_0}$. The encoding process will consist of $d_0=O(1)$ iterations, and in each iteration we encode strings of length $r$ in the tensor by $\code$. The final codeword will be of length $(O(r))^{d_0}=O(n)$, will have constant relative distance, and can be computed by a $\tc^0$ circuit with only $O(n)\cdot r^{1.01}<n^{1+2/d_0}$ wires. (See Section~\ref{sec:threshold:code} for further details.)

\paragraph{Amplifying the distance from $\Omega(1)$ to $1/2-\e$.} Assume that the previous step mapped the input $x\in\bitset^n$ to $\hat{x}\in\bitset^{\hat{n}}$, where $\hat{n}=O(n)$. If $x$ was a non-zero message, then $\hat{x}$ has relative Hamming weight $\Omega(1)$. Our goal now is to increase the Hamming weight of $\hat{x}$ to $1/2-\e$, using as few wires as possible. To do so we rely on the strategy of Naor and Naor~\cite{nn93}, which is based on expander random walks. (This strategy was also recently used by Ta-Shma~\cite{tas17} to construct almost-optimal $\e$-balanced codes.)
	
Specifically, fix a graph $G$ on $\hat{n}$ vertices with constant degree and constant spectral gap. Associate the $\hat{n}$ vertices of $G$ with the coordinates of $\hat{x}$, and consider a random walk on $G$ that starts at a uniformly-chosen vertex and walks $\ell=O(\log(1/\e))$ steps. With probability at least $\e$, such a walk meets the set of coordinates in which $\hat{x}$ is non-zero (since this set has constant density). Thus, if we take such a random walk on the coordinates of $\hat{x}$, and output the parity of a random subset of the bits of $\hat{x}$ that we encountered, with probability at least $1/2-\e$ we will output one.

%To see that this encoding process can be done in sparse $\tc^0$, let us describe the encoding $\bar{x}$ of $\hat{x}$. Every coordinate in $\bar{x}$ is associated with a specific walk $W$ of length $\ell$ on $G$ and with a subset $S\subseteq[\ell]$; thus, $\bar{x}$ has $2^{\log(n)+O(\ell)}=n\cdot\poly(1/\e)$ coordinates. The bit of $\bar{x}$ at a coordinate associated with a walk $W$ and with a subset $S\subseteq[\ell]$ is the parity of the $S$ bits of $\hat{x}$ encountered in the walk $W$. Thus, each bit in $\bar{x}$ is the parity of at most $\ell=O(\log(1/\e))=O(\log(n))$ bits, and therefore computing $\bar{x}$ from $\hat{x}$ only requires $n\cdot\poly(1/\e)\cdot\ell^{1.01}$ wires. Recall that in our setting we need $\e=1/m^2=n^{-\Omega(1)}$; the number of wires is thus at most $n^{1.01}$.

The encoding $\bar{x}$ of $\hat{x}$ is thus the following. Every coordinate in $\bar{x}$ is associated with a specific walk $W$ of length $\ell$ on $G$ and with a subset $S\subseteq[\ell]$; thus, $\bar{x}$ has $2^{\log(n)+O(\ell)}=n\cdot\poly(1/\e)$ coordinates. The bit of $\bar{x}$ at a coordinate associated with a walk $W$ and with a subset $S\subseteq[\ell]$ is the parity of the $S$ bits of $\hat{x}$ encountered in the walk $W$. Thus, each bit in $\bar{x}$ is the parity of at most $\ell=O(\log(1/\e))$ bits in $\hat{x}$, so computing $\bar{x}$ from $\hat{x}$ only requires $n\cdot\poly(1/\e)\cdot\ell^{1.01}=n\cdot\poly(1/\e)$ wires. Recall that in our setting we need $\e=1/m^2=n^{-\Omega(1)}$; the number of wires is thus at most $n^{1.01}$. By the preceding paragraph, if $\hat{x}$ has Hamming weight $\Omega(1)$ then $\bar{x}$ has Hamming weight at least $1/2-\e$.

\section{Preliminaries} \label{sec:pre}
\addtocontents{toc}{\protect\setcounter{tocdepth}{1}}

Throughout the paper, the letter $n$ will always denote the number of inputs to a function or a circuit. We denote random variables by boldface letters, and denote by ${\bf u}_n$ the uniform distribution on $n$ bits.

We are interested in Boolean functions, represented as functions $f:\pmset^n\ra\pmset$. We say that a function $f:\pmset^n\ra\pmset$ accepts an input $x\in\pmset^n$ if $f(x)=-1$. For two Boolean functions $f$ and $g$ over a domain $\mathfrak{D}$, we say that $f$ and $g$ are $\de$-close if $\Pr_{x\in\mathfrak{D}}[f(x)=g(x)]\ge1-\de$.

For a vector $w=(w_1,...,w_n)\in\R^n$, we denote by $\norm{w}_2$ the standard $\ell_2$-norm $\norm{w}_2=\sqrt{\sum_{i\in[n]}w_i^2}$. For $h<n$, we denote $w_{>h}=(w_{h+1},...,w_n)\in\R^{n-h}$ and $w_{\ge h}=(w_h,...,w_n)\in\R^{n-h+1}$. For two vectors $w,x\in\R^n$, we denote $\ip{w,x}=\sum_{i\in[n]}w_i\cdot x_i$.

\subsection{Two probabilistic inequalities} \label{sec:pre:prob}

We will rely on two standard facts from probability theory that assert concentration and anti-concentration bounds for certain distributions. Specifically, we will need a standard version of Hoeffding's inequality, and a corollary of the Berry-Ess\'{e}en theorem:

\begin{theorem} (Hoeffding's inequality; for a proof see, e.g.,~\cite[Sec. 1.7]{dp09}). \label{thm:hoef}
Let $w\in\R^n$, and let ${\bf z}$ be a uniformly-chosen random vector in $\pmset^n$. Then, for any $t>0$ it holds that
\mm{
\Pr\left[ |\ip{w,{\bf z}}|\ge t\cdot\norm{w}_2\right]\le\exp(-\Omega(t^2)) \mathdot
}
\end{theorem}

\begin{theorem} (a corollary of the Berry-Ess\'{e}en theorem; see, e.g.,~\cite[Thm 2.1, Cor 2.2]{dgjsv10}). \label{thm:be}
Let $w\in\R^n$ and $\mu>0$ such that for every $i\in[n]$ it holds that $|w_i|\le\mu\cdot\norm{w}_2$, and let ${\bf z}$ be a uniformly-chosen random vector in $\pmset^n$. Then, for any $\theta\in\R$ and $t>0$ it holds that:
\mm{
\Pr\left[ \ip{w,{\bf z}} \in \theta\pm t\cdot\norm{w}_2 \right] \le 2\cdot( t + \mu ) \mathdot
}
\end{theorem}

\subsection{Linear threshold functions and circuits} \label{sec:pre:ltf}

A linear threshold function (or LTF, in short) $\Phi:\pmset^n\ra\pmset$ is a function of the form $\Phi(x)=\sign(\ip{x,w}-\theta)$, where $w\in\R^n$ is a vector of real ``weights'', and $\theta\in\R$ is a real number (the ``threshold''), and $\ip{x,w}=\sum_{i\in[n]}x_i\cdot w_i$ denotes the standard inner-product over the reals.~\footnote{When dealing with LTFs we can assume, without loss of generality, that $\ip{w,x}\ne\theta$ for every $x\in\pmset^n$ (because for every Boolean function over $\pmset^n$ that is computable by an LTF there exists an LTF that computes the function such that $\ip{w,x}\ne\theta$ for every $x\in\pmset^n$).} Indeed, the majority function is the special case where the weights are identical (e.g., $w_i=1$ for all $i\in[n]$) and the threshold is zero (i.e., $\theta=0$).

We will be interested in {\sf linear threshold circuits}, which are circuits that consist only of LTF gates with unbounded fan-in and fan-out. We assume that linear threshold circuits are \emph{layered}, in the sense that for each gate $\Phi$, all the gates feeding into $\Phi$ have the same distance from the inputs. For $n,d,m\in\N$, let $\mathcal{C}_{n,d,m}$ be the class of linear threshold circuits over $n$ input bits of depth $d\ge1$ and with at most $m$ wires. For some fixed sizes and depths, linear threshold circuits are known to be stronger than circuits with majority gates; however, linear threshold circuits can be simulated by circuits with majority gates with a polynomial size overhead and with one additional layer (see~\cite{ghr92,gk98}). Thus, the class $\tc^0$ as a whole equals the class of linear threshold circuits.

%Recall that $\tc^0$ is the class of constant-depth, polynomial-sized circuits with majority gates. A common theme in the study of $\tc^0$ is the use of an equivalent definition for the class, in which each gate in the circuit computes an LTF. We will use the name {\sf linear threshold circuits} to denote constant-depth, polynomial-sized circuits with LTF gates. For some fixed sizes and depths, linear threshold circuits are known to be stronger than circuits with majority gates; however, linear threshold circuits can be simulated by circuits with majority gates with a polynomial size overhead and with one additional layer (see~\cite{ghr92,gk98}). Thus, the class $\tc^0$ as a whole equals the class of linear threshold circuits.

%A linear threshold function, or LTF in short, is a function $\Phi:\pmset^n\ra\pmset$ of the form $\Phi(x)=\sign(\ip{w,x}-\theta)$, where $w\in\R^n$ and $\theta\in\R$; we typically describe such a function by the pair $\Phi=(w,\theta)$. 
The following are standard definitions (see, e.g.,~\cite{ser07,dgjsv10}), which refer to ``structural'' properties of LTFS and will be useful for us throughout the paper.

\begin{definition} (regularity). \label{def:ltf:reg}
For $\e>0$, we say that a vector $w\in\R^n$ is {\sf $\e$-regular} if for every $i\in[n]$ it holds that $|w_i|\le\e\cdot\norm{w}_2$. An LTF $\Phi=(w,\theta)$ is $\e$-regular if $w$ is $\e$-regular. 
\end{definition}

\begin{definition} (critical index). \label{def:ltf:crit}
When $w\in\R^n$ satisfies $|w_1|\ge|w_2|\ge...\ge|w_n|$, the {\sf $\e$-critical index} of $w$ is defined as the smallest $h\in[n]$ such that $w_{>h}$ is $\e$-regular (and $h=\infty$ if no such $h\in[n]$ exists). The critical index of an LTF $\Phi=(w,\theta)$ is the critical index of $w'$, where $w'\in\R^n$ is the vector that is obtained from $w$ by permuting the coordinates in order to have $|w'_1|\ge...\ge|w'_n|$. 
\end{definition}

\begin{definition} (balanced LTF). \label{def:ltf:bal}
For $t\in\R$, we say that an LTF $\Phi=(w,\theta)$ is $t$-balanced if $|\theta|\le t\cdot\norm{w}_2$; otherwise, we say that $\Phi$ is $t$-imbalanced.
\end{definition}

\paragraph{Representation of linear threshold circuits}

The algorithm in Theorem~\ref{thm:int:main} gets as input an explicit representation of a linear threshold circuit $C$, where the weights and thresholds of the LTFs in $C$ may be arbitrary real numbers. Throughout the paper we will not be specific about how exactly $C$ is represented as an input to the algorithm, since the algorithm works in any reasonable model. In particular, the algorithm only performs \emph{addition, subtraction, and comparison operations} on the weights and thresholds of the LTFs in $C$. 

Explicitly suggesting one convenient model, one may assume that the weights and threshold of each LTF are  integers of unbounded magnitude (since the real numbers can be truncated at some finite precision without changing the function). In this case, the circuit $C$ has a binary representation, and the required time to perform addition, subtraction, and comparison on these integers is linear in the representation size.~\footnote{It is well-known that every LTF over $n$ input bits has a representation with integer weights of magnitude $2^{\tilde{O}(n)}$ (for proof see, e.g.,~\cite{has94}), and therefore the circuit $C$ actually has a representation of size $\poly(n)$. However, we do not know of a polynomial-time algorithm to find such a representation for a given circuit $C$.}

\subsection{Pseudorandomness} 

We need the following two standard definitions of pseudorandom distributions and of pseudorandom generators (or PRGs, in short).

\begin{definition} (pseudorandom distribution). \label{def:prdist}
For $\e>0$ and a domain $\mathfrak{D}$, we say that a distribution ${\bf z}$ over $\mathfrak{D}$ is $\e$-pseudorandom for a class of functions $\mathcal{F}\subseteq\left\{\mathfrak{D}\ra\pmset\right\}$ if for every $f\in\mathcal{F}$ it holds that $\Pr_{z\sim{\bf z}}\left[f(z)=-1\right]\in\Pr_{z\in\mathfrak{D}}\left[f(z)=-1\right]\pm\e$.
\end{definition}

\begin{definition} (pseudorandom generator). \label{def:prg}
Let $\mathcal{F}=\bigcup_{n\in\N}\mathcal{F}_n$, where for every $n\in\N$ it holds that $\mathcal{F}_n$ is a set of functions $\pmset^n\ra\pmset$, and let $\e:\N\ra[0,1]$ and $\ell:\N\ra\N$. An algorithm $G$ is a {\sf pseudorandom generator for $\mathcal{F}$ with error parameter $\e$ and seed length $\ell$} if for every $n\in\N$, when $G$ is given as input $1^n$ and a random seed of length $\ell(n)$, the output distribution of $G$ is $\e$-pseudorandom for $\mathcal{F}_n$.
\end{definition}

We will rely on the following recent construction of a pseudorandom generator for LTFs, by Gopalan, Kane, and Meka~\cite{gkm15}:

\begin{theorem} (a PRG for LTFs;~\cite[Cor. 1.2]{gkm15}). \label{thm:prg:ltfs}
For every $\e>0$, there exists a polynomial-time pseudorandom generator for the class of LTFs with seed length $O\left(\log(n/\e)\cdot(\log\log(n/\e))^2\right)$.
\end{theorem}

A distribution ${\bf z}$ over $\pmset^n$ is {\sf $\de$-almost $t$-wise independent} if for every $S\subseteq[n]$ of size $|S|=t$ it holds that ${\bf z}_S$ is $\de$-close to the uniform distribution over $\pmset^t$ in statistical distance. We will need the following standard tail bound for such distributions.

\begin{fact} (tail bound for almost $t$-wise independent distributions). \label{fact:tail}
Let $t\ge4$ be an even number, and let $\de:\N\ra[0,1]$. Let ${\bf x}_1,...,{\bf x}_n$ be variables in $\bitset$ that are $\de(n)$-almost $t$-wise independent, and denote $\mu=\E\left[\frac1{n}\cdot\sum_{i\in[n]}{\bf x}_i\right]$. Then, for any $\zeta>0$ it holds that $\Pr\left[\abs{\frac1{n}\cdot\sum_{i\in[n]}{\bf x}_i-\mu}\ge\zeta\right]<8\cdot\left(\frac{t\cdot\mu\cdot n+t^2}{\zeta^2\cdot n^2}\right)^{t/2}+(2\cdot n)^t\cdot\de(n)$.

In particular, for $t=\Theta(1)$ and $\zeta=\mu/2$ and $\de(n)=1/p(n)$, where $p(n)$ is a sufficiently large polynomial, we have that
\mm{
\Pr\left[\frac1{n}\cdot\sum_{i\in[n]}{\bf x}_i\in\mu\pm(\mu/2)\right]=O\left( (\mu\cdot n)^{-t/2}\right) \mathdot
}
\end{fact}

%\paragraph{Distributions that are $\e$-pseudorandomly concentrated.} 
We now define the notion of a distribution that is {\sf $\e$-pseudorandomly concentrated}, and show that it is essentially equivalent to the notion of being $\e$-pseudorandom for LTFs. The equivalence was communicated to us by Rocco Servedio, and is attributed to Li-Yang Tan.

\begin{definition} ($\e$-pseudorandomly concentrated distribution). \label{def:prg:conc}
For $n\in\N$ and $\e>0$, we say that a distribution ${\bf z}$ over $\pmset^n$ is {\sf $\e$-pseudorandomly concentrated} if the following holds: For every $w\in\R^n$ and every $a<b\in\R$ it holds that $\Pr\left[ \ip{w,{\bf z}}\in[a,b] \right] \in \Pr\left[ \ip{w,{\bf u}_n}\in[a,b]\right] \pm \e$.
\end{definition}

\begin{claim} (being pseudorandomly concentrated is equivalent to being pseudorandom for LTFs). \label{claim:concltfs}
Let ${\bf z}$ be a distribution over $\pmset^n$. Then,
\begin{enumerate}
	\item \label{it:concltfs:ltf} If ${\bf z}$ is $\e$-pseudorandom for LTFs, then ${\bf z}$ is $(2\e)$-pseudorandomly concentrated.
	\item \label{it:concltfs:conc} If ${\bf z}$ is $\e$-pseudorandomly concentrated, then ${\bf z}$ is $\e$-pseudorandom for LTFs.
\end{enumerate}
\end{claim}

\begin{proof}[{\bf Proof.}]
Let us first prove Item~\eqref{it:concltfs:ltf}. Fix $w\in\R^n$ and $I=[a,b]\subseteq\R$. For any fixed $z\in\pmset^n$, exactly one of three events happens: Either $\ip{w,z}\in I$, or $\ip{w,z}<a$, or $\ip{w,z}>b$. Since the event $\ip{w,z}<a$ can be tested by an LTF (i.e., by the LTF $\Phi(z)=\sign(a-\ip{w,z})$), this event happens with probability $\Pr_{z\in\pmset^n}\left[ \ip{w,z}<a \right] \pm \e$ under a choice of $z\sim{\bf z}$. Similarly, the event $\ip{w,z}>b$ happens with probability $\Pr_{z\in\pmset^n}\left[ \ip{w,z}>b \right] \pm \e$ under a choice of $z\sim{\bf z}$. Thus, the probability under a choice of $z\sim{\bf z}$ that $\ip{w,z}\in I$ is $\Pr_{z\in\pmset^n}\left[ \ip{w,z}\in I \right]\pm2\e$.

To see that Item~\eqref{it:concltfs:conc} holds, let $\Phi=(w,\theta)$ be an LTF over $n$ input bits, and let $M=\norm{w}_1=\sum_{i\in[n]}|w_i|$. Then, for every $z\in\pmset^n$ it holds that $\Phi(z)=-1$ if and only if $z\in[-M,\theta]$. Thus, $\Pr[\Phi({\bf z})=-1]=\Pr[{\bf z}\in[-M,\theta]]\in\Pr[{\bf u}_n\in[-M,\theta]]\pm\e=\Pr[\Phi({\bf u}_n)=-1]\pm\e$.
\end{proof}

\subsection{Restrictions}

A restriction for functions $\pmset^n\ra\pmset$ is a subset of $\pmset^n$. We will be interested in restrictions that are subcubes, and such restrictions can be described by a string $\rho\in\{-1,1,\star\}^n$ in the natural way (i.e., the subcube consists of all strings $x\in\pmset^n$ such that for every $i$ such $\rho_i\ne\star$ it holds that $x_i=\rho_i$). We will sometimes describe a restriction by a pair $\rho=(I,z)$, where $I=\{i\in[n]:\rho_i=\star\}$ is the set of variables that the restriction keeps alive, and $z=(\rho_i)_{i\in([n]\setminus I)}\in\pmset^{[n]\setminus I}$ is the sequence of values that $\rho$ assigns to the variables that are fixed.

We identify strings $r\in\pmset^{(q+1)\cdot n}$, where $n,q\in\N$, with restrictions $\rho=\rho_r\in\{-1,1,\star\}^n$, as follows: Each variable is assigned a block of $q+1$ bits in the string; the variable remains alive if the first $q$ bits in the block are all $1$, and otherwise takes the value of the $(q+1)^{th}$ bit. When we refer to a ``block'' in the string that corresponds to a restriction, we mean a block of $q+1$ bits that corresponds to some variable. When we say that a restriction is chosen from a distribution ${\bf r}$ over $\pmset^{(q+1)\cdot n}$, we mean that a string is chosen according to ${\bf r}$, and interpreted as a restriction.

In addition, we will sometimes identify a \emph{pair} of strings $y\in\pmset^{q\cdot n}$ and $z\in\pmset^n$ with a restriction $\rho=\rho_{y,z}$. In this case, the restriction $\rho=\rho_{y,z}$ is the restriction $\rho_r$ that is obtained by combining $y$ and $z$ to a string $r$ in the natural way (i.e., appending a bit from $z$ to each block of $q$ bits in $y$). Note that the string $y$ determines which variables $\rho$ keeps alive, and the string $z$ determinse the values that $\rho$ assigns to the fixed variables.

\subsection{Seeded extractors and averaging samplers} \label{sec:pre:ext}

We recall the standard definitions of seeded extractors and of averaging samplers, and state the well-known equivalence between the two. In this context it will be more convenient to represent Boolean functions as functions $\bitset^n\ra\bitset$.

\begin{definition} (seeded extractors). \label{def:threshold:ext}
A function $f:\bitset^n\times\bitset^t\ra\bitset^m$ is a {\sf $(k,\e)$-extractor} if for every distribution ${\bf x}$ on $\bitset^n$ such that $\max_{x\in\bitset^n}\left[\Pr[{\bf x}=x]\right]\le2^{-k}$ it holds that the distribution $f({\bf x},{\bf u}_t)$ is $\e$-close to the uniform distribution on ${\bf u}_m$ in statistical distance.
\end{definition}

\begin{definition} (averaging samplers). \label{def:threshold:samp}
A function $f:\bitset^n\times\bitset^t\ra\bitset^m$ is an {\sf averaging sampler with accuracy $\e>0$ and error $\de>0$} if it satisfies the following. For every $T\subseteq\bitset^m$, for all but a $\de$-fraction of the strings $x\in\bitset^n$ it holds that $\Pr_{z\in\bitset^t}[f(x,z)\in T]=|T|/2^m\pm\de$.
\end{definition}

\begin{proposition} (seeded extractors are equivalent to averaging samplers). \label{prop:threshold:extsamp}
Let $f:\bitset^n\times\bitset^t\ra\bitset^m$. Then, the following two assertions hold:
\begin{enumerate}
	\item If $f$ is a $(k,\e)$-extractor, then $f$ is an averaging sampler with accuracy $\e$ and error $\de=2^{k-n}$.
	\item If $f$ is an averaging sampler with accuracy $\e$ and error $\de$, then $f$ is an $(n-\log(\e/\de),2\e)$-extractor.
\end{enumerate}
\end{proposition}

For a proof of Proposition~\ref{prop:threshold:extsamp} see, e.g.,~\cite[Cor. 6.24]{vad12}. In the current paper we will only use the first item of Proposition~\ref{prop:threshold:extsamp}.

\section{A quantified derandomization algorithm for linear threshold circuits} \label{sec:main}
\addtocontents{toc}{\protect\setcounter{tocdepth}{2}}

Let us now state a more general version of Theorem~\ref{thm:int:main} and prove it.

\begin{theorem} (Theorem~\ref{thm:int:main}, restated). \label{thm:main}
Let $d\ge1$, let $\e>0$, and let $\de=d\cdot30^{d-1}\cdot\e$. Then, there exists a deterministic algorithm that for every $n\in\N$, when given as input a circuit $C\in\mathcal{C}_{n,d,n^{1+\e}}$, runs in time $n^{O(\log\log(n))^2}$, and for the parameter $B(n)=\frac1{10}\cdot2^{n^{1-\de}}$ satisfies the following:
\begin{enumerate}
	\item If $C$ accepts all but at most $B(n)$ of its inputs, then the algorithm accepts $C$.
	\item If $C$ rejects all but at most $B(n)$ of its inputs, then the algorithm rejects $C$.
\end{enumerate}
\end{theorem}

To obtain the parameters of Theorem~\ref{thm:int:main}, for any $d\ge1$ let $\e=2^{-10d}$. Then, the algorithm from Theorem~\ref{thm:main} works when the number of exceptional inputs of $C$ is at most $B(n)=\frac1{10}\cdot2^{n^{1-\de}}>2^{n^{1-1/5d}}$. The deterministic algorithm from Theorem~\ref{thm:main} is based on the following \emph{pseudorandom restriction algorithm}, whose construction and proof appear in Section~\ref{sec:depth}.

\begin{proposition} (pseudorandom restriction algorithm). \label{prop:main:rest:circuit}
Let $d\ge1$, let $\e>0$ be a sufficiently small constant, and let $\de=d\cdot30^{d-1}\cdot\e$. Then, there exists a polynomial-time algorithm that for every $n\in\N$, when given as input a circuit $C\in\mathcal{C}_{n,d,n^{1+\e}}$ and a random seed of length $O(\log(n)\cdot(\log\log(n))^2)$, with probability at least $1-n^{-\e/2}$ satisfies the following:
\begin{enumerate}
	\item The algorithm outputs a restriction $\rho\in\pmstar^n$ that keeps at least $n^{1-\de}$ variables alive.
	\item The algorithm outputs an LTF $\Phi:\pmset^{\rho^{-1}(\star)}\ra\pmset$ such that $\Phi$ is $1/10$-close to $C\rest_\rho$ (i.e., $\Pr_{x\in\pmset^{\rho^{-1}(\star)}}[C(x)=\Phi(x)]\ge9/10$).
\end{enumerate}
\end{proposition}

Let us now prove the main result (i.e., Theorem~\ref{thm:main}) relying on Proposition~\ref{prop:main:rest:circuit}.

\begin{proof}[{\bf Proof of Theorem~\ref{thm:main}.}]
We iterate over all seeds for the algorithm from Proposition~\ref{prop:main:rest:circuit}. 
%All but at most $O(n^{-\e})$ of the seeds yield a restriction $\rho\in\pmset^n$ that keeps $n^{1-\de}$ variables alive, and also yield an LTF $\Phi:\pmset^{\rho^{-1}(\star)}\ra\pmset$ that is $1/10$-close to $C\rest_\rho$; we call such seeds \emph{good} seeds. %(There may be good seeds such that $\Phi$ is not $1/10$-close to $C\rest_\rho$, but there will be at most $O(n^{-\e})$ such seeds.)
For each seed that yields both a restriction $\rho$ that keeps at least $n^{1-\de}$ variables alive and an LTF $\Phi$ over $\pmset^{\rho^{-1}(\star)}$, we estimate the acceptance probability of $\Phi$ up to an error of $\frac1{5}$; this is done by iterating over the seeds of the pseudorandom generator from Theorem~\ref{thm:prg:ltfs} (instantiated with error parameter $1/5$). If for most of the seeds, our estimate of the acceptance probability of $\Phi$ is at least $\frac3{5}$, then we accept $C$; and otherwise we reject $C$. The running time of the algorithm is $2^{O(\log(n)\cdot(\log\log(n))^2)}=n^{O(\log\log(n))^2}$.

Recall that all but $O(n^{-\e})$ of the seeds yield $\rho$ and $\Phi$ such that $\rho$ keep at least $n^{1-\de}>\log(10\cdot B(n))$ variables alive and such that $\Phi$ is $1/10$-close to $C\rest_\rho$; we call such seeds \emph{good} seeds. Now, if $C$ accepts all but at most $B(n)$ inputs, then for every good seed, the acceptance probability of $C\rest_\rho$ is at least $9/10$, and thus the acceptance probability of $\Phi$ is at least $\frac4{5}$, which implies that our estimate of the latter will be at least $3/5$. Thus, the algorithm will accept $C$. On the other hand, if $C$ rejects all but at most $B(n)$ inputs, then by a similar argument for all good seeds it holds that the estimate of the acceptance probability of $\Phi$ will be at most $2/5$, and thus the algorithm will reject $C$.
\end{proof}

\subsection{Pseudorandom restriction algorithm} \label{sec:depth}

We prove Proposition~\ref{prop:main:rest:circuit} in three steps. The first step, in Section~\ref{sec:main:rest:ltf}, will be to prove that a suitably-chosen pseudorandom restriction turns any single LTF to be very biased, with high probability. The second step, in Section~\ref{sec:main:rest:layer}, will leverage the first step to construct an algorithm that gets as input a linear threshold circuit, and applies pseudorandom restrictions to reduce the depth of the circuit by one layer. And the final step, in Section~\ref{sec:main:rest:circuit}, will be to iterate the construction of the second step in order to prove Proposition~\ref{prop:main:rest:circuit}.

\subsubsection{Pseudorandom restrictions and a single LTF} \label{sec:main:rest:ltf}

As mentioned in the introduction, an illustrative example for the effects of restrictions on LTFs is the majority function $\Phi(x)=\sign(\sum_{i\in[n]}x_i)$. For $p\in(0,1)$, denote by ${\bm{\mathcal{R}}}_p$ the distribution of restrictions on $n$ variables such that for every $i\in[n]$ independently it holds that the $i^{th}$ variable remains alive with probability $p$, and is otherwise assigned a uniform random bit. Then, we have the following:

\begin{fact} (a random restriction and the majority function). \label{fact:rest:ltf:init}
Let $\Phi(x)=\sign(\sum_{i\in[n]}x_i)$, and let $p=n^{-\Omega(1)}$. Then, for every $t\ge1$, with probability at least $1-O(t\cdot\sqrt{p})$ over $\rho\sim{\bm{\mathcal{R}}}_p$ it holds that $\Phi\rest_{\rho}$ is $t$-imbalanced 
\end{fact}

\begin{proof}[{\bf Proof.}]
Let $I\subseteq[n]$ be the set of variables that $\rho$ keeps alive. With probability $1-\exp(-n^{\Omega(1)})$ it holds that $\norm{w_I}_2\in \sqrt{pn}\pm\sqrt{pn}/2$. Conditioned on $\norm{w_I}_2\le2\cdot\sqrt{pn}$, it also holds that $\norm{w_{[n]\setminus I}}_2\ge\sqrt{n}/2$, which implies that for every $i\in([n]\setminus I)$ it holds that $|w_i|=1\le (2/\sqrt{n})\cdot\norm{w_{[n]\setminus I}}_2$. In this case, by the Berry-Ess\'{e}en theorem (i.e., by Theorem~\ref{thm:be}), for any $t\ge1$, the probability that $\ip{w_{[n]\setminus I},z_{[n]\setminus I}}$ falls in the interval $\pm4t\cdot\sqrt{p}\cdot\norm{w_{[n]\setminus I}}_2$ (which contains the interval $\pm t\cdot\norm{w_I}_2$) is at most $O(t\cdot\sqrt{p}+\frac{2}{\sqrt{n}})=O(t\cdot\sqrt{p})$.
\end{proof}

Our goal in this section is to prove a statement that is similar to Fact~\ref{fact:rest:ltf:init}, but that holds for an \emph{arbitrary LTF} $\Phi$, and holds also when the restriction $\rho$ is sampled pseudorandomly, rather than uniformly. For simplicity, we only state the proposition informally at the moment (for a formal statement see Proposition~\ref{prop:rest:ltf}):

\begin{proposition} (pseudorandom restriction lemma for a single LTF; informal). \label{prop:rest:ltf:informal}
Let $n\in\N$, let $p=n^{-\Omega(1)}$, and let $t=p^{-\Omega(1)}$. Let ${\bf y}$ be a distribution over $\pmset^{\log(1/p)\cdot n}$ that is $p$-almost $O(\log(1/p))$-wise independent, and let ${\bf z}$ be a distribution over $\pmset^n$ that is $p^{\Omega(1)}$-pseudorandomly concentrated. Then, for any LTF $\Phi$ over $n$ input bits, the probability over choice of restriction $\rho\sim({\bf y},{\bf z})$ that $\Phi\rest_\rho$ is $t$-balanced is at most $p^{\Omega(1)}$.
\end{proposition}

\paragraph{A high-level description of the proof.} Let $\Phi=(w,\theta)$ be an LTF over $n$ input bits, and without loss of generality assume that $|w_1|\ge|w_2|\ge...\ge|w_n|$. Denote by $I\subseteq[n]$ the set of variables that ${\rho}$ keeps alive, and by $z_{[n]\setminus I}\in\pmset^{[n]\setminus I}$ the values that ${\rho}$ assigns to the fixed variables. Then, the restricted function is of the form $\Phi\rest_{\rho}=\left(w_I,\theta-\ip{w_{[n]\setminus I},{z}_{[n]\setminus I}}\right)$, and the restricted function is $t$-balanced if and only if the sum $\ip{w_{[n]\setminus I},{ z}_{[n]\setminus I}}$ falls in the interval $\theta\pm2t\cdot\norm{w_I}_2$. Our goal will be to show that this event is unlikely. 

The proof is based on a modification of the case analysis that appears in~\cite[Lem. 34, Sec. 4.2, Apdx. C.]{css16}. Specifically, for the parameter values $\mu=\Omega(1/t)$ and $k=\tilde{O}(t^2)$, we will consider two separate cases.

\bigskip\noindent\emph{Case 1: The $\mu$-critical index of $\Phi$ is at most $k$.}
Let $h\le k$ be the $\mu$-critical index of $\Phi$, and denote $T=[n]\setminus[h]$. We first claim that with probability $1-p^{\Omega(1)}$ over choice of $y\sim{\bf y}$ it holds that $\norm{w_I}_2\le p^{\Omega(1)}\cdot\norm{w_T}_2$. This is the case since with probability at least $1-h\cdot p=1-p^{\Omega(1)}$, all the first $h$ variables are fixed by ${\rho}$, and since the expected value of $\norm{w_{I\cap T}}_2$ is $\sqrt{p}\cdot\norm{w_{T}}_2$.

Condition on any fixed choice of $y\sim{\bf y}$ such that $\norm{w_I}_2\le p^{\Omega(1)}\cdot\norm{w_T}_2$. We will prove that with probability $1-p^{\Omega(1)}$ over a \emph{uniform choice of $z\in\pmset^n$} it holds that $\ip{w_{[n]\setminus I},z_{[n]\setminus I}}$ does not fall in the interval $\theta\pm t\cdot p^{\Omega(1)}\cdot\norm{w_{T}}_2$ (which contains the interval $\theta\pm t\cdot\norm{w_I}_2$, due to our fixed choice of $y$). Since ${\bf z}$ is $p^{\Omega(1)}$-pseudorandomly concentrated, it will follow that this event also holds with probability $1-p^{\Omega(1)}$ under a choice of $z\sim{\bf z}$.

To prove the claim about a uniform choice of $z\in\pmset^n$, condition \emph{any arbitrary fixed values} $z_{[h]}\in\pmset^{h}$ for the first $h$ variables. Then, the probability that $\ip{w_{[n]\setminus I},z_{[n]\setminus I}}$ falls in the interval $\theta\pm t\cdot p^{\Omega(1)}\cdot\norm{w_{T}}_2$ (which is what we want to bound) equals the probability that $\ip{w_{T\setminus I},z_{T\setminus I}}_2$ falls in the interval $\theta'\pm t\cdot p^{\Omega(1)}\cdot\norm{w_T}_2$, where $\theta'=\theta-\ip{w_{[h]},z_{[h]}}$. Since $h$ is the $\mu$-critical index of $w$ we have that $w_T$ is $\mu$-regular; also, since $\norm{w_I}_2\le p^{\Omega(1)}\cdot\norm{w_{T}}_2$ (due to our choice of $y$), it follows that $w_{T\setminus I}$ is also $(2\mu)$-regular and that $\norm{w_T}_2\approx\norm{w_{T\setminus I}}_2$. By the Berry-Ess\'{e}en theorem, the probability that $\ip{w_{T\setminus I},z_{T\setminus I}}$ falls in an interval of length $t\cdot p^{\Omega(1)}\cdot\norm{w_{T\setminus I}}_2$ is at most $O(t\cdot p^{\Omega(1)}+\mu)=p^{\Omega(1)}$ (see Lemma~\ref{lem:rest:ltf:reg}).

\bigskip\noindent\emph{Case 2: The $\mu$-critical index of $\Phi$ is larger than $k$.}
Similarly to the previous case, with probability at least $1-p^{\Omega(1)}$ it holds that all the first $k$ variables are fixed by ${\rho}$. Condition on any fixed $y\sim{\bf y}$ that fixes all the first $k$ variables. What we will show is that with high probability over $z\sim{\bf z}$, the sum $\ip{w_{[n]\setminus I},{\bf z}_{[n]\setminus I}}$ falls outside the interval $\theta\pm(1/4\mu)\norm{w_{>k}}_2$, which contains the interval $\theta\pm t\cdot\norm{w_I}_2$ (since $I\subseteq([n]\setminus[k])$ and $\mu=\Omega(1/t)$).

As before, we first analyze the case in which $z$ is chosen uniformly in $\pmset^n$. To do so we rely on a lemma of Servedio~\cite{ser07}, which asserts that the weights in $w$ decrease exponentially up to the critical index.  Intuitively, since the critical index is large (i.e., more than $k$), the exponential decay of the weights implies that $\norm{w_{>k}}_2$ is small. Thus, when uniformly choosing $z\in\pmset^n$, the sum $\ip{w_{[n]\setminus I},z_{[n]\setminus I}}$ is unlikely to fall in the small interval $\theta\pm(1/4\mu)\cdot\norm{w_{>k}}_2$; specifically, this happens with probability at most $\mu=p^{\Omega(1)}$ (see Claim~\ref{claim:rest:ltf:crit} for a precise statement). 

Since the event $\ip{w_{[n]\setminus I},z_{[n]\setminus I}}\in\theta\pm(1/4\mu)\cdot\norm{w_{>k}}_2$ happens with probability $p^{\Omega(1)}$ when $z\in\pmset^n$ is chosen uniformly, and the distribution ${\bf z}$ is $p^{\Omega(1)}$-pseudorandomly concentrated, the event also happens with probability at most $p^{\Omega(1)}$ over a choice of $z\sim{\bf z}$.

\paragraph{The full proof.} 
We will first prove an auxiliary lemma, which analyzes the effect of uniformly-chosen restrictions on regular LTFs (see Lemma~\ref{lem:rest:ltf:reg}). Then, we will prove a version of Proposition~\ref{prop:rest:ltf:informal} that only holds for LTFs with \emph{bounded critical index} (see Lemma~\ref{lem:rest:ltf:small}), and a version of Proposition~\ref{prop:rest:ltf:informal} that only holds for LTFs with \emph{large critical index} (see Lemma~\ref{lem:rest:ltf:large}). Finally, we will formally state a more general version of Proposition~\ref{prop:rest:ltf:informal} and prove it (see Proposition~\ref{prop:rest:ltf}).

The following auxiliary lemma considers a regular vector $w\in\R^m$, a fixed set of variables $I\subseteq[m]$ that will be kept alive, and a uniformly-chosen assignment $z\in\pmset^m$ for the fixed variables. The lemma will be used in the proof of Lemma~\ref{lem:rest:ltf:small}.

\begin{lemma} (pseudorandom restriction lemma for regular LTFs). \label{lem:rest:ltf:reg}
Let $m\in\N$, let $\mu\in(0,1)$, and let $\lambda\le3/4$. Let $w'\in\R^m$ be a $\mu$-regular vector, and let $I\subseteq[m]$ such that $\norm{w'_I}_2<\lambda\cdot\norm{w'}_2$. Then, for any $\theta'\in\R$ and $t>0$, the probability over uniform choice of $z\in\pmset^m$ that $\ip{w'_{[m]\setminus I},z_{[m]\setminus I}}\in\theta'\pm t\cdot\lambda\cdot\norm{w'}_2$ is at most $O(t\cdot\lambda+\mu)$.
\end{lemma}

\begin{proof}[{\bf Proof.}]
Note that $\norm{w'_{[m]\setminus I}}_2^2>\norm{w'}_2^2/4$; this is the case because $\norm{w'_I}_2^2<\lambda\cdot\norm{w'}_2^2\le\frac3{4}\cdot\norm{w'}_2^2$. It follows that $w'_{[m]\setminus I}$ is $2\mu$-regular, since for every $i\in[m]$ we have that $\abs{w'_i}\le\mu\cdot\norm{w'}_2\le2\mu\cdot\norm{w'_{[m]\setminus I}}_2$. It also follows that the interval $\theta\pm t\cdot\lambda\cdot\norm{w'}_2$ is contained in the interval $\theta\pm 2t\cdot\lambda\cdot\norm{w'_{[m]\setminus I}}_2$. By the Berry-Ess\'{e}en theorem (i.e., by Theorem~\ref{thm:be}), the probability over a uniform choice of $z\in\pmset^{m}$ that the sum $\ip{w_{[m]\setminus I},z_{[m]\setminus I}}$ falls in a fixed interval of length $2t\cdot\lambda\cdot\norm{w_{[m]\setminus I}}$ is at most $O(t\cdot\lambda+\mu)$. 
\end{proof}

The following lemma asserts that a suitably-chosen pseudorandom restriction turns every LTF with \emph{bounded critical index} to be very biased, with high probability. The specific parameters that are chosen for the lemma will be useful for us when proving the general case (i.e., Proposition~\ref{prop:rest:ltf}, which holds for arbitrary LTFs).

\begin{lemma} (pseudorandom restriction lemma for LTFs with small critical index). \label{lem:rest:ltf:small}
Let $n\in\N$, let $p\in[0,1]$ be a power of two, let $c\in\N$ be a constant, and let $t\le p^{-1/(3c-2)}$ and $\mu=1/4t^c$. Let ${\bf y}$ be a distribution over $\pmset^{\log(1/p)\cdot n}$ that is $p$-almost $O(\log(1/p))$-wise independent, and let ${\bf z}$ be a distribution over $\pmset^n$ that is $\mu$-pseudorandomly concentrated. Then, for any LTF $\Phi$ over $n$ input bits with $\mu$-critical index at most $k=10^3\cdot\mu^{-2}\cdot\log^2(1/\mu)$, the probability over choice of $\rho\sim({\bf y},{\bf z})$ that $\Phi\rest_\rho$ is $t$-balanced is at most $\tilde{O}(t^{1+c/2})\cdot\sqrt{p}+O(t^{-c})$.
\end{lemma}

\begin{proof}[{\bf Proof.}]
Let $\Phi=(w,\theta)$ be an LTF gate over $n$ input bits with critical index $h\le k$, and without loss of generality assume that $|w_1|\ge|w_2|\ge...\ge|w_n|$. Let $I\subseteq[n]$ be the random variable that is the set of live variables under ${\bf y}$; then, it holds that:

\begin{subclaim} \label{lem:rest:ltf:small:y}
With probability at least $1-O(\mu+p\cdot k)$ over $y\sim{\bf y}$ it holds that $I\subseteq([n]\setminus[h])$ and that $\norm{w_I}_2\le\sqrt{p/\mu}\cdot\norm{w_{[n]\setminus[h]}}_2$.
\end{subclaim}

\begin{proof}\innerqed
Since ${\bf y}$ is $p$-almost $O(\log(1/p))$-wise independent, each variable is kept alive with probability at most $2p$. Thus, the probability over $y\sim{\bf y}$ that the first $h$ variables are all fixed is at least $1-2p\cdot h$. Also, the expected value of $\norm{w_{I\cap([n]\setminus[h])}}_2^2$ is at most $2p\cdot\norm{w_{[n]\setminus[h]}}_2^2$, and hence with probability at least $1-2\mu$ it holds that $\norm{w_{I\cap([n]\setminus[h])}}_2\le\sqrt{p/\mu}\cdot\norm{w_{[n]\setminus[h]}}_2$. By a union-bound, with probability at least $1-O(\mu+p\cdot h)>1-O(\mu+p\cdot k)$ it holds that $I\subseteq([n]\setminus[h])$ and that $\norm{w_I}_2=\norm{w_{I\cap([n]\setminus[h])}}_2\le\sqrt{p/\mu}\cdot\norm{w_{[n]\setminus[h]}}_2$. 
\end{proof}

Fix any $y\sim{\bf y}$ such that the first $h$ variables are all fixed, and such that $\norm{w_{I}}_2\le\sqrt{p/\mu}\cdot\norm{w_{[n]\setminus[h]}}_2$. Our goal will be to prove that with high probability over $z\sim{\bf z}$ it holds that $\ip{w_{[n]\setminus I},{\bf z}_{[n]\setminus I}}\notin\theta\pm t\cdot\sqrt{p/\mu}\cdot\norm{w_{[n]\setminus[h]}}_2$; this suffices to prove the lemma, since $t\cdot\sqrt{p/\mu}\cdot\norm{w_{[n]\setminus[h]}}\ge t\cdot\norm{w_I}_2$. To do so, we first analyze the setting in which $z\in\pmset^n$ is chosen uniformly, rather than from the distribution ${\bf z}$:

\begin{subclaim} \label{lem:rest:ltf:small:z}
The probability over a uniform choice of $z\in\pmset^n$ that $\ip{w_{[n]\setminus I},z_{[n]\setminus I}}\in\theta\pm t\cdot\sqrt{p/\mu}\cdot\norm{w_{[n]\setminus[h]}}_2$ is at most $O(t\cdot\sqrt{p/\mu}+\mu)$.
\end{subclaim}

\begin{proof} \innerqed
The claim is trivial for $\mu\le2p$, so it suffices to prove the claim under the assumption that $\mu>2p$. Condition on any arbitrary assignment $z_{[h]}\in\pmset^h$ for the first $h$ variables, and note that %after fixing the first $h$ variables according to $z_{[h]}$, 
the vector $w_{>h}\in\pmset^{n-h}$ is $\mu$-regular (since $h$ is the $\mu$-critical index of $\Phi$). 

Let $T=[n]\setminus[h]$. Observe that when conditioning on $z_{[h]}$, the event $\ip{w_{[n]\setminus I},z_{[n]\setminus I}}\in\theta\pm t\cdot\sqrt{p/\mu}\cdot\norm{w_{[n]\setminus[h]}}_2$ happens if and only if the event $\ip{w_{T\setminus I},z_{T\setminus I}}\in\theta'\pm t\cdot\sqrt{p/\mu}\cdot\norm{w_{T}}_2$ happens, where $\theta'=\theta-\ip{w_{[h]},z_{[h]}}$. Since $w_T$ is $\mu$-regular, we can invoke Lemma~\ref{lem:rest:ltf:reg} with $w'=w_{T}$  and with $\lambda=\sqrt{p/\mu}\le3/4$ (the inequality is since $\mu>2p$), and deduce the probability of the event $\ip{w_{T\setminus I},z_{T\setminus I}}\in\theta'\pm t\cdot\sqrt{p/\mu}\cdot\norm{w_{T}}_2$ is at most $O(t\cdot\sqrt{p/\mu}+\mu)$.
\end{proof}

Since ${\bf z}$ is $\mu$-pseudorandomly concentrated, it follows from Claim~\ref{lem:rest:ltf:small:z} that the probability over $z\sim{\bf z}$ that $\ip{w_{[n]\setminus I},{\bf z}_{[n]\setminus I}}\in\theta\pm t\cdot\sqrt{p/\mu}\cdot\norm{w_{[n]\setminus[h]}}_2$ is at most $O(t\cdot\sqrt{p/\mu}+\mu)$. Thus, the probability over choice of $\rho\sim({\bf y},{\bf z})$ that $\Phi\rest_\rho$ is $t$-balanced is at most $O(t\cdot\sqrt{p/\mu}+\mu+p\cdot k)=\tilde{O}(t^{1+c/2})\cdot\sqrt{p}+O(t^{-c})$, where the last equality relied on the hypothesis that $t\le p^{-1/(3c-2)}$.
\end{proof}

The following lemma is similar to Lemma~\ref{lem:rest:ltf:small}, but holds for LTFs with \emph{large critical index}.

\begin{lemma} (pseudorandom restriction lemma for LTFs with large critical index). \label{lem:rest:ltf:large}
Let $n\in\N$, let $p\in[0,1]$ be a power of two, and let $\mu>0$. Let ${\bf y}$ be a distribution over $\pmset^{\log(1/p)\cdot n}$ that is $p$-almost $O(\log(1/p))$-wise independent, and let ${\bf z}$ be a distribution over $\pmset^n$ that is $\mu$-pseudorandomly concentrated. Then, for any LTF $\Phi$ over $n$ input bits with $\mu$-critical index larger than $k=10^3\cdot\mu^{-2}\cdot\log^{2}(1/\mu)$, the probability over choice of $\rho\sim({\bf y}, {\bf z})$ that $\Phi\rest_\rho$ is $(1/4\mu)$-balanced is $\tilde{O}(\mu^{-2})\cdot p+O(\mu)$.
\end{lemma}

\begin{proof}[{\bf Proof.}]
Let $\Phi=(w,\theta)$ be an LTF gate over $n$ input bits with $\mu$-critical index larger than $k$, and without loss of generality assume that $|w_1|\ge|w_2|\ge...\ge|w_n|$. Also, let $I\subseteq[n]$ be the random variable that is the set of live variables under ${\bf y}$. Note that the probability over $y\sim{\bf y}$ that $I\cap[k]\ne\emptyset$ is at most $2p\cdot k=\tilde{O}(\mu^{-2})\cdot p$ (since ${\bf y}$ keeps each variable alive with probability at most $2p$).%is $p$-almost $O(\log(1/p))$-wise independent, and by a union-bound). 

Condition on any arbitrary $y\sim{\bf y}$ such that $[k]\cap I=\emptyset$. Our goal now is to show that the probability over $z\sim{\bf z}$ that $\Phi\rest_\rho$ is $(1/4\mu)$-balanced is $O(\mu)$. We will actually prove a stronger claim: We will show that with probability at least $1-O(\mu)$ it holds that $\ip{w_{[n]\setminus I},{\bf z}_{[n]\setminus I}}\notin\theta\pm(1/4\mu)\cdot\norm{w_{>k}}_2$ (this claim is stronger, since $I\subseteq([n]\setminus [k])$, which implies that $\norm{w_{>k}}_2\ge\norm{w_I}_2$). To prove this assertion we will rely on the following claim, which is essentially from~\cite[Prop. 45]{css16} and generalizes~\cite[Lem. 5.8]{dgjsv10}. (Since the proof is sketched in~\cite{css16}, we include a full proof.)

\begin{subclaim} \label{claim:rest:ltf:crit}
Let $\mu>0$, let $r\in\N$, and let $k_{r,\mu}=\frac{4r\cdot\ln(3/\mu^2)}{\mu^2}$. Let $\Phi=(w,\theta)$ be an LTF over $n$ input bits with $\mu$-critical index larger than $k_{r,\mu}$ such that $|w_1|\ge...\ge|w_n|$, and let $J\subseteq[n]$ such that $J\supseteq[k_{r,\mu}]$. Then, the probability under uniform choice of $z\in\bitset^n$ that $\ip{w_{J},z_{J}}\in\theta\pm(1/4\mu)\cdot\norm{w_{>k_{r,\mu}}}_2$ is at most $2^{-r}$.
\end{subclaim}

\begin{proof} \innerqed
Since the critical index of $\Phi$ is larger than $k_{r,\mu}$, a lemma of Servedio~\cite[Lem. 3]{ser07} asserts that for any $1\le i<j\le k_{r,\mu}$ it holds that
\mm{
|w_j|\le \norm{w_{\ge j}}_2\le \left(1-\mu^2\right)^{(j-i)/2}\cdot\norm{w_{\ge i}}_2\le \left(1-\mu^2\right)^{(j-i)/2}\cdot|w_i|/\mu \eqtag{eq:rocco} \mathdot
}
(For an equivalent statement of the lemma see~\cite[Lem. 5.5]{dgjsv10}.) In particular, fixing $\gamma=\frac{2\ln(3/\mu^2)}{\mu^2}$, for any $i\in\N$ such that $i\cdot\gamma<k_{r,\mu}$ it holds that $|w_{i\cdot\gamma}|<|w_1|/3^i$. 

% On the weird choice of $\gamma$: We need $2/\mu^2$ on the outside, and $\ln( z )$ yields a division by z. I chose $z=3/\mu^2$; one $\mu$ is to kill the $\mu$ in Eq.~\eqref{eq:rocco}, the 3 is for the current claim about exponential decay, and another $\mu$ (i.e., $\mu^2$ instead of $\mu$) is to use in the end of the current sub-proof.

Let $R=1,\gamma,...,r\cdot\gamma<k_{r,\mu}$, and consider any arbitrary fixed value of $z_{J\setminus R}$. Then, by a claim of Diakonikolas \etal~\cite[Clm. 5.7]{dgjsv10}, there exists at most a single value $z_R\in\pmset^r$ such that $\ip{w_{R},z_{R}}\in\left(\theta-\ip{w_{J\setminus R},z_{J\setminus R}}\right)\pm|w_{r\cdot\gamma}|/4$. Thus, the probability under a uniform choice of $z\in\bitset^n$ that $\ip{w_{J},z_{J}}\in\theta\pm|w_{r\cdot\gamma}|/4$ is at most $2^{-r}$. 

The claim follows since $\norm{w_{>k_{r,\mu}}}_2\le\norm{w_{\ge (r+1)\cdot\gamma}}_2\le\mu\cdot|w_{r\cdot\gamma}|$, where the first inequality is since $k_{r,\mu}>(r+1)\cdot\gamma$ and the second inequality is due to Eq.~\eqref{eq:rocco}.
\end{proof}

We invoke Claim~\ref{claim:rest:ltf:crit} with the value $r=\log(1/\mu)$ and with the set $J=[n]\setminus I$, while noting that the critical index of $\Phi$ is indeed larger than $k\ge k_{r,\mu}$. Since the interval $\theta\pm(1/4\mu)\cdot\norm{w_{>k}}_2$ is contained in the interval $\theta\pm(1/4\mu)\cdot\norm{w_{>k_{r,\mu}}}_2$ (because $k\ge k_{r,\mu}$), we deduce that the event $\ip{w_{[n]\setminus I},z_{[n]\setminus I}}\in\theta\pm (1/4\mu)\cdot\norm{w_{>k}}_2$ happens with probability at most $\mu$ under a uniform choice of $z\in\bitset^n$. Since ${\bf z}$ is $\mu$-pseudorandomly concentrated, this event happens with probability at most $O(\mu)$ also under a choice of $z\sim{\bf z}$.
\end{proof}

Finally, we are ready to state a more general version of Proposition~\ref{prop:rest:ltf:informal} and to prove it. The proof will rely on Lemmas~\ref{lem:rest:ltf:small} and~\ref{lem:rest:ltf:large}.

\begin{proposition} (pseudorandom restriction lemma for an arbitrary LTF). \label{prop:rest:ltf}
Let $n\in\N$, let $p\in[0,1]$ be a power of two, let $c\in\N$ be a constant, and let $t\le p^{-1/(3c-2)}$. Let ${\bf y}$ be a distribution over $\pmset^{\log(1/p)\cdot n}$ that is $p$-almost $O(\log(1/p))$-wise independent, and let ${\bf z}$ be a distribution over $\pmset^n$ that is $(1/4t^c)$-pseudorandomly concentrated. Then, for any LTF $\Phi$ over $n$ input bits, the probability over choice of $\rho\sim({\bf y},{\bf z})$ that $\Phi\rest_\rho$ is $t$-balanced is at most $\tilde{O}(t^{1+c/2})\cdot\sqrt{p}+O(t^{-c})$.
\end{proposition}

To obtain the parameters that were stated in Section~\ref{sec:tech:alg}, invoke Proposition~\ref{prop:rest:ltf} with $c=2$. (When $c=2$, the hypothesis that $t\le p^{-1/(3c-2)}=p^{-1/4}$ is not required, since for $t>p^{-1/4}$ the probability bound in the lemma's statement is trivial.)

\begin{proof}[{\bf Proof of Proposition~\ref{prop:rest:ltf}.}]
%The claim is trivial for $t\ge p^{-1/(2+c)}$, so let us assume without loss of generality that $t<p^{-1/(2+c)}$.
Let $\Phi=(w,\theta)$ be an LTF gate over $n$ input bits, let $\mu=1/4t^c$, and let $k=10^3\cdot\mu^{-2}\cdot\log^2(1/\mu)$. If the $\mu$-critical index of $\Phi$ is at most $k$, the asserted probability bound follows immediately from Lemma~\ref{lem:rest:ltf:small}. On the other hand, if the $\mu$-critical index of $\Phi$ is larger than $k$, we can rely on Lemma~\ref{lem:rest:ltf:large}. The lemma asserts that the probability that $\Phi\rest_\rho$ is $(1/4\mu)$-balanced is at most $\tilde{O}(\mu^{-2})\cdot p+O(\mu)<\tilde{O}(t^{1+c/2})\cdot\sqrt{p}+O(t^{-c})$, where the inequality relies on the hypothesis that $t\le p^{-1/(3c-2)}$. Since $(1/4\mu)\ge t$, whenever $\Phi\rest_\rho$ is $(1/4\mu)$-imbalanced it is also $t$-imbalanced. 
\end{proof}

\subsubsection{Pseudorandom restriction algorithm for a ``layer'' of LTFs} \label{sec:main:rest:layer}

The next step is to construct a pseudorandom restriction algorithm that transforms a depth-$d$ linear threshold circuit into a depth-$(d-1)$ linear threshold circuit. The key part in this step is an application of Proposition~\ref{prop:rest:ltf}.

\begin{proposition} (pseudorandom restriction algorithm for a ``layer'' of LTFs). \label{prop:depth}
For every three constants $d\ge2$ and $\e>0$ and $c>0$, there exists a polynomial-time algorithm that gets as input a circuit $C\in\mathcal{C}_{n,d,n^{1+\e}}$ and a random seed of length $O(\log(n)\cdot(\log\log(n))^2)$, and with probability at least $1-n^{-\e}$ outputs the following:
\begin{enumerate}
	\item A restriction $\rho\in\pmstar^n$ that keeps at least $n'=\Omega(n^{1-24\cdot\e})$ variables alive.
	\item A circuit $\widetilde{C}\in\mathcal{C}_{n',d-1,(n')^{1+30\e}}$ that agrees with $C$ on at least $1-n^{-c}$ of the inputs in the subcube that corresponds to $\rho$ (i.e., $\Pr_{x\in\pmset^{|\rho^{-1}(\star)|}}[C\rest_\rho(x)=\widetilde{C}(x)]>1-n^{-c}$).
\end{enumerate}
\end{proposition}

\paragraph{High-level overview of the proof.}

The key step of the algorithm is to apply Proposition~\ref{prop:rest:ltf} with parameters $p=n^{-\beta}$ and $c=1$ and $t=p^{-1/5}$, where $\beta=O(\e)$. The lemma asserts that, in expectation, all but approximately $n^{-\beta/5}$ of the gates will become $t$-imbalanced (for simplicity, ignore polylogarithmic factors for now). Such imbalanced gates are extremely close to a constant function, so we can replace the gates by the corresponding constants and get a circuit that agrees with the original circuit on almost all inputs. 

As for the other $n^{-\beta/5}$-fraction of the gates, we expect that the number of wires feeding into them will decrease by a factor of $p$ after the restriction. Specifically, assume that indeed the fan-in of each gate decreased by a factor of at least $p$; then, the expected number of wires feeding into the balanced gates after the restriction is at most
\mm{
\sum_{\Phi\text{ gate}}\Pr[\Phi\text{ balanced}]\cdot p\cdot(\text{\# wires incoming to }\Phi) &\le n^{-\beta/5}\cdot p\cdot n^{1+\e}\mathdot \eqtag{eq:fanin:highlevel}
}
Thus, with probability at least $1-n^{-\beta/10}$, the number of wires feeding into balanced gates is at most $(n^{\e-\beta/10})\cdot p\cdot n$, which is much smaller than the expected number of living variables (i.e., than $p\cdot n$) if $\beta>10\e$. When this happens, we can afford to simply fix all the variables that feed into balanced gates, making those gates constant too.

The argument above relied on the assumption that the fan-in of each gate $\Phi$ decreased by a factor of at least $p$. We can argue that this indeed holds with high probability for all gates with fan-in at least $n^{\al}$, where $\al>\beta$, but we will need to separately handle gates with fan-in at most $n^{\al}$. This will be done in two steps: The first is an initial preprocessing step (before applying Proposition~\ref{prop:rest:ltf}), in which we fix every variable with fan-out more than $2\cdot n^{\e}$; since there are at most $n^{1+\e}$ wires, this step fixes at most $n/2$ variables. Then, after applying Proposition~\ref{prop:rest:ltf} and fixing the variables that feed into balanced gates with fan-in at least $n^{\al}$, we show that there exists a set $I$ of variables of size approximately $n^{-(\al+\e)}\cdot(p\cdot n)$ such that after fixing all variables outside $I$, each gate with fan-in at most $n^{\al}$ has fan-in at most one (see Claim~\ref{claim:graphtharg}). Thus, we can fix the variables outside $I$, and then replace each gate with fan-in at most $n^{\al}$ with the corresponding variable (or with its negation). At this point all the gates in the bottom layer have been replaced by constants or by variables.

\begin{proof}[{\bf Proof of Proposition~\ref{prop:depth}.}]
Let $G=\{\Phi_1,...,\Phi_r\}$ be the set of gates in the bottom layer of $C$. For $\al=12\e$, let $S\subseteq G$ be the set of gates with fan-in at most $n^{\al}$, and let $L=G\setminus S$ be the set of gates with fan-in more than $n^{\al}$. 

The restriction $\rho$ will be composed of four restrictions $\rho_1,...,\rho_4$. When describing the construction of each restriction, we will always assume that all previous restrictions were successful (we will describe exactly what ``successful'' means for each restriction). Also, after each restriction, we fix additional variables if necessary, in order to obtain an exact number of living variables in the end of the step.

Let ${\bf z}$ be a distribution over $\pmset^n$ that is $(1/q(n))$-pseudorandom for LTFs, where $q$ is a sufficiently large polynomial. We mention in advance that for each $i\in[4]$, the values for variables that are fixed by $\rho_i$ will always be decided by sampling from ${\bf z}$.

\paragraph{The first restriction $\rho_1$: Reduce the fan-out of input gates.}
We sample $z\sim{\bf z}$, and fix all variables with fan-out more than $2\cdot n^{\e}$ to values according to $z$. Since the number of wires between the bottom-layer gates and the input variables is at most $n^{1+\e}$, and each fixing of a variable eliminates $2\cdot n^{\e}$ wires, we will fix no more than $n/2$ variables in this step. Let $n_1=n/2$ be the number of living variables after the first step.%~\footnote{The restriction $\rho_1$ facilitates the application of the restriction $\rho_4$, and is applied now (rather than immediately before applying $\rho_4$) for technical reasons.}

\paragraph{The second restriction $\rho_2$: Applying Proposition~\ref{prop:rest:ltf}.}
We use Proposition~\ref{prop:rest:ltf} with the values $p=n^{-\beta}$, where $\beta=11\e$, and $c=1$, and $t=p^{-1/5}$.~\footnote{For simplicity, we assume that $p=n^{-11\e}$ is a power of two. Otherwise, we can choose $\beta$ to be a value very close to $11\e$ such that $p$ will be a power of two, with no meaningful change to the rest of the proof (the proof only relies on the fact that $10\e<\beta<\al$).}  The distributions that we use are a ($1/\poly(n)$)-almost $O(\log(1/p))$-wise independent distribution ${\bf y}$ over $\pmset^{\log(1/p)\cdot n}$ and the aforementioned distribution ${\bf z}$ over $\pmset^n$.

Let $\mathcal{E}$ be the event in which $\rho_2$ keeps at least $(p\cdot n_1)/2$ variables alive, and for every gate $\Phi\in L$ it holds that $\fanin(\Phi\rest_{\rho_2})\le 2p\cdot\fanin(\Phi)$. We claim that $\mathcal{E}$ happens with probability at least $1-1/\poly(n)$. To see that this is the case, note that the expected number of living variables is $p\cdot n_1=n^{\Omega(1)}$, and that for each gate $\Phi\in G$, the expected fan-in of $\Phi\rest_{\rho_2}$ is $n^{\al-\beta}=n^{\Omega(1)}$. Since the choice of variables to keep alive is $\frac1{\poly(n)}$-almost $O(1)$-independent, we can use Fact~\ref{fact:tail} to deduce that $\Pr[\mathcal{E}]\ge1-\frac1{\poly(n)}$.

Now, assume without loss of generality that $L=\{\Phi_1,...,\Phi_{r'}\}$, for some $r'\le r$. For any $i\in[r']$, denote by $\mathcal{B}_i$ the event that $\Phi_i$ is $t$-balanced. Note that when conditioning on $\mathcal{E}$, the probability of each $\mathcal{B}_i$ is at most $\tilde{O}(n^{-\beta/5})$. Therefore, conditioned on $\mathcal{E}$, the expected number of wires feeding into $t$-balanced gates in $L$ after the restriction is 
\mm{
\E\left[ \sum_{i\in[r']} \mathbf{1}_{\mathcal{B}_i}\cdot\fanin(\Phi_i\rest_{\rho_2}) \Big\vert \mathcal{E}\right] &= 
\sum_{i\in[r']}\Pr[\mathcal{B}_i|\mathcal{E}]\cdot\E[\fanin(\Phi_i\rest_{\rho_2})|\mathcal{E},\mathcal{B}_i] \\
&\le \sum_{i\in[r']}\tilde{O}(n^{-\beta/5})\cdot(2p\cdot\fanin(\Phi_i)) \\
&= \tilde{O}(n^{-\beta/5})\cdot p\cdot n^{1+\e} \mathdot 
}

Hence, conditioned on $\mathcal{E}$, the probability that the number of wires feeding into $t$-balanced gates in $L$ after the restriction is more than $\tilde{O}(n^{-\beta/10})\cdot p\cdot n^{1+\e}=\tilde{O}(n^{\e-\beta/10})\cdot n^{1-\beta}$ is at most $O(n^{-\beta/10})$. We consider the restriction $\rho_2$ successful if $\mathcal{E}$ happens and if the number of wires between $t$-balanced gates in $L$ and input gates is at most $\tilde{O}(n^{\e-\beta/10})\cdot n^{1-\beta}$. In this case, the number of currently-living variables is $n_2 = p\cdot n_1/2 = \frac1{4}\cdot n^{1-\beta}$. 

After applying $\rho_2$, we replace any $t$-imbalanced gate $\Phi_i\in L$ with its most probable value $\sigma_i\in\pmset$. Note that by Theorem~\ref{thm:hoef}, each $t$-imbalanced gate $\Phi_i$ is $(\exp(-n^{\Omega(1)}))$-close to $\sigma_i$ in the subcube that corresponds to the currently-living variables.

\paragraph{The third restriction $\rho_3$: Eliminate $L$-gates that remained unbiased.}
In this step we sample $z\sim{\bf z}$ again, and fix all the variables that feed into $t$-balanced gates according to $z$. Assuming that $\rho_2$ was successful, the number of such variables is at most $\tilde{O}(n^{\e-\beta/10})\cdot n^{1-\beta}=o(n_2)$, where we used the fact that $\beta>10\e$. Denote the restriction applied in this step by $\rho_3$, and note that the number of living variables after applying $\rho_3$ is $n_3=\Omega(n_2)=\Omega(n^{1-11\e})$. 

Our goal now is to claim that for each gate $\Phi_i$ that was replaced by a constant $\sigma\in\pmset$ prior to applying $\rho_3$, it still holds that $\Phi_i$ is close to $\sigma$ in the subcube $\pmset^{\rho_3^{-1}(\star)}$. To do so we will rely on a lemma that asserts the following: If an LTF $\Phi_i$ is $\de$-close to a constant function, then with probability $1-\gamma$ over choice of $z\sim{\bf z}$ it holds that $\Phi_i\rest_\rho$ is $\de'$-close to the same constant function, as long as $\de\le\poly(\de',\gamma)$ and that ${\bf z}$ is $\poly(\gamma)$-pseudorandom for LTFs.

\begin{lemma} (bias preservation lemma). \label{lem:bias}
Let $n\in\N$, and let $\de,\de',\gamma>0$ such that $\de\le(\gamma\cdot\de')^{10}$. Let $\Phi=(w,\theta)$ be an LTF over $n$ input bits that is $\de$-close to a constant function $\sigma\in\pmset$, let $I\subseteq[n]$, and let ${\bf z}$ be a distribution over $\pmset^{[n]\setminus I}$ that is $(\de'\cdot\gamma^2)$-pseudorandom for LTFs. Then, with probability $1-O(\gamma)$ over choice of $z\sim{\bf z}$ it holds that $\Phi\rest_{(I,z)}$ is $\de'$-close to $\sigma$.
\end{lemma}

The proof of Lemma~\ref{lem:bias} is deferred to Section~\ref{sec:bias}. We invoke Lemma~\ref{lem:bias} with $I$ being the set of variables that are kept alive by $\rho_3$, and $\de=\exp(-n^{\Omega(1)})$, and $\gamma=1/\poly(n)$, and $\de'=n^{-10\cdot(2+4\e+c)}$. After union-bounding over at most $r\le n^{1+\e}$ gates that were replaced by constants, with probability $1-1/\poly(n)$ it holds that all these gates are $\de'$-close to constants in the subcube $\pmset^{\rho_3^{-1}(\star)}$.

\paragraph{The fourth restriction $\rho_4$: Eliminate gates with small fan-in.}
We will rely on the following claim, which is an algorithmic version of~\cite[Prop. 36]{css16}:

\begin{subclaim} \label{claim:graphtharg}
For $k'=2\cdot n^{\al+\e}$, we can deterministically find in $\poly(n)$ time a set $I$ of at least $n_3/k'$ living variables such that when fixing all variables not in $I$ to any arbitrary values, the fan-in of each gate in $S$ is at most one. 
\end{subclaim}

\begin{proof} \innerqed
Consider the graph in which the vertices are the input gates $x_1,...,x_{n_3}$, and two vertices $x_i$ and $x_j$ are connected (in the graph) if and only if there exists a gate $\Phi_i\in S$ that is connected (in the circuit) to both $x_i$ and $x_j$. Note that this graph has degree at most $k'$, since every living variable has fan-out at most $2\cdot n^{\e}$, and every gate in $S$ has fan-in at most $n^{\al}$. Therefore, we can greedily construct an independent set $I$ in the graph of size at least $n_3/k'$, which is indeed the set of variables that we wanted.
\end{proof}

The algorithm finds a set $I$ using Claim~\ref{claim:graphtharg}, samples $z\sim{\bf z}$, and fixes all the variables outside $I$ according to $z$. This yields a restriction that reduces the fan-in of each gate in $S$ to one. Thus, each gate $\Phi\in S$ now simply takes the value of an input gate (or its negation), which implies that the gates that are connected to $\Phi$ (in the layer above it) can be connected immediately to the corresponding input gate, and we can remove $\Phi$ from the circuit. The number of living variables is $n_4=n_3/k'=\Omega(n^{1-24\e})$. 

To conclude, we claim that the gates that were previously replaced by constants are still close to constants in the new subcube. This is done by invoking Lemma~\ref{lem:bias} with $I$ being the aforementioned set of size $n_4$, and with parameter values $\de=n^{-10\cdot(2+4\e+c)}$, and $\gamma=n^{-(1+3\e)}$, and $\de'=n^{-(c+1+\e)}$. After union-bounding over the gates that were replaced by constants, with probability at least $1-n^{-2\e}$ it holds that all these gates are $\de'$-close to constants in the final subcube. It follows that the original circuit is $\de''$-close to the new circuit in the final subcube, where $\de''\le\de'\cdot n^{1+\e}\le n^{-c}$.

\paragraph{Accounting for the parameters.}
We obtained a circuit in $\tilde{C}\in\mathcal{C}_{n_4,d-1,n^{1+\e}}$. Since $n^{1+\e}=O(n_4^{\frac{1+\e}{1-24\e}})<n_4^{(1+\e)(1+25\e)}\le n_4^{1+30\e}$, we have that $\tilde{C}\in\mathcal{C}_{n_4,d-1,n_4^{1+30\e}}$. To sample the restriction $\rho=\rho_4\circ...\circ\rho_1$, we sampled from the distribution ${\bf z}$ four times, and from the distribution ${\bf y}$ a single time. A sample from ${\bf y}$ can obtained with seed length $O(\log(n))$, and relying on Theorem~\ref{thm:prg:ltfs}, each sample from ${\bf z}$ can be obtained with seed length $O(\log(n)\cdot(\log\log(n))^2)$.

Finally, let us account for the error probability. The first step is deterministic and always succeeds. In the second step, the algorithm is unable to simplify the circuit if the event $\mathcal{E}$ does not happen, or if the number of wires between $t$-balanced gates in $L$ and input gates is too large. Denoting the latter event by $\mathcal{E}'$, the probability of error is at most $\Pr[\lnot\mathcal{E}]+\Pr[\mathcal{E}'|\mathcal{E}]\le O(n^{-\beta/10})$. The last type of error to account for is the probability that $\tilde{C}$ is not $n^{-c}$-close to $C$ in $\pmset^{\rho^{-1}(\star)}$; as detailed above, this happens with probability at most $n^{-2\e}$. The overall error is thus $O( n^{-\beta/10} + n^{-2\e} )<n^{-\e}$.
\end{proof}

\subsubsection{Pseudorandom restriction algorithm for linear threshold circuits} \label{sec:main:rest:circuit}

We are now ready to construct the pseudorandom restriction algorithm that simplifies any linear threshold circuit to a single LTF gate (i.e., Proposition~\ref{prop:main:rest:circuit}). The proof will consist of $d-1$ applications of Proposition~\ref{prop:depth}. In each application, we will use Lemma~\ref{lem:bias} to claim that all the approximations in previous applications of Proposition~\ref{prop:depth} still hold.

\begin{proposition} (Proposition~\ref{prop:main:rest:circuit}, restated). \label{prop:itdepth:full}
Let $d\ge1$, let $\e>0$ be a sufficiently small constant, and let $\de=d\cdot30^{d-1}\cdot\e$. Then, there exists a polynomial-time algorithm that for every $n\in\N$, when given as input a circuit $C\in\mathcal{C}_{n,d,n^{1+\e}}$ and a random seed of length $O(\log(n)\cdot(\log\log(n))^2)$, with probability at least $1-n^{-\e/2}$ satisfies the following:
\begin{enumerate}
	\item The algorithm outputs a restriction $\rho\in\pmstar^n$ that keeps at least $n^{1-\de}$ variables alive.
	\item The algorithm outputs an LTF $\Phi:\pmset^{\rho^{-1}(\star)}\ra\pmset$ such that $\Phi$ is $1/10$-close to $C\rest_\rho$ (i.e., $\Pr_{x\in\pmset^{\rho^{-1}(\star)}}[C(x)=\Phi(x)]\ge9/10$).
\end{enumerate}
\end{proposition}

\begin{proof}[{\bf Proof.}]
We repeatedly invoke Proposition~\ref{prop:depth}, for $d-1$ times. For $i\in[d-1]$, let $\rho^{(i)}$ be the restriction that is obtained in the $i^{th}$ invocation of Proposition~\ref{prop:depth}, and let $\rho=\rho^{(d-1)}\circ...\circ\rho^{(1)}$ be the final restriction. Let $C_0=C$, and for $i\in[d-1]$, let $C_{i}$ be the circuit that is obtained after the $i^{th}$ invocation of Proposition~\ref{prop:depth}. Also let $\e_0=\e$ and $\e_{i}=30\cdot\e_{i-1}=30^{i}\cdot\e$, and let $n_0=n$ and $n_{i}=\Omega\left( (n_{i-1})^{1-24\e_{i-1}} \right)$. 

We say that an invocation of Proposition~\ref{prop:depth} is \emph{successful} if the two items in the proposition's statement are satisfied (i.e., the algorithm outputs a restriction that keeps sufficiently many live variables, and a circuit of smaller depth that agrees with the original circuit on almost all inputs). Assuming all invocations of Proposition~\ref{prop:depth} are successful, for each $i\in[d-1]$ it holds that $C_{i}\in\mathcal{C}_{n_{i},d-i,n_i^{1+\e_{i}}}$, and in particular $C_{d-1}$ is a single LTF $\Phi$. Also, in this case, the number of living variables after all invocations is
\mm{
n_{d-1} = n^{\Pi_{i=0}^{d-2}(1-24\e_{i})} > n^{1-24\cdot\sum_{i=0}^{d-2}\e_{i}} > n^{1-24\cdot d\cdot\e_{d-2}} > n^{1-\de} \mathdot \eqtag{eq:livingvars}
}

The required seed length for the $d-1$ invocations of Proposition~\ref{prop:depth} is $\tilde{O}(\log(n))$. To bound the probability of error, for each $i\in[d-1]$, assume that all previous $i-1$ invocations were successful, and note that the probability that the $i^{th}$ invocation of Proposition~\ref{prop:depth} fails is at most $n_{i-1}^{-\e_{i-1}}<(n^{1-\de})^{-\e}$ (the inequality is since we assumed that the previous invocations of Proposition~\ref{prop:depth} were successful, which implies that $n_{i-1}\ge n^{1-\de}$, by a calculation similar to Eq.~\eqref{eq:livingvars}). Thus, the accumulated probability of error is at most $d\cdot(n^{1-\de})^{-\e}<n^{-\e/2}$, where the inequality relied on the fact that $\e$ is sufficiently small. 

Condition on all the $d-1$ invocations of Proposition~\ref{prop:depth} being successful. Recall that in this case, for every $i\in[d-1]$ it holds that $C_{i}$ is $n^{-c}$-close to $C_{i-1}\rest_{\rho^{(i)}}$; we now claim that, with high probability, this approximation continues to hold even in the subcube that corresponds to the final restriction $\rho$.

\begin{subclaim}
For every $i\in[d-1]$, with probability $1-1/\poly(n)$ it holds that $\left(C_{i-1}\right)\rest_\rho$ is $1/10d$-close to $\left(C_{i}\right)\rest_\rho$.
\end{subclaim}
\begin{proof}\innerqed
For each $j\in\{i,...,d-1\}$, recall that $\rho^{(j)}$ is the composition of four restrictions, denoted by $\rho^{(j)}_1,...,\rho^{(j)}_4$. Fix $i\in[d-1]$, condition on any fixed choice for $\rho^{(i)}_1$ and $\rho^{(i)}_2$, and let $C'=(C_{i-1})\rest_{\rho^{(i)}_1,\rho^{(i)}_2}$. Recall that immediately after applying $\rho^{(i)}_2$, the algorithm from Proposition~\ref{prop:depth} replaces a set of $m\le n^{1+\e_{d-(i-1)}}$ LTF gates, denoted $\Phi_1,...,\Phi_{m}$, with a corresponding set of constants $\sigma_1,...,\sigma_{m}\in\pmset$. Let $\widetilde{C'}$ be the circuit that is obtained from $C'$ by the aforementioned replacement. Finally, note that for every choice of final restriction $\rho$ it holds that $\left(C_{i-1}\right)\rest_\rho=C'\rest_\rho$ and $\left(C_{i}\right)\rest_\rho=\widetilde{C'}\rest_\rho$. 

Our goal now will be to show that for every fixed $k\in[m]$, with probability $1-1/\poly(n)$ over choice of $\rho$ it holds that $(\Phi_k)\rest_\rho$ is $1/(10dm)$-close to $\sigma_k$. This suffices to conclude the proof, since it follows (by a union-bound over the $m$ gates) that with probability $1-1/\poly(n)$, for every $k\in[m]$ it holds that $(\Phi_k)\rest_\rho$ is $1/(10dm)$-close to $\sigma_k$; and whenever the latter event happens we have that $C'\rest_\rho$ is $1/(10d)$-close to $\widetilde{C'}\rest_\rho$.

Towards the aforementioned goal, fix $k\in[m]$, and recall that $\Phi_k$ is $\de_0$-close to some constant function $\sigma_k\in\pmset$, where $\de_0=\exp\left( n_{i-1}^{-\Omega(1)} \right)=\exp\left( n^{-\Omega(1)} \right)$, where the inequality is since $n_{i-1}=n^{\Omega(1)}$ (recall that we conditioned on all invocations of Proposition~\ref{prop:depth} being successful). Observe that the final restriction $\rho$ is composed of $t\eqdef4\cdot(d-i-1)+2$ additional restrictions on the domain of $\Phi_k$: Two additional restrictions $\rho^{(i)}_3$ and $\rho^{(i)}_4$ in the $i^{th}$ invocation of Proposition~\ref{prop:depth}, and for each $j\in\{i+1,...,d-1\}$, four restrictions $\rho^{(j)}_1,...,\rho^{(j)}_4$ in the $j^{th}$ invocation of Proposition~\ref{prop:depth}. Recall that each of the $t$ restrictions is chosen by first choosing (deterministically or pseudorandomly) a set of variables to keep alive, and then \emph{independently} choosing values for the fixed variables. Therefore, we will now repeatedly use Lemma~\ref{lem:bias}, to claim that each restriction preserves the closeness of $\Phi_k$ to $\sigma_k$.

For convenience, rename the $t$ restrictions $\rho^{(i)}_3,\rho^{(i)}_4,\rho^{(i+1)}_1, ... ,\rho^{(i+1)}_4,...,\rho^{(d-1)}_1,...,\rho^{(d-1)}_4$, and denote them by $\rho'^{(1)},...,\rho'^{(t)}$. Let $\gamma=n^{-c}$ for a sufficiently large constant $c>1$. Note that $\de_{0}<n^{-10^{2t}\cdot c}$, and for every $r\in[t]$ let $\de_{r}=\de_{r-1}^{1/10^2}$; it follows that for every $r\in[t]$ it holds that $\de_{r-1}\le(\gamma\cdot\de_{r})^{10}$. We prove by induction on $r\in[t]$ that with probability at least $1-O(n^{-c})$ it holds that $(\Phi_k)\rest_{\rho'^{(1)}\circ...\circ\rho'^{(r)}}$ is $\de_r$-close to $\sigma_k$. For the base case $r=1$ we rely on the hypothesis that $\Phi_k$ is $\de_0$-close to $\sigma_k$, and use Lemma~\ref{lem:bias} with the values $\de=\de_{0}$ and $\de'=\de_{1}$ and $\gamma=n^{-c}$ as above. For the induction step $r>1$, we condition on $(\Phi_k)\rest_{\rho'^{(1)}\circ...\circ\rho'^{(r-1)}}$ being $\de_{r-1}$-close to $\sigma_k$, and again use Lemma~\ref{lem:bias} with the values $\de=\de_{r-1}$ and $\de'=\de_{r}$ and $\gamma=n^{-c}$. Hence, with probability at least $1-O(n^{-c})$ it holds that $\left(\Phi_k\right)\rest_\rho$ is $\de_{t}$-close to $\sigma_k$, where $\de_t=n^{-c}<1/(10dm)$.
\end{proof}

Thus, with probability $1-1/\poly(n)$, for every $i\in[d-1]$ it holds that $(C_{i-1})\rest_\rho$ is $1/10d$-close to $(C_{i})\rest_\rho$. Whenever this holds, by a union-bound it follows that $C\rest_\rho=(C_0)\rest_\rho$ is $1/10$-close to $(C_{d-1})\rest_\rho=C_{d-1}=\Phi$.
\end{proof}

\subsection{Proof of the bias preservation lemma} \label{sec:bias}

In this section we prove Lemma~\ref{lem:bias}. Loosely speaking, the lemma  asserts that an LTF $\Phi$ that is close to a constant $\sigma\in\pmset$ remains close to $\sigma$ when the domain is restricted by a restriction $\rho$ in which the values for the fixed variables are chosen from a distribution that is pseudorandom for LTFs. For the proof we will need the following lemma from~\cite[Lem. 15]{tell17} (the original notations are adapted for the current context).

\begin{lemma} (randomized tests). \label{lem:randtest}
Let $n\in\N$, and let $\e_1,\e_2,\e_3,\e_4,\e_5>0$ be error parameters.
\begin{itemize}
	\item Let $G\subseteq\pmset^n$, and let $E\subseteq G$ such that $\Pr_{z\in\pmset^n}[z\in E]\ge1-\e_1$.
	\item Let ${\bf T}$ be a distribution over functions $T:\pmset^n\ra\pmset$ such that for every $z\in E$ it holds that $\Pr_{T\sim{\bf T}}[T(z)=-1]\ge1-\e_2$, and for every $z\notin G$ it holds that $\Pr_{T\sim{\bf T}}[T(z)=1]\ge1-\e_3$.
	\item Let ${\bf z}$ be a distribution that is $\e_5$-pseudorandom for all but an $\e_4$-fraction of the tests in ${\bf T}$; that is, the probability over $T\sim{\bf T}$ that $\abs{\Pr[T({\bf u}_n)=-1]-\Pr[T({\bf z})=-1]}>\e_5$ is at most $\e_4$.
\end{itemize}
Then, the probability that ${\bf z}\in G$ is at least $1-(\e_1+\e_2+\e_3+2\e_4+\e_5)$.
\end{lemma}

Fix a set $I\subseteq[n]$ of variables that the restriction keeps alive. Relying on Lemma~\ref{lem:randtest}, the proof idea for Lemma~\ref{lem:bias} is to design a distribution ${\bf T}$ over tests that gets as input $z\in\pmset^{[n]\setminus I}$, and tests whether or not $\Phi$ is close to $\sigma$ in the subcube corresponding to the restriction $\rho=\rho_{I,z}$. 
%This randomized test ${\bf T}$ will simply be the result of sampling random points inside $\rho$ and checking whether or not $\Phi$ evaluates to $\sigma$ on these points. Relying on Lemma~\ref{lem:randtest}, any distribution over restrictions that is pseudorandom for the deterministic tests $T\in\mathtt{supp}({\bf T})$ yields a restriction $\rho$ such that $\Phi$ is indeed close to $\sigma$ in the subcube that corresponds to $\rho$.

\begin{lemma} (Lemma~\ref{lem:bias}, restated). \label{lem:bias:full}
Let $n\in\N$, and let $\de,\de',\gamma>0$ such that $\de\le(\gamma\cdot\de')^{10}$. Let $\Phi=(w,\theta)$ be an LTF over $n$ input bits that is $\de$-close to a constant function $\sigma\in\pmset$, let $I\subseteq[n]$, and let ${\bf z}$ be a distribution over $\pmset^{[n]\setminus I}$ that is $(\de'\cdot\gamma^2)$-pseudorandom for LTFs. Then, with probability $1-O(\gamma)$ over choice of $z\sim{\bf z}$ it holds that $\Phi\rest_{(I,z)}$ is $\de'$-close to $\sigma$.
\end{lemma}

\paragraph{A high-level description of the proof.} For every $z\in\pmset^{[n]\setminus I}$, consider the corresponding subcube $\mathfrak{C}_z=\left\{y\in\pmset^n:\forall i\in([n]\setminus I),y_i=z_i\right\}$. Our goal is to show that with high probability over $z\sim{\bf z}$ it holds that $\Phi$ is close to $\sigma$ in $\mathfrak{C}_z$. To do so, we will construct a distribution ${\bf T}$ of tests such that for any fixed $z\in\pmset^{[n]\setminus I}$, the distribution ${\bf T}(z)$ is equivalent to the following random process: Sample $t=\poly(n)$ random points $y^{(1)},...,y^{(t)}$ in $\mathfrak{C}_z$, and accept if and only if $\Phi(y^{(i)})=\sigma$ for every $i\in[t]$. 

To construct the distribution ${\bf T}$, for every $x\in\pmset^{|I|}$ we define a corresponding test $T_x$ as follows: The test $T_x$ gets input $z\in\pmset^{[n]\setminus I}$, extends $z$ to an $n$-bit string $y\in\pmset^n$ using the values specified in $x$ (i.e., $y_i=x_i$ for every $i\in I$, and $y_i=z_i$ otherwise), and accepts $z$ if and only if $\Phi(y)=\sigma$. Observe that $T_x$ simply computes an LTF of its input $z$ (see Eq.~\eqref{eq:bias:ltf}). Also note that for any \emph{fixed input $z\in\pmset^{[n]\setminus I}$}, a uniform choice of $x\in\pmset^{|I|}$ yields a uniform point $y\in\mathfrak{C}_z$. Each test in ${\bf T}$ corresponds to a tuple $\bar{x}=(x^{(1)},...,x^{(t)})\in\pmset^{t\cdot|I|}$, and computes the function $T_{\bar{x}}(z)=\land_{i\in[t]}T_{x^{(i)}}(z)$.

Assume that $\Phi$ is initially $\de$-close to $\sigma$, for $\de\le1/\poly(n)$. We say that an input $z\in\pmset^{[n]\setminus I}$ is {\sf excellent} if $\Phi$ is $\sqrt{\de}$-close to $\sigma$ in $\mathfrak{C}_z$, and we say that $z$ is {\sf bad} if $\Phi$ is not $\de'$-close to $\sigma$ in $\mathfrak{C}_z$, where $\de'=\de^{\Omega(1)}$. Let $E$ be the set of excellent inputs, and let $B$ be the set of bad inputs. If we choose the parameter $t$ (i.e., the number of sample points) such that $\frac{O(\log(n))}{\de'}< t<\frac1{\sqrt{\de}\cdot\poly(n)}$, then the distribution ${\bf T}$ accepts every $z\in E$ with probability $1-1/\poly(n)$, and rejects every $z\in B$, with probability $1-1/\poly(n)$.

What remains to show is that a distribution ${\bf z}$ that is $(1/\poly(n))$-pseudorandom for LTFs is also $(1/\poly(n))$-pseudorandom for almost all tests in the support of ${\bf T}$. To do so, note that almost all inputs $z\in\pmset^{[n]\setminus I}$ are excellent, and each excellent input is accepted with high probability by a random test $T_{\bar{x}}\sim{\bf T}$. Thus, almost all of the residual deterministic tests $T_{\bar{x}}$ in the support of ${\bf T}$ \emph{accept almost all of their inputs}; in particular, at least $1-1/\poly(n)$ of the residual tests have acceptance probability at least $1-1/\poly(n)$. Every such test is the conjunction of $t=\poly(n)$ LTFs, and each of these LTFs has acceptance probability at least $1-1/\poly(n)$. By a union-bound over the $t$ LTFs, the acceptance probability of such $T_{\bar{x}}$ under ${\bf z}$ is also $1-t\cdot(1/\poly(n))=1-1/\poly(n)$.

\begin{proof}[{\bf Proof of Lemma~\ref{lem:bias:full}.}]
Without loss of generality, assume that $\Phi$ is $\de$-close to the constant $\sigma=-1$. For any Boolean function $f$ over a domain $\mathfrak{D}$, let $\acc(f)=\Pr_{x\sim\mathfrak{D}}[f(x)=-1]$. Also, denote $J=[n]\setminus I$ and $n'=|J|$, and for any $z\in\bitset^{n'}$, denote by $\rho_z$ the restriction $\rho_z=(I,z)$ (i.e., we suppress $I$ in the notation $\rho_z$, since $I$ is fixed).

Let $G=\left\{ z\in\bitset^{n'} : \acc(\Phi\rest_{\rho_z})\ge1-\de' \right\}$. Our goal is to show that $\Pr_{z\sim{\bf z}}[z\in G]\ge1-O(\gamma)$. Let $E=\left\{ z\in\bitset^{n'} : \acc(\Phi\rest_{\rho_z})\ge1-\sqrt{\de} \right\}$. Note that when $z\in\pmset^{n'}$ is chosen uniformly it holds that $\E_{z\in\pmset^{n'}}\left[ \acc(\Phi\rest_{\rho_z})\right]=\Pr_{x\in\pmset^n}[\Phi(x)=-1]\ge1-\de$. Therefore, $\Pr_{z\in\pmset^{n'}}[z\in E]\ge1-\sqrt{\de}$. 

We now construct a distribution ${\bf T}$ over tests $\pmset^{n'}\ra\pmset$ that distinguishes, with high probability, between $z\in E$ and $z\notin G$. For $x\in\bitset^{|I|}$, let $T_x$ be the function that gets as input $z\in\bitset^{n'}$, and outputs the value $\Phi(y)$, where $y_{J}=z$ and $y_I=x$. Note that for any fixed $z\in\pmset^{n'}$, when uniformly choosing $x\in\pmset^{|I|}$ it holds that $\Pr\left[ T_{x}(z)=-1 \right]=\acc(\Phi\rest_{\rho_z})$.  Also, $T_x$ is an LTF of its input $z$, because
\mm{
T_x(z)=\sign\left( \ip{y,w} -\theta \right) = \sign\left( \ip{z,w_{J}} - (\theta-\ip{x,w_I}) \right) \mathdot \eqtag{eq:bias:ltf}
}

For $t=O\left(\frac{\log(1/\gamma)}{\de'}\right)$ and $\bar{x}=(x^{(1)},...,x^{(t)})\in\bitset^{t\cdot|I|}$, let $T_{\bar{x}}:\pmset^{n'}\ra\pmset$ be the function such that $T_{\bar{x}}(z)=-1$ if and only if for every $i\in[t]$ it holds that $T_{x^{(i)}}(z)=-1$ (i.e., $T_{\bar{x}}$ is the conjunction $\land_{i\in[t]}T_{x^{(i)}}$). Our distribution ${\bf T}$ is the uniform distribution over the set $\left\{ T_{\bar{x}} : \bar{x}\in\bitset^{t\cdot|I|} \right\}$.  Observe that:
\begin{itemize}
	\item For any fixed $z\in E$ it holds that $\Pr_{T_{\bar{x}}\sim{\bf T}}\left[ T_{\bar{x}}(z)=-1 \right]\ge1-t\cdot\sqrt{\de}$.
	\item For any fixed $z\notin G$ it holds that $\Pr_{T_{\bar{x}}\sim{\bf T}}\left[ T_{\bar{x}}(z)=-1 \right]\le \gamma$.
\end{itemize}

We want to show that almost all of the tests $\{T_{\bar{x}}\}_{\bar{x}\in\bitset^{t\cdot|I|}}$ in the support of ${\bf T}$ accept almost all of their inputs. To see that this is the case, observe that
\mm{
\E_{\bar{x}}\left[ \acc(T_{\bar{x}}) \right] = \Pr_{\bar{x},z}[T_{\bar{x}}(z)=-1] \ge \Pr_z[z\in E]\cdot\min_{z\in E}\left\{ \Pr_{\bar{x}}[ T_{\bar{x}}(z)=-1 ]\right\} \mathcomma
}
which is lower-bounded by $1-\xi^2$, where $\xi^2=(t+1)\cdot\sqrt{\de}$. Therefore, the fraction of tests $T_{\bar{x}}$ that reject more than $\xi$ of their inputs is at most $\xi$.

Now, let $T_{\bar{x}}$ be a test such that $\acc(T_{\bar{x}})\ge1-\xi$. Since $T_{\bar{x}}$ is a conjunction of $T_{x^{(1)}},...,T_{x^{(t)}}$, for each $i\in[t]$ it holds that $\acc(T_{x^{(i)}})\ge1-\xi$. Also, for each $i\in[t]$ it holds that ${\bf z}$ is $\eta$-pseudorandom for $T_{x^{(i)}}$, where $\eta\le(\gamma^2\cdot\de')$, and therefore $\Pr_{z\sim{\bf z}}[T_{x^{(i)}}(z)=-1]\ge1-\xi-\eta$. It follows that $\Pr_{z\sim{\bf z}}[T_{\bar{x}}(z)=-1]\ge1-t\cdot(\xi+\eta)$.

We invoke Lemma~\ref{lem:randtest} with the parameters $\e_1=\sqrt{\de}$, $\e_2=t\cdot\sqrt{\de}$, $\e_3=\gamma$, $\e_4=\xi$, and $\e_5=t\cdot(\xi+\eta)$, and deduce that 
\mm{
\Pr_{z\sim{\bf z}}[z\notin G] &\le (t+1)\cdot\sqrt{\de} + \gamma + 2\cdot\sqrt{t+1}\cdot\de^{1/4}+t\cdot(\sqrt{t+1}\cdot\de^{1/4}+\eta) \\
&= O\left( \gamma + t^{3/2}\cdot\de^{1/4} + t\cdot\eta \right) \\
&= O\left( \gamma + (\gamma\cdot\de')^{-3/2}\cdot\de^{1/4} + \eta/(\gamma\cdot\de') \right) \mathcomma
}
which is $O(\gamma)$ since $\eta\le(\gamma^2\cdot\de')$ and by our hypotheses regarding $\gamma$, $\de$, and $\de'$.
\end{proof}

\section{Reduction of standard derandomization to quantified derandomization} \label{sec:threshold}

In this section we prove Theorem~\ref{thm:int:threshold}. The core of the proof is the construction of a suitable averaging sampler (equivalently, seeded extractor) that is computable by a $\tc^0$ circuit with a super-linear number of wires. We therefore start by describing this construction. In the current section, as in Section~\ref{sec:pre:ext}, it will be more convenient to represent Boolean functions as functions $\bitset^n\ra\bitset$, rather than $\pmset^n\ra\pmset$.

In Section~\ref{sec:threshold:design} we recall the definition of weak combinatorial designs, and construct such designs that are suitable for our parameter setting. In Section~\ref{sec:threshold:code} we show how to compute a code with distance $1/2-o(1)$ by a $\tc^0$ circuit with a super-linear number of wires. In Section~\ref{sec:threshold:extractor} we combine the two preceding ingredients to construct an averaging sampler in $\tc^0$. Finally, in Section~\ref{sec:threshold:thm} we prove Theorem~\ref{thm:int:threshold}.

\subsection{Weak combinatorial designs for Trevisan's extractor} \label{sec:threshold:design}

Let us recall the notion of weak combinatorial designs, which was introduced by Raz, Reingold, and Vadhan~\cite{rrv02}.

\begin{definition} (weak designs). \label{def:threshold:design}
For positive integers $m,\ell,t\in\N$ and an integer $\rho>1$, an {\sf $(m,\ell,t,\rho)$ weak design} is a collection of sets $S_1,...,S_m\subseteq[t]$ such that for every $i\in[m]$ it holds that $|S_i|=\ell$ and $\sum_{j<i}2^{|S_i\cap S_j|} \le (m-1)\cdot\rho$.
\end{definition}

Raz, Reingold, and Vadhan~\cite{rrv02} showed a construction of weak designs with universe size $t=\ceil{\frac{\ell}{\ln(\rho)}}\cdot\ell$. In our parameter setting we will have $\log(\rho)\approx0.99\cdot\ell$, and for such value the construction in~\cite{rrv02} yields $t=2\cdot\ell$. We want to have $t\approx1.01\cdot\ell$, and therefore now show a more refined construction.

\begin{lemma} (constructing weak designs). \label{lem:threshold:designs}
There exists an algorithm that gets as input $m\in\N$ and $\ell\in\N$ and $\rho\in\N$ such that $\log(\rho)=(1-\al)\cdot\ell$, where $\al\in(0,1/4)$, and satisfies the following. The algorithm runs in time $\poly(m,2^{\ell})$ and outputs an $(m,\ell,t,\rho)$ weak design, where $t=\ceil{(1+4\al)\cdot\ell}$.
\end{lemma}

\begin{proof}[{\bf Proof.}]
Let $t=\ceil{(1+4\al)\cdot\ell}$. The algorithm constructs the sets $S_1,...,S_m\subseteq[t]$ in iterations. In each iteration $i\in[m]$ the algorithm finds $S_i$ such that $\sum_{j<i}2^{|S_i\cap S_j|}\le(i-1)\cdot\rho$. To do so, the algorithm initially fixes a partition of $[t]$ into $\ell$ blocks. The first $t-\ell$ blocks, denoted $B_1,...,B_{t-\ell}$, are each comprised of two elements (i.e., for $j\in[t-\ell]$ it holds that $B_j=\{2j-1,2j\}$). The remaining $2\ell-t$ blocks, denoted $B_{t-\ell+1},...,B_{\ell}$, each consist of a single element (i.e., for $j\in\{t-\ell+1,...,\ell\}$ it holds that $B_j=\{t-\ell+j\}$).

For $i\in[m]$, let us describe the $i^{th}$ iteration, after $S_1,...,S_{i-1}$ were already chosen in previous iterations. Consider a set $S_i$ that is chosen by independently choosing one random element from each of the $\ell$ blocks to include in $S_i$.~\footnote{That is, for each $k\in[\ell]$ let $X_k$ be a random element from the block $B_k$, such that for $k\ne k'\in[\ell]$ it holds that $X_k$ and $X_{k'}$ are independent. Then, $S_i=\cup_{k\in[\ell]}X_k$.} For $j\in[i-1]$ and $k\in[\ell]$, let $Y_{j,k}$ be the indicator variable of whether the element from the $k^{th}$ block that is included in $S_j$ is also included in $S_i$ (i.e., $Y_{j,k}=1$ iff $B_k\cap S_j\cap S_i\ne\emptyset$). Note that for $k\ne k'\in[m]$ it holds that $Y_{j,k}$ and $Y_{j,k'}$ are independent. Thus, the expected value of $\sum_{j<i}2^{|S_i\cap S_j|}$ is
\mm{
\E\left[ \sum_{j<i}2^{|S_i\cap S_j|} \right] &= \sum_{j<i}\E\left[ 2^{\sum_{k\in[\ell]}Y_{j,k}} \right] \\
&= \sum_{j<i}\E\left[ \prod_{k\in[\ell]}2^{Y_{j,k}} \right] \\
&= \sum_{j<i} \prod_{k\in[\ell]} \E\left[ 2^{Y_{j,k}} \right] \\
&= (i-1)\cdot\left(3/2\right)^{t-\ell}\cdot2^{2\ell-t} \mathcomma \eqtag{eq:sillycalc}
}
where the last equality is because for every $k\in[t-\ell]$ it holds that $\Pr[Y_{j,k}=1]=1/2$ (since $|B_k|=2$), and for every $k\in\{t-\ell+1,...,\ell\}$ it holds that $Y_{j,k}\equiv1$ (since $B_k$ is a singleton). Now, plugging-in $t=\ceil{(1-4\al)\cdot\ell}$ and $\ell=\frac{\log(\rho)}{1-\al}$ into Eq.~\eqref{eq:sillycalc}, we can upper-bound the expression by $(i-1)\cdot\rho$.~\footnote{Denoting $c=\log(e)/2$ and $t=(1+4\beta)\cdot\ell$, where $\beta\ge\al$, we have that $2^{2\ell-t}\cdot(3/2)^{t-\ell}<2^{2\ell-t}\cdot e^{(t-\ell)/2} = 2^{2\ell-t+c\cdot(t-\ell)}\le2^{\frac{1-4(1-c)\cdot\beta}{1-\al}\cdot\log(\rho)}<\rho$.} 
Hence, the algorithm can find a set $S_i$ such that $\sum_{j<i}2^{|S_i\cap S_j|}\le(i-1)\cdot\rho$ by trying out all $2^{t-\ell}<2^{\ell}$ possibilities.
\end{proof}

As shown in~\cite{rrv02}, Trevisan's proof~\cite{tre01} that the Nisan-Wigderson construction~\cite{nw94} yields an extractor also extends to the setting when the combinatorial design is a weak design as in Definition~\ref{def:threshold:design}. Specifically:

\begin{theorem} (extractors from weak designs~\cite[Prop. 10]{rrv02}). \label{thm:threshold:trevisan}
Let $m<k<n$ be three integers, and let $\e>0$. Let $\code:\bitset^n\ra\bitset^{\bar{n}}$ be a code such that in every Hamming ball of radius $1/2-\de$ in $\bitset^{\bar{n}}$ there exist at most $1/\de^2$ codewords, where $\de=\e/4m$. Let $S_1,...,S_m\subseteq[t]$ be an $(m,\ell,t,\rho)$ weak design with $\ell=\log(\bar{n})$ and $\rho=\frac{k-3\cdot\log(m/\e)-t-3}{m}$. 

Then, the function $E:\bitset^n\times\bitset^t\ra\bitset^m$ that is defined by $E(x,z)=\linebreak(\code(x)_{z_{S_1}},...,\code(x)_{z_{S_m}})$ is a $(k,\e)$-extractor.
\end{theorem}

By combining Theorem~\ref{thm:threshold:trevisan} and Proposition~\ref{prop:threshold:extsamp}, we obtain the following:

\begin{corollary} (samplers from weak designs). \label{cor:threshold:trevisan}
Let $m<k<n$ be three integers, and let $\e>0$. Let $\code:\bitset^n\ra\bitset^{\bar{n}}$ be a code such that in every Hamming ball of radius $1/2-\de$ in $\bitset^{\bar{n}}$ there exist at most $1/\de^2$ codewords, where $\de=\e/4m$. Let $S_1,...,S_m\subseteq[t]$ be an $(m,\ell,t,\rho)$ weak design with $\ell=\log(\bar{n})$ and $\rho=\frac{k-3\cdot\log(m/\e)-t-3}{m}$. 

Then, the function $Samp:\bitset^n\times\bitset^t\ra\bitset^m$ that is defined by $Samp(x,z)=\linebreak(\code(x)_{z_{S_1}},...,\code(x)_{z_{S_m}})$ is an averaging sampler with accuracy $\e$ and error $2^{k-n}$.
\end{corollary}

\subsection{An $\e$-balanced code in sparse $\tc^0$} \label{sec:threshold:code}

Following Corollary~\ref{cor:threshold:trevisan}, our goal now is to construct a $\tc^0$ circuit with a super-linear number wires that computes an \emph{error-correcting code} that is list-decodable up to distance $1/2-\de$ with list size $\poly(1/\de)$ and rate $\poly(1/\de)$. We will do this by constructing a code with distance $1/2-\e$, where $\e=\de^2$, and then relying on the Johnson bound. In fact, we will actually construct an $\e$-balanced code (i.e., a linear code such that all codewords have relative Hamming weight $1/2\pm\e$).

As described in the introduction, the construction will consist of two parts. We will first construct a code with constant relative distance, and then show how to amplify the distance from $\Omega(1)$ to $1/2-\e$.

\begin{proposition} (a code with constant relative distance in sparse $\tc^0$). \label{prop:threshold:code:const}
There exists a polynomial-time algorithm that is given as input $1^{n}$ and a constant $d\in\N$, and outputs a $\tc^0$ circuit $C$ that satisfies the following:
\begin{enumerate}
	\item The circuit $C$ maps $n$ input bits to $\hat{n}=O(n)$ input bits.
	\item For every $x\in\bitset^n$ such that $x\ne0^n$, the relative Hamming weight of $C(x)$ is at least $3^{-d}$.
	\item Each output bit of $C$ is a linear function of the input bits.
	\item The circuit $C$ has depth $2d$ and $n^{1+O(1/d)}$ wires.
\end{enumerate}
\end{proposition}

\begin{proof}[{\bf Proof.}]
Assume that $n$ is of the form $r^{d}$, for $r\in\N$ (if necessary, pad the input with zeroes such that the input length will be a power of $2^{d}$). Fix a linear code $\code$ that maps strings of length $r$ to strings of length $\bar{r}=O(r)$ and has relative distance at least $1/3$ (e.g., we can use the $\e$-balanced codes of~\cite{nn93,tas17}).

Let $x\in\bitset^n$ be an input for the circuit $C$. We think of $x$ as a tensor $M^{(0)}$ of dimensions $[r]^d$; that is, for every $\vec{t}\in[r]^d$, the $\vec{t}^{th}$ entry of $M^{(0)}$ is denoted by $M^{(0)}_{\vec{t}}\in\bitset$. The circuit $C$ will iterative compute a sequence $M^{(1)},...,M^{(d)}$ of tensors, and the message $x=M^{(0)}$ will be mapped to the final codeword $\hat{x}=M^{(d)}$.

For each $i\in[d]$, the tensor $M^{(i)}$ is defined as follows. The dimensions of $M^{(i)}$ are $[\bar{r}]^{i}\times[r]^{d-i}$. For every pair $(\vec{t_{\le i-1}},\vec{t_{\ge i+1}})\in[\bar{r}]^{i-1}\times[r]^{d-i}$, we denote by $M^{(i-1)}_{\vec{t}_{\le i-1},\star,\vec{t}_{\ge i+1}}$ the $r$-bit vector $M^{(i-1)}_{\vec{t}_{\le i-1},\star,\vec{t}_{\ge i+1}}\eqdef M^{(i-1)}_{(\vec{t}_{\le i-1},1,\vec{t}_{\ge i+1})},...,M^{(i-1)}_{(\vec{t}_{\le i-1},m,\vec{t}_{\ge i+1})}\in\bitset^r$. Then, for every $\vec{t}\in[\bar{r}]^{i}\times[r]^{d-i}$, we think of $\vec{t}$ as a triplet $\vec{t}=(\vec{t}_{\le i-1},u,\vec{t}_{\ge i+1})\in[\bar{r}]^{i-1}\times[\bar{r}]\times[r]^{d-i}$, and define $M^{(i)}_{\vec{t}} =  \left(\code\left(M^{(i-1)}_{\vec{x}_{\le i-1},\star,\vec{x}_{\ge i+1}}\right)\right)_v$ (i.e., $M^{(i)}_{(\vec{t}_{\le i-1},v,\vec{t}_{\ge i+1})}$ is the $v^{th}$ coordinate of the encoding of $M^{(i-1)}_{\vec{t}_{\le i-1},\star,\vec{t}_{\ge i+1}}$ by $\code$). 

The final codeword $\hat{x}=M^{(d)}$ is of dimensions $[\bar{r}]^{d}$, which means that it represents a string of length $\hat{n}=(O(r))^d=O(n)$. The fact that every non-zero message $x\in\bitset^n$ is mapped to a codeword $\hat{x}\in\bitset^{\hat{n}}$ with relative Hamming weight at least $(1/3)^d$ follows from the properties of $\code$ and from well-known properties of tensor codes; for completeness, we include a proof in Appendix~\ref{app:tensor}. Also note that each bit of $\hat{x}$ is indeed a linear function of $x$, because $\code$ is linear (which means that in each iteration $i\in[d]$, every bit of $M^{(i)}$ is a linear function of $M^{(i-1)}$).

Finally, let us fix $i\in[d]$, and describe how to compute $M^{(i)}$ from $M^{(i-1)}$ in depth two with $O(n\cdot r^2)$ wires. Since $\code$ is linear, for each $\vec{t}=(\vec{t}_{\le i-1},v,\vec{t}_{\ge i+1})\in[\bar{r}]^{i-1}\times[\bar{r}]\times[r]^{d-i}$ it holds that $M^{(i)}_{\vec{t}}=\code\left( M^{(i-1)}_{\vec{t}_{\le i-1},\star,\vec{t}_{\ge i+1}} \right)_v$ is a linear function of the $r$-bit string $M^{(i-1)}_{\vec{t}_1,\star,\vec{t}_2}\in\bitset^r$. Thus, each entry of $M^{(i)}$ can be computed from $M^{(i-1)}$ by a depth-$2$ $\tc^0$ circuit with $O(r^{2})$ wires (see, e.g.,~\cite[Sec. 3]{ps94}), which means that $M^{(i)}$ can be computed from $M^{(i-1)}$ by a depth-$2$ $\tc^0$ circuit with $O(n\cdot r^2)$ wires. Overall, the final circuit $C$ is of depth $2d$ (since it is comprised of $d$ circuits of depth two), and the number of wires in $C$ is at most $O(n\cdot r^2) < n^{1+O(1/d)}$.
\end{proof}

We now show how to amplify the distance of the code from Proposition~\ref{prop:threshold:code:const} from $\Omega(1)$ to $1/2-\e$.

\begin{proposition} (amplifying the distance of the code from Proposition~\ref{prop:threshold:code:const}). \label{prop:threshold:code:amp}
There exists a polynomial-time algorithm that is given as input $1^{\hat{n}}$, a constant $\rho>0$, and $\e=\e(\hat{n})>0$, and outputs a $\tc^0$ circuit $C$ such that:
\begin{enumerate}
	\item The circuit $C$ maps $\hat{n}$ input bits to $\bar{n}=\hat{n}\cdot(1/\e)^{O(1/\rho)}$ output bits.
	\item For every $\hat{x}\in\bitset^{\hat{n}}$ with relative Hamming weight at least $\rho$, the relative Hamming weight of $\bar{x}=C(\hat{x})$ is between $1/2-\e$ and $1/2$.
	\item Each output bit of $C$ is a linear function of the input bits.
	\item The circuit $C$ has depth two and $\hat{n}\cdot(1/\e)^{O(1/\rho)}$ wires.
\end{enumerate}
\end{proposition}

\begin{proof}[{\bf Proof.}]
The algorithm first constructs an expander graph $G$ on $\hat{n}$ vertices; that is, a $d_G$-regular graph over the vertex-set $[\hat{n}]$ vertices with constant spectral gap.~\footnote{For a suitable construction see, e.g.,~\cite[Thm E.10]{gol08}. This specific construction requires $\hat{n}$ to be a square, so we might need to pad the input $x\in\bitset^{\hat{n}}$ with zeroes such that it will be of length $4^{k}=(2^{k})^2$, for $k\in\N$. Since such a padding will not affect the rest of the argument, we ignore this issue.} Consider a random walk that starts from a uniform $i\in[\hat{n}]$ and walks $\ell-1$ steps, where $\ell=\frac{c_G}{\rho}\cdot \log(1/\e)$ and $c_G$ is a sufficiently large constant that depends only on $G$. By the hitting property of expander random walks (see, e.g.,~\cite[Thm 8.28]{gol08}), with probability at least $1-\e$ such a walk hits $i\in[\hat{n}]$ such that $x_i\ne0$ (this is because the set $\{i\in[\hat{n}]:x_i\ne0\}$ has density at least $\rho$). Thus, if we first take such a random walk, and then output a random parity of the values of $\hat{x}$ at the coordinates corresponding to the vertices in the walk, the output will equal one with probability at least $1/2-\e$ and at most $1/2$.

The mapping of $\hat{x}$ to $\bar{x}=C(\hat{x})$ is obtained by considering all the possible outcomes of the random process above. Specifically, for every random walk $W=\left(i^{(W)}_1,...,i^{(W)}_{\ell}\right)$ of length $\ell-1$ on $G$, and every subset $S\subseteq[\ell]$, we have a corresponding coordinate $(W,S)$ in $C(\hat{x})$. The value of $C(\hat{x})$ at coordinate $(W,S)$  is the parity of the bits of $\hat{x}$ in the  locations corresponding to $S$ in walk $W$; that is, $C(\hat{x})_{(W,S)}=\xor_{j\in S}\hat{x}_{i^{(W)}_j}$. 

Note that the length of $C(\hat{x})$ is $\hat{n}\cdot (d_G)^{\ell-1}\cdot 2^{\ell}=\hat{n}\cdot(1/\e)^{c'_G/\rho}$, where $c'_G$ is a large constant that only depends on $G$. Also, the mapping of $\hat{x}$ to $C(\hat{x})$ is linear, and moreover every coordinate of $C(\hat{x})$ is the parity of $\ell$ coordinates of $\hat{x}$. Thus, $C(\hat{x})$ can be computed by a $\tc^0$ circuit of depth two using at most $\hat{n}\cdot(1/\e)^{c/\rho}\cdot\ell^2<\hat{n}\cdot(1/\e)^{2c/\rho}$ wires.
\end{proof}

By combining Propositions~\ref{prop:threshold:code:const} and~\ref{prop:threshold:code:amp} we obtain the following:

\begin{proposition} (an $\e$-balanced code in sparse $\tc^0$). \label{prop:threshold:code}
There exists a polynomial-time algorithm that gets inputs $1^n$ and $\e=\e(n)$ and a constant $d\in\N$, and outputs a $\tc^0$ circuit such that:
\begin{enumerate}
	\item The circuit computes a linear code that maps messages of length $n$ to codewords of length $\bar{n}=n\cdot(1/\e)^{O(3^d)}$ such that every codeword has relative Hamming weight $1/2\pm\e$.
	\item The circuit has depth $2d$ and $n^{1+O(1/d)}+n\cdot(1/\e)^{O(3^d)}$ wires.
\end{enumerate}
\end{proposition}

Relying on the Johnson bound, we obtain the list-decodable code that is needed for Corollary~\ref{cor:threshold:trevisan} as a corollary of Proposition~\ref{prop:threshold:code}:

\begin{corollary} (a list-decodable code in sparse $\tc^0$). \label{cor:threshold:code}
There exists a polynomial-time algorithm that gets inputs $1^n$ and $\de=\de(n)$ and a constant $d\in\N$, and outputs a $\tc^0$ circuit such that:
\begin{enumerate}
	\item The circuit computes a linear code mapping messages of length $n$ to codewords of length $\bar{n}=n\cdot(1/\de)^{O(3^d)}$ such that in any Hamming ball of radius $1/2-\de$ in $\bitset^{\bar{n}}$ there exist at most $O(1/\de^2)$ codewords.
	\item The circuit has depth $2d$ and $n^{1+O(1/d)}+n\cdot(1/\de)^{O(3^d)}$ wires.
\end{enumerate}
\end{corollary}

\begin{proof}[{\bf Proof.}]
We invoke Proposition~\ref{prop:threshold:code} with $\e=\de^2$. The code that the circuit computes has distance $1/2-\de^2$. Relying on the Johnson bound (see, e.g.,~\cite[Thm 19.23]{ab09}), in such a code every Hamming ball of radius $\de$ contains at most $1/\de^2$ codewords.
\end{proof}

\subsection{An averaging sampler in sparse $\tc^0$} \label{sec:threshold:extractor}

We now combine Lemma~\ref{lem:threshold:designs}, Corollary~\ref{cor:threshold:trevisan}, and Corollary~\ref{cor:threshold:code}, to get an averaging sampler that can be computed by a $\tc^0$ circuit with a super-linear number of wires. The sampler will get an input of length $n$, and for two constants $0<\gamma\ll\beta<1$, the sampler will output $m=n^{\gamma}$ bits and will have accuracy $1/m$ and error $2^{n^{\beta}-n}$.

\begin{proposition} (an averaging sampler in sparse $\tc^0$). \label{prop:threshold:sampler}
There exists a polynomial-time algorithm that gets as input $1^n$ and three constants $d\in\N$ and $\gamma\le\frac1{c\cdot d\cdot3^d}$ (where $c>1$ is some universal constant) and $\beta\ge4/5$, and outputs a $\tc^0$ circuit $C$ that satisfies the following:
\begin{enumerate}
	\item The circuit $C$ gets input $x\in\bitset^n$ and outputs $2^t<n^{(1+O(1/d))\cdot(5-4\beta)}$ strings of length $m=n^{\gamma}$.
	\item The function $Samp:\bitset^n\times\bitset^t\ra\bitset^m$ such that $Samp(x,i)=C(x)_i$ (i.e., $Samp(x,i)\in\bitset^m$ is the $i^{th}$ output string of $C(x)$) is an averaging sampler with accuracy $\e=1/m$ and error $2^{n^{\beta}-n}$.
		\item The depth of $C$ is $2d+1$ and its number of wires is at most $n^{(1+O(1/d))\cdot(5-4\beta)}$.
\end{enumerate}
In particular, if $\beta\ge1-1/5d$, then both the number of outputs of $C$ (i.e., $2^t$) and the number of wires in $C$ are less than $n^{1+O(1/d)}$.
\end{proposition}

\begin{proof}[{\bf Proof.}]
We first use Corollary~\ref{cor:threshold:code} with the parameter value $\de=\e/4m$ to construct a circuit $C_0$ of depth $2d$ that encodes its input $x\in\bitset^n$ to a codeword $\bar{x}$ of length $\bar{n}$. Then, we use Lemma~\ref{lem:threshold:designs} to construct an $(m,\ell,t,\rho)$ weak design $S_1,...,S_m\subseteq[t]$ with the following parameters: For $\al=1-\beta+(c\cdot3^{d+1})\cdot\gamma<1/4$ (the inequality is since $\beta>4/5$ and $\gamma$ is sufficiently small), we construct a design with $\ell=\log(\bar{n})$ and $\rho=2^{(1-\al)\cdot\ell}$ and $t=\ceil{(1+4\al)\cdot\ell}$. Now, define a function $Samp:\bitset^n\times\bitset^t\ra\bitset^m$ as in Corollary~\ref{cor:threshold:trevisan}; that is, for $x\in\bitset^n$ and $z\in\bitset^t$, the $m$-bit string $Samp(x,z)$ is the projection of $\bar{x}$ to the coordinates $z_{S_1},...,z_{S_m}$. The circuit $C$ outputs the $2^t$ strings corresponding to $\{Samp(x,z)\}_{z\in\bitset^t}$, where each output string is a projections of $m$ bits of $\bar{x}$.

Let $k=n^{\beta}$. An elementary calculation shows that $\rho=2^{(1-\al)\cdot\ell}<\frac{k-3\cdot\log(m/\e)-t-3}{m}$.~\footnote{To see that this holds, let $c'>1$ be the universal constant such that $\bar{n}\le n\cdot(m/\e)^{c'\cdot3^d}$. Then, note that $\al=1-\beta+(c'\cdot3^{d+1})\cdot\gamma>\frac{1-\beta+(2c'\cdot3^d+1)\cdot\gamma}{1+2c'\cdot \gamma\cdot3^d}=1-\frac{\beta-\gamma}{1+2c'\cdot\gamma\cdot3^d}$. It follows that $\log(\rho)=(1-\al)\cdot\ell<\log(k/2m)$, since $1-\al<\frac{\beta-\gamma-1/\log(n)}{1+2c'\cdot\gamma\cdot3^d}$. We can thus deduce that $\rho\le k/2m<\frac{k-3\cdot\log(m/\e)-t-3}{m}$.} Thus, relying on Corollary~\ref{cor:threshold:trevisan}, the function $Samp$ is an averaging sampler with accuracy $\e$ and error $2^{k-n}$.  The depth of $C$ is $2d+1$ (since the depth of $C_0$ is $2d$, and the $2^t$ outputs are projections of $\bar{x}$). Finally, the number of wires in $C_0$ is $n^{1+O(1/d)}+n\cdot(m/\e)^{O(3^d)}<n^{1+O(1/d)}$, and the number of wires between $\bar{x}$ and the outputs is $2^t\cdot m=2^{\ceil{(1+4\al)\cdot\log(\bar{n})}}\cdot m=n^{(1+O(\gamma\cdot3^d))(1+4\al)}=n^{(1+O(1/d))\cdot(5-4\beta)}$.
\end{proof}

\subsection{Proof of Theorem~\ref{thm:int:threshold}} \label{sec:threshold:thm}

Let us now formally state Theorem~\ref{thm:int:threshold} and prove it using the averaging sampler from Proposition~\ref{prop:threshold:sampler}.  Towards stating the theorem, for any $n,d,k\in\N$, denote by $\mathcal{C}_{n,d,n^k}$ either the class of linear threshold circuits over $n$ input bits of depth $d$ and with at most $n^k$ wires. 

\begin{theorem} (Theorem~\ref{thm:int:threshold}, restated). \label{thm:threshold}
Assume that for every $d\in\N$ and for some $\beta=\beta_d\ge4/5$ there exists an algorithm that gets as input a circuit $C'\in\mathcal{C}_{n,d,n^{(1+O(1/d))\cdot(5-4\beta)}}$, runs in time $T(n)$, and satisfies the following: If $C'$ rejects all but at most $2^{n^{\beta}}$ of its inputs, then the algorithm rejects $C'$, and if $C'$ accepts all but at most $2^{n^{\beta}}$ of its inputs, then the algorithm accepts $C'$.

Then, there exists an algorithm that for every $k\in\N$ and $d\in\N$, when given as input a circuit $C\in\mathcal{C}_{m,d,m^k}$, runs in time $T(m^{O(k\cdot d\cdot3^d)})$ (where the $O$-notation hides some fixed universal constant), and satisfies the following: If $C$ accepts at least $2/3$ of its inputs then the algorithm accepts $C$, and if $C$ rejects at least $2/3$ of its inputs then the algorithm rejects $C$.
\end{theorem}

To obtain the parameters of Theorem~\ref{thm:int:threshold}, use the value $\beta_d=1-1/5d$, in which case the number of wires of $C'$ is $n^{1+O(1/d)}$; and for every $k\in\N$, we can assume that $d$ is sufficiently large such that $O(k\cdot d\cdot 3^d\cdot4^{-d})<1$, in which case the running time of the algorithm is at most $T(m^{O(k\cdot d\cdot3^d)})=2^{m^{1-\Omega(1)}}$ (due to the hypothesis that $T(n)=2^{n^{1/4^d}}$).

\begin{proof}[{\bf Proof of Theorem~\ref{thm:threshold}.}]
Let $C\in\mathcal{C}_{m,d,m^k}$ be an input to the algorithm, let $\gamma=1/c\cdot k\cdot d\cdot3^d$ for a sufficiently large universal constant $c>1$, and let $\beta=\beta_{3d+2}$. We will construct a circuit $C'\in\mathcal{C}_{n,3d+2,n^{(1+O(1/d))\cdot(5-4\beta)}}$, where $n=m^{1/\gamma}$, such that the following holds: If $C$ rejects at least a $2/3$ fraction of its inputs, then $C'$ rejects all but at most $2^{n^{\beta}}$ inputs; and if $C$ accepts at least a $2/3$ fraction of its inputs, then $C'$ accepts all but $2^{n^{\beta}}$ of its inputs. Then, we can use the quantified derandomization algorithm for $C'$, which runs in time $T(n)=T\left(m^{c\cdot k\cdot d\cdot 3^d}\right)$, to decide whether the acceptance probability of $C$ is at least $2/3$ or at most $1/3$.

To construct $C'$, we first use Proposition~\ref{prop:threshold:sampler} to construct a $\tc^0$ circuit $Samp:\bitset^n\times\bitset^t\ra\bitset^m$ that is an averaging sampler with the following properties: The input length is $n$, the output length is $m=n^{\gamma}$, the accuracy is $\e=n^{\Omega(1)}<1/100$, and the error is $\de=2^{n^{\beta}-n}$; the number of wires in $Samp$ is at most $n^{(1+O(1/d))\cdot(5-4\beta)}$, and its depth is $2d+1$. The circuit $C'$ first computes the sampler $Samp$, then evalutes $C$ in parallel on each of the $2^t<n^{(1+O(1/d))\cdot(5-4\beta)}$ outputs of the sampler, and finally computes the majority of the $2^t$ evaluations of $C$. That is, $C'(x)=MAJ_{z\in\bitset^t}\left[ C(Samp(x,z)) \right]$. The circuit $C'$ is of depth $(2d+1)+d+1=3d+2$, and its number of wires is at most $n^{(1+O(1/d))\cdot(5-4\beta)}+m^k=n^{(1+O(1/d))\cdot(5-4\beta)}$, where we relied on the fact that $m^k<n$.

Note that for any $x\in\bitset^n$ such that $\Pr_{z\in\bitset^t}\left[ C(Samp(x,z))=1 \right ] \in \Pr[C({\bf u}_n)=1] \pm \e$, we have that $C'(x)$ outputs the most frequent value of $C$. Since the accuracy of the sampler is $2^{n^{\beta}-n}$, the number of strings in $\bitset^n$ such that $\Pr_{z\in\bitset^t}\left[ C(Samp(x,z))=1 \right ]\notin \Pr[C({\bf u}_n)=1] \pm \e$ is at most $2^{n^{\beta}}$. Thus, the number of strings $x\in\bitset^n$ such that $C'(x)$ does not output the most frequent value of $C$ is at most $2^{n^{\beta}}$.
\end{proof}

Observe that the circuit $C'$ that we constructed in the proof of Theorem~\ref{thm:threshold} consists of the sampler from Proposition~\ref{prop:threshold:sampler}, which only uses majority gates; of copies of the initial circuit $C$; and of an additional majority gate. Thus, the statement of Theorem~\ref{thm:threshold} holds even if we interpret $\mathcal{C}_{n,d,w}$ as the class of circuits with majority gates (rather than linear threshold circuits) over $n$ input bits of depth $d$ and with at most $w$ wires.

\section{Quantified derandomization of depth-2 linear threshold circuits} \label{sec:depthtwo}

In this section we construct a quantified derandomization algorithm for depth-2 linear threshold circuits with $n^{3/2-\Omega(1)}$ wires. In fact, we construct a \emph{pseudorandom generator} for the class of depth-2 linear threshold circuits with $n^{3/2-\Omega(1)}$ wires that either accept all but $B(n)=2^{n^{\Omega(1)}}$ of their inputs or reject all but $B(n)$ of their inputs. That is, we construct an algorithm $G$ that gets as input a seed $s$ of length $\tilde{O}(\log(n))$, and outputs an $n$-bit string such that for every $C\in\mathcal{C}_{n,2,n^{3/2-\Omega(1)}}$ the following holds: If $C$ accepts all but $B(n)=2^{n^{\Omega(1)}}$ of its inputs, then the probability that $C(G(s))=1$ is very high, and if $C$ rejects all but $B(n)$ of its inputs, then the probability that $C(G(s))=0$ is very low.

%Compared with the algorithm in our main result (i.e., Theorem~\ref{thm:int:main}), when instantiated to depth two, the pseudorandom generator constructed in this appendix can handle circuits with more wires ($n^{3/2-\Omega(1)}$ instead of $n^{1+2^{-20}}$), but the number of exceptional inputs that it can handle is smaller (i.e., $B(n)=2^{n^{\Omega(1)}}$ instead of $B(n)=2^{n^{1-\Omega(1)}}$). Also, 
The pseudorandom generator that we construct in this appendix is incomparable to the pseudorandom generator of Servedio and Tan~\cite{st17}. On the one hand, their generator is $\frac1{\poly(n)}$-pseudorandom for \emph{every} depth-two linear threshold circuit, whereas our generator only ``fools'' circuits with acceptance probability that is either very high or very low. Moreover, their generator can handle circuits with $n^{2-\Omega(1)}$ wires, whereas our generator can only handle circuits with $n^{3/2-\Omega(1)}$ wires. But on the other hand, their generator requires a seed of length $n^{1-\Omega(1)}$, whereas our generator only requires a seed of length $\tilde{O}(\log(n))$.

Recall that our main quantified derandomization algorithm (from Theorem~\ref{thm:int:main}) leverages the techniques underlying the correlation bounds of Chen, Santhanam, and Srinivasan~\cite{css16} for depth-$d$ linear threshold circuits. The generator in this section leverages the techniques underlying the correlation bounds of Kane and Williams~\cite{kw16} for depth-2 linear threshold circuits. 

Specifically, our first step is to prove a derandomized version of the restriction lemma of Kane and Williams~\cite{kw16}. We actually state a slightly generalized version, which is implicit in the original argument. We say that a distribution ${\bf y}$ over $\bitset^n$ is {\sf $p$-bounded in pairs} if for every $i\ne j\in[n]$ it holds that $\Pr[{\bf y}_i=1]\le p$ and $\Pr[{\bf y}_i=1\land{\bf y}_j=1]\le p^2$. One example for a distribution that is $p$-bounded in pairs is the distribution ${\bf y}$ in which each coordinate is independently set to $1$ with probability $p$. Another example, which is used in~\cite{kw16}, is the following: Consider a equipartition of $[n]$ to $p\cdot n$ disjoint sets $S_1,...,S_{p\cdot n}$; then, sampling $y\sim{\bf y}$ is equivalent to uniformly choosing a single coordinate in each set $S_i$ in the partition, fixing $y$ in the chosen coordinates to one, and fixing $y$ in all other coordinates to zero (so that the Hamming weight of $y\sim{\bf y}$ is always $p\cdot n$).

\begin{proposition} (derandomized version of~\cite[Lem. 3.1]{kw16}). \label{prop:kw:rest}
Let $\Phi=(w,\theta)$ be an LTF on $m$ input bits. For $p>0$, let ${\bf y}$ be a distribution over $\bitset^n$ that is $p$-bounded in pairs, and let ${\bf z}$ be a distribution over $\pmset^n$ that is $\frac1{\poly(m)}$-pseudorandomly concentrated. Let ${\bm\rho}$ be the distribution over restrictions obtained by sampling $y\sim{\bf y}$ in order to determine which variables are kept alive (the $i^{th}$ variable is kept alive if and only if ${\bf y}_i=1$), and independently sampling $z\sim{\bf z}$ to determine values for the fixed variables. Then,
\mm{
\Pr_{\rho\sim{\bm\rho}}[\Phi\rest_\rho\text{ depends on more than one input bit}]=O(m\cdot p^{3/2}) \mathdot
}
\end{proposition}

\begin{proof}[{\bf Proof.}]
For every choice of $y\sim{\bf y}$, let $I=I_y\subseteq[n]$ be the set of live variables (i.e., $I=\{i\in[n]:y_i=1\}$). Then, the probability that $\Phi\rest_\rho$ depends on more than one input bit is at most 
\mm{
&\Pr_{\rho\sim{\bm\rho}}\left[|I|\ge 2 \land \Phi\rest_\rho\text{ is not constant} \right] \\
&= \E_{y\sim{\bf y}}\left[ \Pr_{z\sim{\bf z}}\left[ |I|\ge2 \land  \Phi\rest_\rho\text{ is not constant} \right] \right] \\
&= \E_{y\sim{\bf y}}\left[ \mathbf{1}_{|I|\ge2}\cdot\Pr_{z\sim{\bf z}}\left[ \Phi\rest_\rho\text{ is not constant} \right] \right] \mathcomma \eqtag{eq:kw:highlevel}
}
where the first equality relied on the fact that $y$ and $z$ are sampled independently, and the second equality is since the random variable $I$ only depends on $y$ (and not on $z$).

Fix an arbitrary choice of $y$, and let us upper-bound the probability over $z\sim{\bf z}$ that $\Phi\rest_\rho$ is not constant. Note that $\Phi\rest_\rho$ is a constant function if and only if
\mm{
\abs{\theta-\ip{w_{[m]\setminus I},z_{[m]\setminus I}}}>\norm{w_I}_1 \iff \ip{w_{[m]\setminus I},z_{[m]\setminus I}} \not\in \theta \pm \norm{w_I}_1 \mathdot \eqtag{eq:kw:const}
}

For each $i\in[m]$, let $k_i$ be the index of the $i^{th}$ variable when the variables are sorted according to the magnitudes $|w_i|$ in ascending order (breaking ties arbitrarily). In~\cite[Proof of Lemma 1.1]{kw16} it is shown that the probability over a \emph{uniform choice of $z$} that Eq.~\eqref{eq:kw:const} holds is at most $\sum_{i\in I}\frac{O(1)}{\sqrt{k_i}}$. Since ${\bf z}$ is $(1/\poly(m))$-pseudorandomly concentrated, the probability under $z\sim{\bf z}$ that Eq.~\eqref{eq:kw:const} holds is at most $\sum_{i\in I}\frac{O(1)}{\sqrt{k_i}}+\frac1{\poly(m)}$. Therefore, the expression in Eq.~\eqref{eq:kw:highlevel} is upper-bounded by
\mm{
&\E_{y\sim{\bf y}}\left[ \mathbf{1}_{|I|\ge2}\cdot \sum_{i\in I}\frac{O(1)}{\sqrt{k_i}} \right] + \frac1{\poly(m)} \\
&= \E_{y\sim{\bf y}}\left[ \sum_{i\in[m]}\frac{O(1)}{\sqrt{k_i}} \cdot\mathbf{1}_{i\in I\land |I|\ge2}\right] + \frac1{\poly(m)} \\
&= \sum_{i\in[m]}\frac{O(1)}{\sqrt{k_i}}\cdot\Pr_{y\sim{\bf y}}\left[i\in I\land |I|\ge2\right] + \frac1{\poly(m)} \mathdot \eqtag{eq:kw:nearfinal}
}

For any fixed $i\in[m]$, we upper-bound the probability of the event $i\in I\land|I|\ge2$ in two ways: The first upper-bound is $\Pr[i\in I]\le p$, and the second upper-bound is $\Pr[\exists j\in[m]\setminus\{i\},j\in I\land i\in I]<m\cdot p^2$ (since ${\bf y}$ is $p$-bounded in pairs). Hence,
\mm{
\Pr_{y\sim{\bf y}}[i\in I\land |I|\ge2] \le \min\left\{ p, m\cdot p^2 \right\} 
\le \sqrt{ m\cdot p^{3} }  \mathcomma
}
which implies that the expression in Eq.~\eqref{eq:kw:nearfinal} is upper-bounded by
\mm{
\sqrt{ m\cdot p^{3}}\cdot\sum_{i\in[m]}\frac{O(1)}{\sqrt{k_i}}+\frac1{\poly(m)} 
= O\left( \sqrt{m}\cdot p^{3/2}\cdot\sum_{i\in[m]}\frac{1}{\sqrt{i}} \right) 
= O\left( m\cdot p^{3/2} \right) \mathdot \tag*{\qedhere}
}
\end{proof}

Our pseudorandom generator, which is contructed next, is based on an application of Proposition~\ref{prop:kw:rest} as well as on the pseudorandom generator of Gopalan, Kane, and Meka (i.e., Theorem~\ref{thm:prg:ltfs}).

\begin{theorem} (quantified derandomization of depth-2 linear threshold circuits with $n^{3/2-\Omega(1)}$ wires).
There exists a polynomial-time algorithm $G$ that is given as input a random seed $s$ of length $\tilde{O}(\log(n))$ and a constant $\e>0$, and outputs a string $G(s,\e)\in\bitset^n$ such that for every $C\in\mathcal{C}_{n,2,n^{3/2-\e}}$ the following holds:
\begin{enumerate}
	\item If $C$ accepts all but at most $B(n)=2^{n^{\e/2}}$ inputs, then $\Pr_s[C(G(s,\e))=1]=1-o(1)$.
	\item If $C$ rejects all but at most $B(n)$ inputs, then $\Pr_s[C(G(s,\e))=1]=o(1)$.
\end{enumerate}
\end{theorem}

\begin{proof}[{\bf Proof.}]
Let $\de\in(\e/2,2\e/3)$ such that $p=n^{-(1-\de)}$ is a power of two. The algorithm first samples a restriction that meets the requirements of Proposition~\ref{prop:kw:rest}, as follows: The distribution ${\bf y}$ over $\bitset^n$ is obtained by sampling a string $y'$ from a distribution over $\bitset^{\log(1/p)\cdot n}$ that is $\frac1{\poly(n)}$-almost $O(\log(n))$-wise independent, and setting ${\bf y}_i=1$ if and only if the $i^{th}$ block in $y'$ is all zeroes; and the distribution ${\bf z}$ is $\frac1{\poly(n)}$-pseudorandomly concentrated. The required seed length to sample such a restriction is dominated by the seed length required to sample $z\sim{\bf z}$, which (using Theorem~\ref{thm:prg:ltfs}) is $O(\log(n)\cdot(\log\log(n))^2)$.

We say that a restriction $\rho$ is {\sf successful} if the circuit $C\rest_\rho$ can be computed by a single LTF, and if at least $\frac1{2}\cdot (p\cdot n)=\frac1{2}\cdot n^{\de}$ variables remain alive under $\rho$. We first claim that the probability that $\rho$ is successful is $1-o(1)$. According to Fact~\ref{fact:tail}, with probability $1-1/\poly(n)$ at least $\frac1{2}\cdot n^{\de}$ variables remain alive under $\rho$. To see that with high probability $C\rest_\rho$ can be computed by a single LTF, let $\mathcal{G}$ be the set of gates in the bottom layer of $C$. We say that a gate $\Phi$ is {\sf non-trivial} if $\Phi$ depends on more than a single input bit; note that any trivial gate can be replaced by a constant or by an input bit (or its negation). Then, the expected number of non-trivial gates in the bottom layer of $C\rest_\rho$ is
\mm{
\E_\rho\left[ \sum_{\Phi\in\mathcal{G}}\mathbf{1}_{\Phi\rest_\rho\text{ is non-trivial}} \right] &= \sum_{\Phi\in\mathcal{G}}\Pr_\rho[\Phi\rest_\rho\text{ is non-trivial}] \\
&= O\left( \sum_{\Phi\in\mathcal{G}}\text{fan-in}(\Phi)\cdot p^{3/2} \right)\\
&= O\left( n^{3/2-\e}\cdot n^{3\de/2-3/2} \right) \mathcomma
}
which is $o(1)$, since $\de<2\e/3$. Therefore, the probability that there are no non-trivial gates in the bottom layer of $C\rest_\rho$ is $1-o(1)$. 

After sampling the restriction $\rho$, the algorithm samples a string $x\in\bitset^{|\rho^{-1}(\star)|}$ using the pseudorandom generator $G'$ for LTFs from Theorem~\ref{thm:prg:ltfs}, instantiated with error parameter $1/\poly(n)$, and outputs the $n$-bit string that is obtained by completing $x$ to an $n$-bit string according to $\rho$. 

To see that the algorithm is correct, assume that $C$ accepts all but $2^{n^{\e/2}}$ of its inputs. Then, for every successful restriction $\rho$, the acceptance probability of $C\rest_\rho$ is $1-o(1)$ (since $\rho$ keeps at least $\frac1{2}\cdot n^{\de}=\omega(n^{\e/2})$ variables alive). Thus,
\mm{
\Pr_s[C(G(s,\e))=0] &\le \Pr_\rho[\rho\text{ not successful}] + \Pr_{s}[C(G(s,\e))=0|\rho\text{ successful}] \\
&\le o(1) + \max_{\rho\text{ successful}}\Pr_{s'}[C\rest_\rho(G'(s'))=0] \mathcomma
}
which is $o(1)$ since $G'$ is $\frac1{\poly(n)}$-pseudorandom for LTFs. Similarly, if $C$ rejects all but $2^{n^{\e/2}}$ of its inputs, then $\Pr[C(G(s))=1]=o(1)$.
\end{proof}

\section{Restrictions for sparse $\tc^0$ circuits: A potential path towards $\nexp\not\subseteq\tc^0$} \label{sec:open}

Recall that the best currently-known lower bounds for $\tc^0$ circuits of arbitrary constant depth $d$ are for circuits with $n^{1+\exp(-d)}$ wires. We now present an open problem that involves restrictions for $\tc^0$ circuits with only $n^{1+O(1/d)}$ wires, and show that a resolution of this open problem would imply that $\nexp\not\subseteq\tc^0$.

Towards presenting the problem, fix some class $\mathcal{C}_{simple}$ of ``simple'' functions such that the following holds: There exists a deterministic algorithm that gets as input $C'\in\mathcal{C}_{simple}$, runs in sufficiently small sub-exponential time, and distinguishes between the case that the acceptance probability of $C'$ is at least $2/3$ and the case that the acceptance probability of $C'$ is at most $1/3$. Then, the problem is the following:

\begin{problem} (deterministic restriction algorithm for sparse $\tc^0$ circuits). \label{prob:rest}
Construct a deterministic algorithm that gets as input a $\tc^0$ circuit $C:\pmset^n\ra\pmset$ of depth $d$ with $n^{1+O(1/d)}$ wires, runs in time at most $2^{n^{1/4^d}}$, and finds a set $S\subseteq\pmset^n$ and $C'\in\mathcal{C}_{simple}$ such that $|S|\ge10\cdot2^{n^{1-1/5d}}$ and $C\rest_{S}$ is $(1/10)$-close to $C'$.
\end{problem}

A resolution of Open Problem~\ref{prob:rest} would imply that there exists an algorithm for quantified derandomization of $\tc^0$ circuits of depth $d$ with $n^{1+O(1/d)}$ wires and $B(n)=2^{n^{1-1/5d}}$ exceptional inputs that runs in sufficiently small sub-exponential time (i.e., in time $2^{n^{1/4^d}}$). This is the case because a quantified derandomization algorithm can act similarly to our algorithm from the proof of Theorem~\ref{thm:int:main}, as follows: First find a set $S$ such that $|S|\ge10\cdot2^{n^{1-1/5d}}$ and $C\rest_S$ is $(1/10)$-close to some $C'\in\mathcal{C}_{simple}$; then, note that $C\rest_S$ has either very high acceptance probability or very low acceptance probability (because $C$ has at most $B(n)\le|S|/10$ exceptional inputs); and finally, estimate the acceptance probability of $C\rest_S$ (by estimating the acceptance probability of $C'$) in order to decide whether $C$ accepts all but at most $B(n)$ of its inputs or rejects all but at most $B(n)$ of its inputs. Thus, relying on Corollary~\ref{cor:int:threshold}, a resolution of Open Problem~\ref{prob:rest} would imply that $\nexp\not\subseteq\tc^0$.

\phantomsection
\section*{Acknowledgements} \label{sec:ack}
\addcontentsline{toc}{section}{Acknowledgements}

This work was initiated and partially conducted while the author was visiting Rocco Servedio at Columbia, and under Rocco's guidance. The author is very grateful to Rocco, who declined co-authorship of the paper, for his guidance, for many useful ideas, and for numerous inspiring conversations. The author thanks his advisor, Oded Goldreich, for the very useful idea to use tensor codes in the proof of Theorem~\ref{thm:int:threshold}, and for his guidance and support during the research and writing process. The author also thanks Amnon Ta-Shma for a useful conversation about constructing extractors in $\tc^0$.

This research was partially supported by the Minerva Foundation with funds from the Federal German Ministry for Education and Research. The research was also supported by the Prof. Rahamimoff Travel Grant for Young Scientists of the US-Israel Binational Science Foundation (BSF).

\bibliographystyle{alpha}
\bibliography{refs}

\newcommand{\etalchar}[1]{$^{#1}$}
 \newcommand{\proc}{Proc. } \newcommand{\stoc}[1]{\proc #1 Annual ACM Symposium
  on Theory of Computing ({STOC})} \newcommand{\focs}[1]{\proc #1 Annual IEEE
  Symposium on Foundations of Computer Science ({FOCS})}
  \newcommand{\ccc}[1]{\proc #1 Annual IEEE Conference on Computational
  Complexity ({CCC})} \newcommand{\icalp}[1]{\proc #1 International Colloquium
  on Automata, Languages and Programming ({ICALP})} \newcommand{\soda}[1]{\proc
  #1 Annual ACM-SIAM Symposium on Discrete Algorithms ({SODA})}
  \newcommand{\apprx}[1]{\proc #1 International Workshop on Approximation
  Algorithms for Combinatorial Optimization Problems ({APPROX})}
  \newcommand{\rnd}[1]{\proc #1 International Workshop on Randomization and
  Approximation Techniques in Computer Science ({RANDOM})}
  \newcommand{\fsttcs}[1]{\proc #1 Annual Conference on Foundations of Software
  Technology and Theoretical Computer Science ({FSTTCS})}
  \newcommand{\mfcs}[1]{\proc #1 International Symposium on Mathematical
  Foundations of Computer Science} \newcommand{\eccc}{Electronic Colloquium on
  Computational Complexity: {ECCC}} \newcommand{\jacm}{Journal of the ACM}
  \newcommand{\cc}{Computational Complexity} \newcommand{\jcss}{Journal of
  Computer and System Sciences} \newcommand{\siamj}{{SIAM} Journal of
  Computing} \newcommand{\ipl}{Information Processing Letters}
  \newcommand{\tocj}{Theory of Computing} \newcommand{\tcs}{Theoretical
  Computer Science} \newcommand{\toit}{IEEE Transactions on Information Theory}
\begin{thebibliography}{BYRST02}

\bibitem[Aar17]{aar17}
Scott Aaronson.
\newblock ${P}\mathop{{ =}}\limits^{{?}}{NP}$, 2017.
\newblock Accessed at \url{http://www.scottaaronson.com/papers/pnp.pdf}, June
  20, 2017.

\bibitem[AB09]{ab09}
Sanjeev Arora and Boaz Barak.
\newblock {\em Computational complexity: A modern approach}.
\newblock Cambridge University Press, Cambridge, 2009.

\bibitem[ABN{\etalchar{+}}92]{abnnr92}
N.~Alon, J.~Bruck, J.~Naor, M.~Naor, and R.~M. Roth.
\newblock Construction of asymptotically good low-rate error-correcting codes
  through pseudo-random graphs.
\newblock {\em \toit}, 38(2):509--516, 1992.

\bibitem[AS15]{as15}
Kazuyuki Amano and Atsushi Saito.
\newblock A nonuniform circuit class with multilayer of threshold gates having
  super quasi polynomial size lower bounds against {NEXP}.
\newblock In {\em Proc. 9th International Conference on Language and Automata
  Theory and Applications ({LATA})}, pages 461--472. 2015.

\bibitem[BBL92]{bbl92}
Paul Beame, Erik Brisson, and Richard Ladner.
\newblock The complexity of computing symmetric functions using threshold
  circuits.
\newblock {\em \tcs}, 100(1):253--265, 1992.

\bibitem[BIS12]{bis12}
Paul Beame, Russell Impagliazzo, and Srikanth Srinivasan.
\newblock Approximating {${\rm AC}^0$} by small height decision trees and a
  deterministic algorithm for {$\#{\rm AC}^0SAT$}.
\newblock In {\em \ccc{27th}}, pages 117--125. 2012.

\bibitem[BM84]{bm84}
Manuel Blum and Silvio Micali.
\newblock How to generate cryptographically strong sequences of pseudo-random
  bits.
\newblock {\em \siamj}, 13(4):850--864, 1984.

\bibitem[Bra10]{bra10}
Mark Braverman.
\newblock Polylogarithmic independence fools \emph{AC}\({}^{\mbox{0}}\)
  circuits.
\newblock {\em \jacm}, 57(5), 2010.

\bibitem[BV14]{bv14}
Eli Ben-Sasson and Emanuele Viola.
\newblock Short {PCP}s with projection queries.
\newblock In {\em \icalp{41st}}, pages 163--173. 2014.

\bibitem[BYRST02]{brst02}
Z.~Bar-Yossef, O.~Reingold, R.~Shaltiel, and L.~Trevisan.
\newblock Streaming computation of combinatorial objects.
\newblock In {\em \ccc{17th}}, pages 133--142, 2002.

\bibitem[CKK{\etalchar{+}}15]{ckksz15}
Ruiwen Chen, Valentine Kabanets, Antonina Kolokolova, Ronen Shaltiel, and David
  Zuckerman.
\newblock Mining circuit lower bound proofs for meta-algorithms.
\newblock {\em \cc}, 24(2):333--392, 2015.

\bibitem[CL16]{cl16}
Kuan Cheng and Xin Li.
\newblock Randomness extraction in {AC0} and with small locality.
\newblock {\em \eccc}, 23:18, 2016.

\bibitem[CSS16]{css16}
Ruiwen Chen, Rahul Santhanam, and Srikanth Srinivasan.
\newblock Average-case lower bounds and satisfiability algorithms for small
  threshold circuits.
\newblock In {\em \ccc{31st}}, pages 1:1--1:35, 2016.

\bibitem[DGJ{\etalchar{+}}10]{dgjsv10}
Ilias Diakonikolas, Parikshit Gopalan, Ragesh Jaiswal, Rocco~A. Servedio, and
  Emanuele Viola.
\newblock Bounded independence fools halfspaces.
\newblock {\em \siamj}, 39(8):3441--3462, 2010.

\bibitem[DP09]{dp09}
Devdatt Dubhashi and Alessandro Panconesi.
\newblock {\em Concentration of Measure for the Analysis of Randomized
  Algorithms}.
\newblock Cambridge University Press, 2009.

\bibitem[GHR92]{ghr92}
Mikael Goldmann, Johan H{\aa}stad, and Alexander Razborov.
\newblock Majority gates vs.\ general weighted threshold gates.
\newblock In {\em Proc. 7th {A}nnual {S}tructure in {C}omplexity {T}heory
  {C}onference}, pages 2--13, 1992.

\bibitem[GK98]{gk98}
Mikael Goldmann and Marek Karpinski.
\newblock Simulating threshold circuits by majority circuits.
\newblock {\em \siamj}, 27(1):230--246, 1998.

\bibitem[GKM15]{gkm15}
Parikshit Gopalan, Daniel Kane, and Raghu Meka.
\newblock Pseudorandomness via the discrete {F}ourier transform.
\newblock In {\em \focs{56th}}, pages 903--922. 2015.

\bibitem[GMR13]{gmr13}
Parikshit Gopalan, Raghu Meka, and Omer Reingold.
\newblock Dnf sparsification and a faster deterministic counting algorithm.
\newblock {\em \cc}, 22(2):275--310, 2013.

\bibitem[Gol08]{gol08}
Oded Goldreich.
\newblock {\em Computational Complexity: A Conceptual Perspective}.
\newblock Cambridge University Press, New York, NY, USA, 2008.

\bibitem[GOWZ10]{gowz10}
Parikshit Gopalan, Ryan O'Donnell, Yi~Wu, and David Zuckerman.
\newblock Fooling functions of halfspaces under product distributions.
\newblock In {\em \ccc{25th}}, pages 223--234. 2010.

\bibitem[GT91]{gt91}
Hans~Dietmar Gr{\"o}ger and Gy{\"o}rgy" Tur{\'a}n.
\newblock On linear decision trees computing boolean functions.
\newblock In {\em \icalp{18th}}, 1991.

\bibitem[GVW15]{gvw15}
Oded Goldreich, Emanuele Viola, and Avi Wigderson.
\newblock On randomness extraction in {AC0}.
\newblock In {\em \ccc{30th}}, pages 601--668, 2015.

\bibitem[GW14]{gw14}
Oded Goldreich and Avi Widgerson.
\newblock On derandomizing algorithms that err extremely rarely.
\newblock In {\em \stoc{46th}}, pages 109--118. 2014.
\newblock Full version available online at \emph{\eccc}, 20:152 (Rev. 2), 2013.

\bibitem[H{\aa}s94]{has94}
Johan H{\aa}stad.
\newblock On the size of weights for threshold gates.
\newblock {\em SIAM Journal on Discrete Mathematics}, 7(3):484--492, 1994.

\bibitem[Hea08]{hea08}
Alexander~D. Healy.
\newblock Randomness-efficient sampling within {${\rm NC}^1$}.
\newblock {\em \cc}, 17(1):3--37, 2008.

\bibitem[HKM12]{hkm12}
Prahladh Harsha, Adam Klivans, and Raghu Meka.
\newblock An invariance principle for polytopes.
\newblock {\em \jacm}, 59(6):29:1--29:25, 2012.

\bibitem[IKW02]{ikw02}
Russell Impagliazzo, Valentine Kabanets, and Avi Wigderson.
\newblock In search of an easy witness: exponential time vs.\ probabilistic
  polynomial time.
\newblock {\em \jcss}, 65(4):672--694, 2002.

\bibitem[IMP12]{imp12}
Russell Impagliazzo, William Matthews, and Ramamohan Paturi.
\newblock A satisfiability algorithm for \emph{AC}\({}^{\mbox{0}}\).
\newblock In {\em \soda{23rd}}, pages 961--972, 2012.

\bibitem[IMZ12]{imz12}
Russell Impagliazzo, Raghu Meka, and David Zuckerman.
\newblock Pseudorandomness from shrinkage.
\newblock In {\em \focs{53rd}}, pages 111--119. 2012.

\bibitem[IPS97]{ips97}
Russell Impagliazzo, Ramamohan Paturi, and Michael~E. Saks.
\newblock Size-depth tradeoffs for threshold circuits.
\newblock {\em \siamj}, 26(3):693--707, 1997.

\bibitem[IPS13]{ips13}
Russell Impagliazzo, Ramamohan Paturi, and Stefan Schneider.
\newblock A satisfiability algorithm for sparse depth two threshold circuits.
\newblock In {\em \focs{54th}}, pages 479--488. 2013.

\bibitem[IW98]{iw98}
R.~Impagliazzo and A.~Wigderson.
\newblock Randomness vs. time: De-randomization under a uniform assumption.
\newblock In {\em \focs{39th}}, pages 734--, 1998.

\bibitem[Kan11]{kan11}
Daniel~M. Kane.
\newblock A small {PRG} for polynomial threshold functions of {G}aussians.
\newblock In {\em \focs{52nd}}, pages 257--266. 2011.

\bibitem[Kan14]{kan14}
Daniel~M. Kane.
\newblock A pseudorandom generator for polynomial threshold functions of
  {G}aussian with subpolynomial seed length.
\newblock In {\em \ccc{29th}}, pages 217--228. 2014.

\bibitem[KM15]{km15}
Pravesh~K. Kothari and Raghu Meka.
\newblock Almost optimal pseudorandom generators for spherical caps.
\newblock In {\em \stoc{47th}}, pages 247--256. 2015.

\bibitem[KRS12]{krs12}
Zohar~S. Karnin, Yuval Rabani, and Amir Shpilka.
\newblock Explicit dimension reduction and its applications.
\newblock {\em \siamj}, 41(1):219--249, 2012.

\bibitem[KW16]{kw16}
Daniel~M. Kane and Ryan Williams.
\newblock Super-linear gate and super-quadratic wire lower bounds for depth-two
  and depth-three threshold circuits.
\newblock In {\em \stoc{48th}}, pages 633--643, 2016.

\bibitem[LMN93]{lmn93}
Nathan Linial, Yishay Mansour, and Noam Nisan.
\newblock Constant depth circuits, {F}ourier transform, and learnability.
\newblock {\em Journal of the Association for Computing Machinery},
  40(3):607--620, 1993.

\bibitem[MZ13]{mz13}
Raghu Meka and David Zuckerman.
\newblock Pseudorandom generators for polynomial threshold functions.
\newblock {\em \siamj}, 42(3):1275--1301, 2013.

\bibitem[Nis93]{nis93}
Noam Nisan.
\newblock The communication complexity of threshold gates.
\newblock In {\em Combinatorics, {P}aul {E}rd\H os is eighty, {V}ol.\ 1},
  Bolyai Society Mathematical Studies, pages 301--315. 1993.

\bibitem[NN93]{nn93}
Joseph Naor and Moni Naor.
\newblock Small-bias probability spaces: efficient constructions and
  applications.
\newblock {\em \siamj}, 22(4):838--856, 1993.

\bibitem[NW94]{nw94}
Noam Nisan and Avi Wigderson.
\newblock Hardness vs.\ randomness.
\newblock {\em \jcss}, 49(2):149--167, 1994.

\bibitem[O'D14]{odo14}
Ryan O'Donnell.
\newblock {\em Analysis of Boolean Functions}.
\newblock Cambridge University Press, 2014.

\bibitem[PS94]{ps94}
Ramamohan Paturi and Michael~E. Saks.
\newblock Approximating threshold circuits by rational functions.
\newblock {\em Information and Computation}, 112(2):257--272, 1994.

\bibitem[ROS94]{ros94}
V.~P. Roychowdhury, A.~Orlitsky, and Kai-Yeung Siu.
\newblock Lower bounds on threshold and related circuits via communication
  complexity.
\newblock {\em \toit}, 40(2):467--474, 1994.

\bibitem[RRV02]{rrv02}
Ran Raz, Omer Reingold, and Salil Vadhan.
\newblock Extracting all the randomness and reducing the error in {T}revisan's
  extractors.
\newblock {\em \jcss}, 65(1):97--128, 2002.

\bibitem[RS10]{rs10}
Yuval Rabani and Amir Shpilka.
\newblock Explicit construction of a small epsilon-net for linear threshold
  functions.
\newblock {\em \siamj}, 39(8):3501--3520, 2010.

\bibitem[San10]{san10}
Rahul Santhanam.
\newblock Fighting perebor: new and improved algorithms for formula and {QBF}
  satisfiability.
\newblock In {\em \focs{51st}}, pages 183--192. 2010.

\bibitem[Ser07]{ser07}
Rocco~A. Servedio.
\newblock Every linear threshold function has a low-weight approximator.
\newblock {\em \cc}, 16(2):180--209, 2007.

\bibitem[Smo90]{smo90}
Roman Smolensky.
\newblock On interpolation by analytic functions with special properties and
  some weak lower bounds on the size of circuits with symmetric gates.
\newblock In {\em \focs{31st}}, pages 628--631, 1990.

\bibitem[SSTT16]{sstt16}
Takayuki Sakai, Kazuhisa Seto, Suguru Tamaki, and Junichi Teruyama.
\newblock Bounded depth circuits with weighted symmetric gates: satisfiability,
  lower bounds and compression.
\newblock In {\em \mfcs{41st}}. 2016.

\bibitem[ST12]{st12}
K.~Seto and S.~Tamaki.
\newblock A satisfiability algorithm and average-case hardness for formulas
  over the full binary basis.
\newblock In {\em \ccc{27th}}, pages 107--116, 2012.

\bibitem[ST17a]{st17b}
Rocco Servedio and Li-Yang Tan.
\newblock Deterministic search for {CNF} satisfying assignments in almost
  polynomial time.
\newblock In {\em \focs{58th}}, 2017.

\bibitem[ST17b]{st17}
Rocco Servedio and Li-Yang Tan.
\newblock Learning and fooling depth-two threshold circuits.
\newblock Unpublished manuscript, 2017.

\bibitem[SW13]{sw13}
Rahul Santhanam and Ryan Williams.
\newblock On medium-uniformity and circuit lower bounds.
\newblock In {\em \ccc{28th}}, pages 15--23. 2013.

\bibitem[Tam16]{tam16}
Suguru Tamaki.
\newblock A satisfiability algorithm for depth two circuits with a
  sub-quadratic number of symmetric and threshold gates.
\newblock {\em \eccc}, 23:100, 2016.

\bibitem[Tel17]{tell17}
Roei Tell.
\newblock Improved bounds for quantified derandomization of constant-depth
  circuits and polynomials.
\newblock In {\em \ccc{32nd}}, pages 18:1 -- 18:49, 2017.

\bibitem[Tre01]{tre01}
Luca Trevisan.
\newblock Extractors and pseudorandom generators.
\newblock {\em \jacm}, 48(4):860--879, 2001.

\bibitem[TS17]{tas17}
Amnon Ta-Shma.
\newblock Explicit, almost optimal, epsilon-balanced codes.
\newblock In {\em \stoc{49th}}, 2017.

\bibitem[TX13]{tx13}
Luca Trevisan and TongKe Xue.
\newblock A derandomized switching lemma and an improved derandomization of
  {AC}0.
\newblock In {\em \ccc{28th}}, pages 242--247. 2013.

\bibitem[Vad12]{vad12}
Salil~P. Vadhan.
\newblock {\em Pseudorandomness}.
\newblock Foundations and Trends in Theoretical Computer Science. Now
  Publishers, 2012.

\bibitem[Vio05]{vio05}
Emanuele Viola.
\newblock The complexity of constructing pseudorandom generators from hard
  functions.
\newblock {\em \cc}, 13(3-4):147--188, 2005.

\bibitem[Wil11]{wil11}
Ryan Williams.
\newblock Non-uniform {ACC} circuit lower bounds.
\newblock In {\em \ccc{26th}}, pages 115--125. 2011.

\bibitem[Wil13]{wil13}
Ryan Williams.
\newblock Improving exhaustive search implies superpolynomial lower bounds.
\newblock {\em \siamj}, 42(3):1218--1244, 2013.

\bibitem[Wil14a]{wil14}
Ryan Williams.
\newblock Algorithms for circuits and circuits for algorithms: Connecting the
  tractable and intractable.
\newblock In {\em Proc. International Congress of Mathematicians (ICM)}, pages
  659--682, 2014.

\bibitem[Wil14b]{wil14b}
Ryan Williams.
\newblock New algorithms and lower bounds for circuits with linear threshold
  gates.
\newblock In {\em \stoc{55th}}, pages 194--202, 2014.

\bibitem[Yao82]{yao82}
Andrew~C. Yao.
\newblock Theory and application of trapdoor functions.
\newblock In {\em \focs{23rd}}, pages 80--91, 1982.

\end{thebibliography}

\begin{appendices}

\section{Quantified derandomization and lower bounds} \label{app:quant:lb}

In this appendix we prove that ``black-box'' \emph{quantified} derandomization of a class $\mathcal{C}$ yields lower bounds for $\mathcal{C}$, in the same way that standard derandomization does. For simplicity, we focus on the case of derandomization with one-sided error. Let us first recall the notion of a hitting-set generator, which yields a ``black-box'' quantified derandomization with one-sided error of a circuit class.

\begin{definition} (hitting-set generator). \label{def:hsg}
Let $\mathcal{F}=\bigcup_{n\in\N}\mathcal{F}_n$, where for every $n\in\N$ it holds that $\mathcal{F}_n$ is a set of functions $\bitset^n\ra\bitset$, and let $\ell:\N\ra\N$. An algorithm $H$ is a {\sf hitting-set generator for $\mathcal{F}$ with seed length $\ell$} if for every $n\in\N$ and every $f\in\mathcal{F}_n$ there exists $s\in\bitset^{\ell(n)}$ such that $f(H(s))=1$.
\end{definition}

In the following proposition, we assume that there exists a hitting-set generator with non-trivial seed length $\ell(n)<n$ for circuits with $B(n)\ge2^{\ell}$ exceptional inputs, and show that this implies lower bounds for the corresponding circuit class.

\begin{proposition} (quantified derandomization implies lower bounds). \label{prop:quant:lb}
Let $\ell:\N\ra\N$ such that $\ell(n)<n$, and let $B:\N\ra\N$ such that $B(n)\ge2^{\ell(n)}$. Let $\mathcal{C}$ be a circuit class, and let $\mathcal{C}^{\le B}\subseteq\mathcal{C}$ be the subclass of circuits that reject at most $B(n)$ of their inputs. Assume that there exists a $2^{O(\ell)}$-time computable hitting-set generator $H$ with seed length $\ell$ for $\mathcal{C}^{\le B}$. Then, there exists a function in $DTIME(2^{O(\ell(n))})$ that cannot be computed by \emph{any} circuit in $\mathcal{C}$.
\end{proposition}

\begin{proof}[{\bf Proof.}]
The ``hard'' function for $\mathcal{C}$, denoted $f$, is the indicator function of $\bitset^n\setminus \{H(s):s\in\bitset^{\ell(n)}\}$; that is, $f(x)=0$ if and only if there exists $s\in\bitset^{\ell(n)}$ such that $x=H(s)$. Note that any $C\in\mathcal{C}$ that computes $f$ rejects at most $2^{\ell}\le B(n)$ inputs, and thus $C\in\mathcal{C}^{\le B}$. However, this means that $H$ is a hitting-set generator for $C$, and so there exists $s\in\bitset^{\ell(n)}$ such that $C(H(s))=1$. Since $f(H(s))=0$, we obtain a contradiction to the hypothesis that $C$ computes $f$.
\end{proof}

\section{Proof of a technical claim from Section~\ref{sec:threshold}} \label{app:tensor}

In the proof of Proposition~\ref{prop:threshold:code:const}, we omitted the proof of the following claim: For every $x\in\bitset^n$ such that $x\ne0^n$, the relative Hamming weight $\hat{x}=C(x)$ is at least $(1/3)^d$. The proof of this claim, which we now detail, follows from a standard property of tensor codes: If a code $\code$ has distance $\de>0$, then the tensor code of order $d$ that is based on $\code$ has distance $\de^d$. 

\begin{claimdot} \label{claim:tensor}
Let $C$ be the circuit constructed in the proof of Proposition~\ref{prop:threshold:code:const}, and let  $x\in\bitset^n$ such that $x\ne0^n$. Then, the relative Hamming weight $\hat{x}=C(x)$ is at least $(1/3)^d$.
\end{claimdot}

\begin{proof}[{\bf Proof.}]
Recall that the code $\code$ maps any non-zero message of length $m$ to a codeword of length $\bar{m}$ with at least $r\eqdef\bar{m}/3$ non-zero entries. Our hypothesis is that $x=M^{(0)}$ is not the all-zero message, and we will now prove that for each $i\in[d]$ it holds that $M^{(i)}$ has at least $r^{i}$ non-zero entries. The proof is by induction, and will rely on a stronger induction hypothesis: We prove that for each $i\in\{0,...,d\}$ there exists $\vec{x}_{\ge i+1}\in[m]^{d-i}$ such that the number of vectors $\vec{x}_{\le i}\in[\bar{m}]^{i}$ for which $M^{(i)}_{\vec{x}_{\le i},\vec{x}_{\ge i+1}}\ne0$ is at least $r^i$. 

For the base case $i=1$, note that by our hypothesis there exists $\vec{x}\in[m]^d$ such that $M^{(0)}_{\vec{x}}\ne0$. Therefore, the $m$-bit vector $M^{(0)}_{\star,\vec{x}_{\ge2}}=M^{(0)}_{1,\vec{x}_2,...,\vec{x}_d},...,M^{(0)}_{m,\vec{x}_2,...,\vec{x}_d}$ is non-zero. By the properties of $\code$ it holds that $\code\left( M^{(0)}_{\star,\vec{x}_{\ge2}} \right)$ has at least $r$ non-zero entries. The bits of $\code\left( M^{(0)}_{\star,\vec{x}_{\ge2}} \right)$ appear in $M^{(i)}$ in locations $(1,\vec{x}_{2},...,\vec{x}_{d}),...,(\bar{m},\vec{x}_2,...,\vec{x}_d)$. Therefore, the claim is proved for $i=1$ with the vector $\vec{x}_{\ge2}=\vec{x}_2,...,\vec{x}_d\in[m]^{d-1}$.

For the induction step, let $i\ge2$. By the induction hypothesis, for some $\vec{x}_{\ge i}\in[m]^{d-(i-1)}$ there exist at least $r^{i-1}$ vectors $\vec{x}_{\le i-1}^{(1)},...,\vec{x}_{\le i-1}^{(r^{i-1})}\in[\bar{m}]^{i-1}$ such that $M^{(i-1)}_{\vec{x}_{\le i-1}^{(j)},\vec{x}_{\ge i}}\ne0$ for all $j\in[r^{i-1}]$. Fix $j\in[r^{i-1}]$. Since $M^{(i-1)}_{\vec{x}^{(j)}_{\le i-1},\vec{x}_{\ge i}}\ne0$, it follows that the string $M^{(i-1)}_{\vec{x}^{(j)}_{\le i-1},\star,\vec{x}^{(j)}_{\ge i+1}}=M^{(i-1)}_{\vec{x}^{(j)}_{\le i-1},1,\vec{x}^{(j)}_{\ge i+1}},...,M^{(i-1)}_{\vec{x}^{(j)}_{\le i-1},m,\vec{x}^{(j)}_{\ge i+1}}\in\bitset^m$ is non-zero. Thus, by the properties of $\code$, the string $\code\left( M^{(i-1)}_{\vec{x}^{(j)}_{\le i-1},\star,\vec{x}^{(j)}_{\ge i+1}} \right)$ contains at least $r$ non-zero entries.

Now, for every $j\in[r^{i-1}]$, let $\vec{X}^{(j)}\eqdef\left\{\left(\vec{x}^{(j)}_{\le i-1},1,\vec{x}^{(j)}_{\ge i+1}\right),...,\left(\vec{x}^{(j)}_{\le i-1},\bar{m},\vec{x}^{(j)}_{\ge i+1}\right)\right\}$ be the set of $\bar{m}$ locations in $M^{(i)}$ in which the string $\code\left( M^{(i-1)}_{\vec{x}^{(j)}_{\le i-1},\star,\vec{x}^{(j)}_{\ge i+1}} \right)$ appears. Note that for every $j\ne j'\in[r^{i-1}]$ it holds that all locations in $X^{(j)}$ and $X^{(j')}$ are distinct; that is, for every $k,k'\in[\bar{m}]$ it holds that $\left(\vec{x}^{(j)}_{\le i-1},k,\vec{x}^{(j)}_{\ge i+1}\right) \ne \left(\vec{x}^{(j')}_{\le i-1},k',\vec{x}^{(j)}_{\ge i+1}\right)$. Since for each $j\in[r^{i-1}]$ it holds that $X^{(j)}$ contains at least $r$ locations in which $M^{(i)}$ is non-zero, we deduce that $M^{(i)}$ has at least $r^i$ non-zero entries.
\end{proof}
\end{appendices}

\end{document}